%% file: 3CQBC.tex
\documentclass[onecolumn,11pt,draftclsnofoot]{IEEEtran}
\IEEEoverridecommandlockouts
\input{LatexPackages.tex}

\input{LatexMacroDefinitions.tex}

\def\apl{a_{\mbox{{\tiny$\oplus$}}}}
\def\hatapl{\hata_{\mbox{{\tiny$\oplus$}}}}

\def\dpl{d_{\mbox{{\tiny$\oplus$}}}}
\def\hatapl{\hata_{\mbox{{\tiny$\oplus$}}}}
\def\lambdapl{\lambda_{\mbox{{\tiny$\oplus$}}}}

\def\hata{\hat{a}}

\def\hatu{\hat{u}}

\def\hatw{\hat{w}}

\def\hatm{\hat{m}}

\def\hatv{\hat{v}}

\def\haty{\hat{y}}

\def\hatb{\hat{b}}

\global\long\def\11{\mathbbm{1}}

\def\ulineR{\underline{R}}

\def\ulineY{\underline{B}}

\def\ulinem{\underline{m}}

\def\ulineX{\underline{X}}
\def\ulinex{\underline{x}}
\def\ulineU{\underline{U}}
\def\ulineu{\underline{u}}
\def\ulineW{\underline{W}}

\def\ulineCalX{\underline{\mathcal{X}}}

\def\ulineCalU{\underline{\mathcal{U}}}
\def\ulineCalV{\underline{\mathcal{V}}}

\def\ulineCalM{\underline{\mathcal{M}}}

\def\ulineZ{\underline{Z}}
\def\ulinez{\underline{z}}
\def\ulineY{\underline{Y}}

\def\underlineSemiPrivateRV{\underline{\SemiPrivateRV}}
\def\underlinePrivateRV{\underline{\PrivateRV}}

\def\ulineX{\underline{X}}
\def\ulineY{\underline{Y}}
\def\ulineV{\underline{V}}
\def\ulinev{\underline{v}}
\def\ulinea{\underline{a}}
\def\ulinelambda{\underline{\lambda}}

\def\olinekappa{\overline{\kappa}}

\def\3To1BC{$3\mhyphen$to$-1$}

\def\define{:{=}~}

\def\naturals{\mathbb{N}}

\def\UM{\mathscr{U}\!\mathcal{M}-}

\def\fieldpi{\mathcal{F}_{\pi}}

\def\hatm{\hat{m}}
\def\cl{\mbox{cl}}
\def\cocl{\mbox{cocl}}

\def\SetOfDistributions{\mathbb{D}}

\def\SemiPrivateRVSet{\mathcal{V}}

\def\PrivateRV{V}
\def\SemiPrivateRV{U}
\def\underlinePrivateRV{\underline{\PrivateRV}}
\def\Expectation{\mathbb{E}}
\def\fieldq{\mathcal{F}_{q}}

\def\fieldpij{\mathcal{F}_{\pi_{j}}}
\newif\ifProofForORDBC

\def\parsec{\par\noindent}
\def\med{\medskip\parsec}

\def\CalA{\mathcal{A}}
\def\CalB{\mathcal{B}}
\def\CalC{\mathcal{C}}
\def\CalD{\mathcal{D}}
\def\CalE{\mathcal{E}}
\def\CalF{\mathcal{F}}
\def\CalG{\mathcal{G}}
\def\CalH{\mathcal{H}}

\def\CalL{\mathcal{L}}
\def\CalM{\mathcal{M}}

\def\CalP{\mathcal{P}}
\def\CalQ{\mathcal{Q}}

\def\CalT{\mathcal{T}}
\def\CalU{\mathcal{U}}
\def\CalV{\mathcal{V}}
\def\CalW{\mathcal{W}}
\def\CalX{\mathcal{X}}
\def\CalY{\mathcal{Y}}

\def\11{\mathbbm{1}}

\def\3To1BC{$3\mhyphen$to$-1$}

\def\define{:{=}~}

\def\naturals{\mathbb{N}}

\def\UM{\mathscr{U}\!\mathcal{M}-}

\def\fieldpi{\mathcal{F}_{\pi}}

\mathchardef\mhyphen="2D

\def\Expectation{\mathbb{E}}

\usepackage{stackengine}

\def\define{\mathrel{\ensurestackMath{\stackon[1pt]{=}{\scriptstyle\Delta}}}}

\newif\ifJournal

\usepackage{amssymb}
\usepackage{amsmath}
\usepackage{mathrsfs}
\usepackage{ulem}
\usepackage{epsf,epsfig}
\usepackage{cite}
\usepackage{color}
\usepackage{dsfont}
\usepackage{bbm}
\usepackage{amsthm}
\usepackage{physics}
\usepackage{cite}

\usepackage{float} 
\usepackage{accents}

\mathchardef\mhyphen="2D
\def\olinekappa{\overline{\kappa}}
\def\3To1BC{$3\mhyphen$to$-1$}

\def\ulinelambda{\underline{\lambda}}
\def\dbrackthree{\llbracket 3 \rrbracket}

\def\SemiPrivateRVSet{\mathcal{U}}
\def\underlineSemiPrivateRV{\underline{\SemiPrivateRV}}
\def\underlinePrivateRV{\underline{\PrivateRV}}

\def\PrivateRV{V}
\def\SemiPrivateRV{U}
\def\underlinePrivateRV{\underline{\PrivateRV}}
\def\Expectation{\mathbb{E}}
\def\fieldq{\mathcal{F}_{\upsilon}}

\def\Prime{\upsilon}

\def\fieldq{\mathcal{F}_{q}}

\def\fieldpij{\mathcal{F}_{\upsilon_{j}}}

\def\doubleunderline#1{\underline{\underline{#1}}}
\def\dulineU{\doubleunderline{U}}
\def\dulineCalU{\doubleunderline{\CalU}}

\def\dulinev{\doubleunderline{v}}
\def\dulineq{\doubleunderline{q}}

\def\dulineQ{\doubleunderline{Q}}
\def\dulineq{\doubleunderline{q}}
\def\dulineu{\doubleunderline{u}}
\def\dulineCalQ{\doubleunderline{\CalQ}}
\def\dulineT{\doubleunderline{T}}
\def\dulineCalU{\doubleunderline{\CalU}}

\begin{document}

\sloppy

\title{\huge Achievable Rate Regions for $3-$User Classical-Quantum Broadcast Channels}

\author{
\IEEEauthorblockN{Fatma Gouiaa and Arun Padakandla\\}
\IEEEauthorblockA{Communication Systems, EURECOM}
}
\maketitle

\begin{abstract}
We consider the scenario of communicating on a $3\mhyphen$user classical-quantum broadcast channel. We undertake an information theoretic study and focus on the problem of characterizing an inner bound to its capacity region. We design a new coding scheme based \textit{partitioned coset codes} - an ensemble of codes possessing algebraic properties. Analyzing its information-theoretic performance, we characterize a new inner bound. We identify examples for which the derived inner bound is strictly larger than that achievable using IID random codes.

Proceeding further, we incorporate Sen's technique of tilting smoothing and augmentation to perform simultaneous decoding via a simultaneous decoding POVM and thereby characterize a further enlarged achievable rate region for communicating classical bits over the $3-$user classical-quantum broadcast channel. Finally, in our last step, we characterize a new inner bound to the classical-quantum capacity region of the $3-$user classical-quantum broadcast channel that subsumes all previous known inner bounds by combining the conventional unstructured IID codes with structured coset code strategies.
\end{abstract}

\section{Introduction}
\label{Sec:Introduction}

We consider the scenario of communicating on a $3\mhyphen$user classical-quantum broadcast channel ($3\mhyphen$CQBC) depicted in Fig.~\ref{Fig:CommnOver3CQBC}. Our focus is on the problem of designing an efficient coding scheme and characterizing an inner bound (achievable rate region) to the capacity region of a general $3\mhyphen$CQBC. The current known coding schemes \cite{201110TIT_YarHayDev,201512TIT_SavWil,201605TIT_RadSenWar} for CQBCs are based on conventional IID random codes, also referred to herein as unstructured codes. In this work, we undertake a study of coding schemes based on \textit{coset codes} for CQBC communication and present the following contributions.

We propose a new coding scheme based on \textit{partitioned coset codes} (PCC) - an ensemble of codes possessing algebraic closure properties. We analyze its performance to derive a new inner bound to the capacity region of a general $3\mhyphen$CQBC. We identify examples of $3\mhyphen$CQBCs for which the derived inner bound is strictly larger than the current known largest. Our findings maybe viewed as another step \cite{202107ISIT_AnwPadPra3CQIC, 202212TIT_SohAtiPadPra} in our pursuit of designing  coding schemes based on coset codes for network CQ communication, and in particular, a first step in deriving a new achievable rate region for a general $3\mhyphen$CQBC.

An information theoretic study of CQBCs was initiated by Yard et. al. \cite{201110TIT_YarHayDev} in the context of $2\mhyphen$CQBCs wherein the superposition coding scheme was studied. Furthering this study, Savov and Wilde \cite{201512TIT_SavWil} proved achievability of Marton's binning \cite{197905TIT_Mar} for a general $2\mhyphen$CQBC. While these studied the asymptotic regime, Radhakrishnan et. al. \cite{201605TIT_RadSenWar} proved achievability of Marton's inner bound \cite{197905TIT_Mar} in the one-shot regime which also extended to clear proofs for the former regime considered in \cite{201110TIT_YarHayDev,201512TIT_SavWil}.

The works \cite{201110TIT_YarHayDev,201512TIT_SavWil,201605TIT_RadSenWar} were aimed at generalizing Marton's classical coding \cite{197905TIT_Mar} schemes that is based on IID codes. Fueled by the simplicity of IID codes and more importantly, a lack of evidence for their sub-optimality in one-to-many communication scenarios, most studies of the broadcast channel (BC) problem have restricted attention to IID coding schemes. Indeed, even within the larger class of BCs that include continuous valued sets, any number of receivers (Rxs) and multiple antennae, we were unaware of any BC for which IID coding schemes were sub-optimal. In 2013, \cite{201307ISIT_PadPraDBC} identified a $3\mhyphen$user classical BC ($3\mhyphen$CBC) for which a linear coding scheme strictly outperforms even a most general IID coding scheme that incorporates all known strategies \cite{197201TIT_Cov,197905TIT_Mar,198101TIT_HanKob}. This article builds on the ideas in \cite{201804TIT_PadPra} and is driven by a motivation to elevate the same to $3\mhyphen$CQBCs. In the sequel, we discuss why codes endowed with algebraic closure properties can enhance one-to-many communication.
\begin{figure}
    \centering
    \includegraphics[width=3.5in]{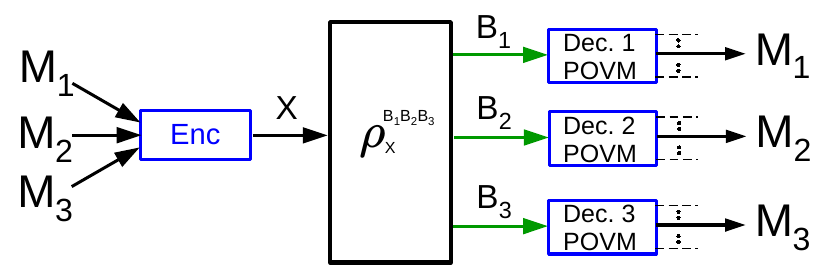}
    \caption{Communicating classical information over a $3\mhyphen$CQBC.}
    \label{Fig:CommnOver3CQBC}
\end{figure}

Communication on a BC entails fusing codewords chosen for the different Rxs through a single input. From the perspective of any Rx, a specific aggregation of the codewords chosen for the other Rxs acts as interference. See Fig.~\ref{Fig:InterferencOn3BC}. The Tx can precode for this interference via Marton's binning. In general, precoding entails a rate loss. In other words, suppose $W$ is the interference seen by Rx $1$, then the rate that Rx $1$ can achieve by decoding $W$ and peeling it off can be strictly larger than what it can achieve if the Tx precodes for $W$. This motivates every Rx to decode as large a fraction of the interference that it can and precode only for the minimal residual uncertainty.

In contrast to a BC with $2$ Rxs, the interference $W$ on a BC with $3$ Rxs is a \textit{bivariate} function of $V_{2},V_{3}$ - the signals of the other Rxs (Fig.~\ref{Fig:InterferencOn3BC}). A joint design of the $V_{2}\mhyphen,V_{3}\mhyphen$codes endowing them with structure can enable Rx $1$ decode $W$ efficiently even while being unable to decode $V_{2}$ and $V_{3}$. We elaborate on this phenomena through chosen examples (Ex.~\ref{Ex:GenAdd3CQBCExample}, \ref{Ex:RotatedExampleWithAdd}) followed by a self contained discussion that informs us of the structure (Sec.~\ref{SubSec:KeyElements}) of a general coding scheme. 

In this article, we present three inner bounds (Thms.~\ref{Thm:3CQBCCodingTheorem1}, \ref{SubSec:StepICodingTheorem} and \ref{Sec:Step2}). The former two illustrate several of the new elements - code structure, decoding rule, design and analysis of a simultaneous decoding POVM. Let us now elaborate on novelty involved in the designed coding strategies and analyzing their performances. The coding strategy in Sec.~\ref{Sec:AchRegionI} equips each user with just one codebook. Rx $1$ decodes a bivariate function of the interference while Rxs $2$ and $3$ only decode into their codebooks. A general coding scheme for a $3-$CQBC must permit precoding and message splitting to enable each Rx decode both univariate and bivariate components of signals intended for the other Rxs. As elaborated in the context of a $3\mhyphen$CBC \cite[Sec.~IV]{201804TIT_PadPra}, this involves multiple codebooks. Specifically, this will involve design and analysis of a simultaneous decoding POVM. Design and analysis of a simultaneous decoding POVM remained an important stumbling block in network quantum information theory since its recognition by Winter in \cite{200107TIT_Win}. Recently, discovering new ideas of \textit{tilting, smoothing and augmentation} (TSA) and new analysis techniques, Sen \cite{202103SAD_Sen,202102SAD_Sen} has presented an approach to design and analyze simultaneous decoding. Combining our coset code strategy with Sen's TSA \cite{202103SAD_Sen,202102SAD_Sen}, we design and analyze performance of a coset code strategy involving multiple codebooks. This enables all the Rxs decode parts of the interference and peel them off, thereby overcome the rate loss if they were precoded for. Thm.~\ref{SubSec:StepICodingTheorem} designs and analyzes a strategy involving each Rx decode only a bivariate interference component. Thm.~\ref{Thm:Step2} goes further and characterizes an inner bound that combines the conventional unstructured IID code strategy with the coset code strategy. As  Prop.~\ref{Prop:OutperformsForQuantExample} states, the inner bounds derived in this article can strictly enlarge the largest proven achievable via conventional unstructured IID codes.
\begin{figure}
    \centering
    \includegraphics[width=3.8in]{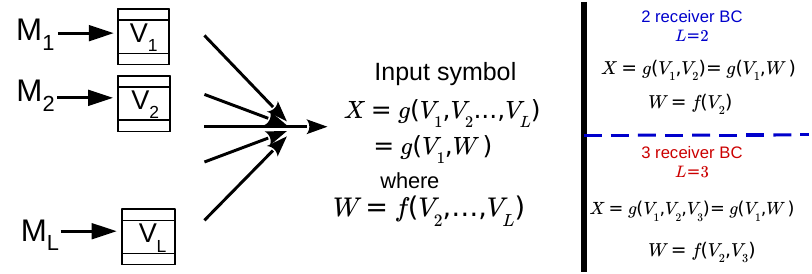}
    \caption{From Rx $1$'s perspective, the interference $W$ it experiences is a specific combination of the other user codewords. While on $2\mhyphen$user BC, $W$ is a univariate function of $V_{2}$, on a $3\mhyphen$user, $W$ is a bi-variate function of $V_{2},V_{3}$}
    \label{Fig:InterferencOn3BC}
\end{figure}

\section{Preliminaries and Problem Statement}
\label{Sec:Preliminaries}
For $n\in \mathbb{N}$, $[n] \define \left\{1,\cdots,n \right\}$. For a Hilbert space $\mathcal{H}$, $\mathcal{L}(\mathcal{H}),\mathcal{P}(\mathcal{H})$ and $\mathcal{D}(\mathcal{H})$ denote the collection of linear, positive and density operators acting on $\mathcal{H}$ respectively. 
We let an \underline{underline} denote an appropriate aggregation of objects. For example, $\ulineCalV \define \CalV_{1}\times \CalV_{2} \times \CalV_{3}$ and in regards to Hilbert spaces $\mathcal{H}_{Y_{i}}: i \in [3]$, we let $\mathcal{H}_{\ulineY} \define \otimes_{i=1}^{3}\mathcal{H}_{Y_{i}}$. For $j \in \{2,3\}$, $\msout{j}$ denotes the complement index, i.e., $\{ j,\msout{j}\} = \{2,3\}$. We dopt the notion of typicality, typical subspaces and typical projectors as in \cite[Chap.~15]{BkWilde_2017}. Specifically, if the collection $(\rho_{x} \in \CalD(\CalH): x \in \CalX)$ is a collection of density operators with spectral decompositions $\rho_{x}=\sum_{y \in \CalY}p_{Y|X}(y|x)\ketbra{e_{y|x}}$ for all $x \in \CalX$ satisfying $\braket{e_{y|x}}{e_{\haty|x}} = \delta_{y\haty}$ for all $x \in \CalX$ and $\rho=\sum_{x \in \CalX}p_{X}(x)\rho_{x}$ has spectral decomposition $\rho=\sum_{y \in \CalY}q(y)\ketbra{f_{y}}$, then 
\begin{eqnarray}
\label{Eqn:TypicalProjector}
\pi^{Y}=\sum_{y^{n}\in\CalY^{n}}\bigotimes_{t=1}^{n}\ketbra{f_{y_{t}}}\mathds{1}_{\{ y^{n} \in T_{\eta}^{n}(q)\}}, ~ \pi_{x^{n}} \define \sum_{y^{n} \in \CalY^{n}}\bigotimes_{t=1}^{n}\ketbra{f_{y_{t}}}\mathds{1}_{\{ (x^{n},y^{n}) \in T_{\eta}^{n}(p_{X}p_{Y|X})\}}, 
\end{eqnarray}
where $T_{\eta}^{n}(q) \subseteq \CalY^{n}$ and $T_{\eta}^{n}(p_{X}p_{Y|X}) \subseteq\CalX^{n}\times \CalY^{n}$ are the typical subsets in $\CalX^{n}$ and $\CalX^{n}\times \CalY^{n}$ with respect to $q_{Y}$ and $p_{X}p_{Y|X}$ respectively. We abbreviate conditional and unconditional typical projector as C-Typ-Proj and U-Typ-Proj respectively.
\begin{remark}
 \label{Rem:TypicalProjector}
While the conditional typical projector defined in \eqref{Eqn:TypicalProjector} is not identical to that defined in \cite[Defn.~15.2.3]{BkWilde_2017}, it is functionally equivalent and yields all the usual typicality bounds in \cite[Chap.~15]{BkWilde_2017}. In particular, in contrast to summing $y^{n}$ over the conditional typical subset $T_{\eta}^{n}(p_{Y|X}|x^{n})$, we have summed over all $y^{n}$ for which $(x^{n},y^{n})$ is an element of the jointly typical set $T_{\eta}^{n}(p_{X}p_{Y|X})$. One consequence of this is that $\pi_{x^{n}} = 0$, the zero projector, whenever $x^{n} \notin T_{\eta}^{n}(p_{X})$.
\end{remark}

Consider a (generic) \textit{$3\mhyphen$CQBC} $(\rho_{x} \in \mathcal{D}(\mathcal{H}_{\ulineY}): x \in \CalX,\kappa)$ specified through (i) a finite set $\mathcal{X}$, (ii) three Hilbert spaces $\mathcal{H}_{Y_{1}}: j \in [3]$, (iii) a collection $( \rho_{x} \in \mathcal{D}(\mathcal{H}_{\ulineY} ): x \in \CalX )$ and (iv) a cost function $\kappa :\mathcal{X} \rightarrow [0,\infty)$. The cost function is assumed to be additive, i.e., the cost of preparing the state $\otimes_{t=1}^{n}\rho_{x_{t}}$ is $\olinekappa^{n}(x^{n}) \define \frac{1}{n}\sum_{t=1}^{n}\kappa(x_{t})$. Reliable communication on a $3\mhyphen$CQBC entails identifying a code. 
\begin{definition}
A \textit{$3\mhyphen$CQBC code} $c=(n,\ulineCalM,e,\ulinelambda)$ consists of three (i) message index sets $[\mathcal{M}_{j}]: j \in [3]$, (ii) an encoder map $e: [\mathcal{M}_{1}] \times [\mathcal{M}_{2}] \times [\mathcal{M}_{3}] \rightarrow \mathcal{X}^{n}$ and (iii) POVMs $\lambda_{j} \define \{ \lambda_{j,m_{j}}: \mathcal{H}_{Y_{j}}^{\otimes n} \rightarrow \mathcal{H}_{Y_{j}}^{\otimes n} : m_{j} \in [\mathcal{M}_{j}] \} : j \in [3]$. The average probability of error of the $3\mhyphen$CQBC code $(n,\ulineCalM,e,\ulinelambda)$ is
\begin{eqnarray}
 \label{Eqn:AvgErrorProb}
 \overline{\xi}(e,\ulinelambda) \define 1-\frac{1}{\mathcal{M}_{1}\mathcal{M}_{2}\mathcal{M}_{3}}\sum_{\ulinem \in \ulineCalM}\tr(\lambda_{\ulinem}\rho_{c,\ulinem}^{\otimes n}).
 \nonumber
\end{eqnarray}
where $\lambda_{\ulinem} \define \otimes_{j=1}^{3}\lambda_{j,m_{j}}$, $\rho_{c,\ulinem}^{\otimes n} \define \otimes_{t=1}^{n}\rho_{x_{t}}$ where $(x_{t}:1\leq t \leq n) = x^{n}(\ulinem) \define e(\ulinem)$. Average cost per symbol of transmitting message $\ulinem \in \ulineCalM \in \tau(e|\ulinem) \define \olinekappa^{n}(e(\ulinem))$ and the average cost per symbol of $3\mhyphen$CQBC code is $\tau(e) \define \frac{1}{|\ulineCalM|}\sum_{\ulinem \in \ulineCalM}\tau(e|\ulinem)$.
\end{definition}
\begin{definition}A rate-cost quadruple $(R_{1},R_{2},R_{3},\tau) \in [0,\infty)^{4}$ is \textit{achievable} if there exists a sequence of $3\mhyphen$CQBC codes $(n,\ulineCalM^{(n)},e^{(n)},\ulinelambda^{(n)})$ for which $\displaystyle\lim_{n \rightarrow \infty}\overline{\xi}(e^{(n)},\ulinelambda^{(n)}) = 0$,
\begin{eqnarray}
 \label{Eqn:3CQICAchievability}
 \lim_{n \rightarrow \infty} n^{-1}\log \mathcal{M}_{j}^{(n)} = R_{j} :j \in [3], \mbox{ and }\lim_{n \rightarrow \infty} \tau(e^{(n)}) \leq \tau .
 \nonumber
\end{eqnarray}
The capacity region $\mathcal{C}$ of the $3\mhyphen$CQBC $(\rho_{x} \in \mathcal{D}(\mathcal{H}_{\ulineY}): x \in \CalX)$ is the set of all achievable rate-cost vectors.
\end{definition}

\section{Coset codes in Broadcast Communication}
\label{Sec:IdeaViaExample}
The main ideas and broad structure of the proposed coding scheme is illustrated through our discussion for Ex.~\ref{Ex:GenAdd3CQBCExample}, \ref{Ex:RotatedExampleWithAdd}. We begin with a classical example which provides a pedagogical step to present the non-commuting Ex.~\ref{Ex:RotatedExampleWithAdd}.
\begin{example}
 \label{Ex:GenAdd3CQBCExample}
Let $\ulineCalX = \CalX_{1}\times \CalX_{2} \times \CalX_{3}$ denote the input set with $\CalX_{j} = \{0,1\}$ for $j \in [3]$. For $b \in \{0,1\}$, let $\sigma_{b}(\eta) \define (1-\eta)\ketbra{b}+ \eta\ketbra{1-b}$. For $\ulinex = (x_{1},x_{2},x_{3}) \in \ulineCalX$, let $\rho_{\ulinex} = \sigma_{x_{1}\oplus x_{2} \oplus x_{3}}(\epsilon) \otimes \sigma_{x_{2}}(\delta) \otimes \sigma_{x_{3}}(\delta)$. The symbol $X_{1}$ is costed via a Hamming cost function i.e., the cost function $\kappa:\ulineCalX \rightarrow \{0,1\}$ is given by $\kappa(\ulinex) = x_{1}$. The average cost constraint $\tau \in (0,\frac{1}{2})$ satisfies $\tau * \epsilon < \delta$ and $h_{b}(\delta) < \frac{1+h_{b}(\tau * \epsilon)}{2}$.
\end{example}
Since $\rho_{\ulinex} : \ulinex \in \ulineCalX$ are commuting, Ex.~\ref{Ex:GenAdd3CQBCExample} can be identified via a $3\mhyphen$user classical BC (Fig.~\ref{Fig:3DBCInExample}) equipped with input set $\ulineCalX$ and output sets $\CalY_{1}=\CalY_{2}=\CalY_{3}=\{0,1\}$. Its input $X=(X_{1},X_{2},X_{3}) \in \ulineCalX$ and outputs $Y_{j}: j \in [3]$ are related via $Y_{1}=X_{1}\oplus X_{2} \oplus X_{3}\oplus N_{1}, Y_{j}=X_{j} \oplus N_{j}$ where $N_{1}, N_{2},N_{3}$ are mutually independent Ber$(\epsilon)$, Ber$(\delta)$, Ber$(\delta)$ RVs respectively. We focus our study on the following question. \textit{What is the maximum achievable rate for Rx $1$ while Rxs $2$ and $3$ are fed at their respective capacities $C_{2}\define C_{3}\define 1-h_{b}(\delta)$?}

\begin{figure}
    \centering
\includegraphics[width=2in]{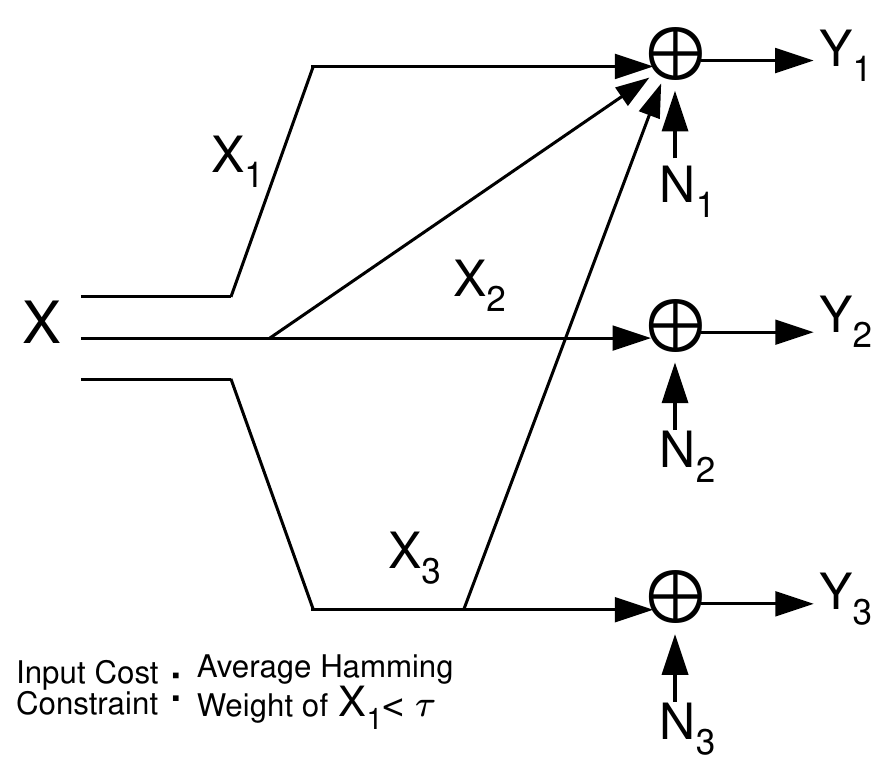}
    \caption{Equivalent Classical $3\mhyphen$user BC of Ex.~\ref{Ex:GenAdd3CQBCExample}.}
    \label{Fig:3DBCInExample}
\end{figure}
Since Rxs $2$ and $3$ can be fed information only through $X_{2}$ and $X_{3}$ we are forced to choose $X_{2},X_{3}$ to be independent Ber$(\frac{1}{2})$ RVs throughout. Observe that Rx $1$ experiences an interference $X_{2}\oplus X_{3}$ which is now a Ber$(\frac{1}{2})$ RV. Furthermore, we recall that the average Hamming weight of $X_{1}$ is constrained to $\tau < \frac{1}{2}$. This has two imports. Firstly, the Tx cannot cancel off this interference that Rx $1$ experiences. Secondly, the maximum rate achievable for Rx $1$ is $C_{1}\define h_{b}(\tau * \epsilon)-h_{b}(\epsilon)$, where $*$ denotes binary convolution. Can the interference mitigation strategies available via unstructured coding achieve a rate $C_{1}$ for Rx $1$? In \cite[Sec.~III]{201804TIT_PadPra}, we prove that no current known unstructured coding scheme can achieve a rate $C_{1}$ for Rx $1$. In the following, we explain the reasoning behind the deficiency of unstructured coding schemes.

The two interference mitigation strategies known are 1) Precoding for interference at the Tx, i.e. Marton's binning and 2) decoding interference at the Rx partly or wholly via Han-Kobayashi's message splitting. Since the unstructured coding schemes builds independent codes with independently picked codewords for every component of the signal, Rx $1$ is unable to decode $X_{2}\oplus X_{3}$ without decoding $X_{2},X_{3}$. A detailed proof in \cite[Sec.~III]{201804TIT_PadPra} establishes that unstructured coding schemes can enable the Rx $1$ decode only univariate functions $V_{2} = f_{2}(X_{2})$ and $V_{3}=f_{3}(X_{3})$ of $X_{2},X_{3}$ efficiently. If it has to decode a bi-variate function $g(X_{2},X_{3})$, the unstructured IID nature of the codes force the Rx to decode $X_{2}$ \textit{and} $X_{3}$.

Can Rx $1$ decode $X_{2},X_{3}$ and achieve a rate $C_{1}$ for itself? Let us compare $\mathscr{C}_{1}$, defined as the capacity of $\ulineX - Y_{1}$ link subject to $p_{X_{2}X_{3}}$ being uniform, with $C_{1}+C_{2}+C_{3}$. We have
\begin{eqnarray}
 \mathscr{C}_{1}\!=\!\sup_{p_{X_{1}}}I(Y_{1};\ulineX)\! = \!1\!-\!h_{b}(\epsilon)<\sum_{i=1}^{3}\!C_{i}\!=\!C_{1}\!+\!2\!-\!2h_{b}(\delta),\!\!\!
 \label{Eqn:UnstructuredConstraint}
\end{eqnarray}

\noindent where the inequality in \eqref{Eqn:UnstructuredConstraint} follows by substituting $C_{1},C_{2},C_{3}$ and noting that $h_{b}(\delta) < \frac{1+h_{b}(\tau * \epsilon)}{2}$ for this example. This precludes Rx $1$ from decoding $X_{2},X_{3}$ \textit{and} achieve a rate of $C_{1}$ for Rx $1$. Any unstructured coding scheme therefore leaves Rx $1$ with residual uncertainty of the interference $X_{2}\oplus X_{3}$. This forces the unstructured coding scheme to resort to the only other interference mitigating strategy - precoding for this residual uncertainty. In contrast to decoding interference, precoding in general entails a \textit{rate loss}. In other words, the rate achievable by Rx $1$ if it can decode the entire interference and subtract it off is strictly larger that what it can achieve if some component of the interference is precoded for by the Tx.
\begin{fact}
 \label{Fact:UnstructuredCodingSchemeSub-Optimal}
 Consider Ex.~\ref{Ex:GenAdd3CQBCExample}. The rate triple $(h_{b}(\tau * \epsilon)-h_{b}(\epsilon), 1-h_{b}(\delta),1-h_{b}(\delta))$ is not achievable via any current known unstructured coding scheme.
\end{fact}

We now present a simple linear coding scheme that can achieve the rate triple $(C_{1},C_{2},C_{3})$. In contrast to building independent codebooks for Rx $2$ and $3$, suppose we choose cosets $\lambda_{2},\lambda_{3}$ of a \textit{common linear code} of rate $1-h_{b}(\delta)$ to communicate to both Rxs $2$ and $3$. Indeed, there exists cosets of a linear code that achieve capacity of a binary symmetric channel. Since the sum of two cosets is another coset of the \textit{same linear} code, the interference patterns seen by Rx $1$ are constrained to another coset $\lambda_{2}\oplus \lambda_{3}$ of the same linear code of rate $1-h_{b}(\delta)$. Instead of attempting to decode the \textit{pair} of Rx $2,3$'s codewords, suppose Rx $1$ decodes its the corresponding sum $\lambda_{2}\oplus \lambda_{3}$. Specifically, suppose Rx ${1}$ attempts to jointly decode it's codeword and the sum of Rx $2$ and $3$'s codewords, the latter being present in $\lambda_{2}\oplus \lambda_{3}$, then it would be able to achieve a rate of $C_{1}$ for itself if 
\begin{eqnarray}
 \mathscr{C}_{1}=1-h_{b}(\epsilon) > C_{1}+ \max\{C_{2},C_{3}\} = C_{1}+1-h_{b}(\delta).\!\!\!
 \label{Eqn:LinCodeCondition}
\end{eqnarray}

\noindent Substituting $C_{1}$, note that \eqref{Eqn:LinCodeCondition} is satisfied since $\tau * \epsilon < \delta$ for Ex.~\ref{Ex:GenAdd3CQBCExample}. Thus the coset code strategy permits Rx $1$ achieve a rate $C_{1}$ even while Rx $2$ and $3$ achieve rates $C_{2},C_{3}$ respectively.
\begin{example}
 \label{Ex:RotatedExampleWithAdd}
Let $\ulineCalX = \CalX_{1}\times \CalX_{2} \times \CalX_{3}$ denote the input set with $\CalX_{j} = \{0,1\}$ for $j \in [3]$. For $b \in \{0,1\}$, let $\sigma_{b}^{1}(\eta) \define (1-\eta)\ketbra{b}+ \eta\ketbra{1-b}$, and for $j=2,3$, let $\sigma_{0}^{j}=\ketbra{0}$ and $\sigma_{1}^{j}=\ketbra{v_{\theta_{j}}}$ where $\ket{v_{\theta_{j}}} = [\cos\theta_{j}~\sin\theta_{j}]^{T}$. For $\ulinex = (x_{1},x_{2},x_{3}) \in \ulineCalX$, let $\rho_{\ulinex} = \sigma_{x_{1}\oplus x_{2} \oplus x_{3}}^{1}(\epsilon) \otimes \sigma_{x_{2}}^{2} \otimes \sigma_{x_{3}}^{3}$. The symbol $X_{1}$ is costed via a Hamming cost function i.e., the cost function $\kappa:\ulineCalX \rightarrow \{0,1\}$ is given by $\kappa(\ulinex) = x_{1}$. Let $0 < \theta_{2}< \theta_{3}<\frac{\pi}{2}$, $\tilde{\delta}_{j} = \frac{1+\cos\theta_{j}}{2}$ and $\tau,\epsilon$ satisfy
\begin{eqnarray}
 \label{Eqn:ConditionForRotatedAdditionExample}
 h_{b}(\tau*\epsilon)+h_{b}(\tilde{\delta}_{2})+h_{b}(\tilde{\delta}_{3}) \stackrel{(a)}{>} 1 \stackrel{(b)}{>} h_{b}(\tau*\epsilon)+h_{b}(\tilde{\delta}_{3}).
\end{eqnarray}
\end{example}
Comparing Exs.~\ref{Ex:GenAdd3CQBCExample} and Ex.~\ref{Ex:RotatedExampleWithAdd}, note that the latter differs from the former only in components of Rxs $2$ and $3$. Moreover, it can be verified that for no choice of basis, is the collection $\rho_{\ulinex}: \ulinex \in \ulineCalX$ commuting. From \cite[Ex.5.6]{BkHolevo_2019}, it is also clear that for $j=2,3$, the capacity of user $j$ is $C_{j}=h_{b}(\tilde{\delta}_{j})$ and is achieved by choosing $p_{X_{j}}$ uniform. Since $\frac{\pi}{2}>\theta_{3}>\theta_{2}>0$, we have $C_{3}>C_{2}$. If we subtract $h_{b}(\epsilon)$ to the two inequalities in \eqref{Eqn:ConditionForRotatedAdditionExample}, we obtain $C_{1}+C_{2}+C_{3}\stackrel{(a)}{>}\mathscr{C}_{1}\stackrel{(b)}{>}C_{1}+\max\{C_{2},C_{3}\}$, where $C_{1},\mathscr{C}_{1}$ are as defined for Ex.~\ref{Ex:GenAdd3CQBCExample}. From our discussion for Ex.~\ref{Ex:GenAdd3CQBCExample} and inequality $(a)$, it is clear that unstructured codes cannot achieve the rate triple $(C_{1},C_{2},C_{3})$. If we can design capacity achieving linear codes $\lambda_{2},\lambda_{3}$ for Rxs $2,3$ respectively in such a way that $\lambda_{2}$ is a \textit{sub-coset} of $\lambda_{3}$, then the interference patterns are contained within $\lambda_{2}\oplus \lambda_{3}$ which is now a coset of $\lambda_{3}$. Rx $1$ can therefore decode into this collection which is of rate atmost $C_{3}=h_{b}(\tilde{\delta}_{3})$. Inequality $(b)$ is analogous to the condition stated in \eqref{Eqn:LinCodeCondition} suggesting that coset codes can be achieve rate triple $(C_{1},C_{2},C_{3})$.

A reader may suspect whether gains obtained above via linear codes are only for `additive' channels. Our study of \cite[Ex.~2]{201804TIT_PadPra} and other settings \cite{201506ISIT_PadPra} analytically prove that we can obtain gains for non-additive scenarios too. To harness such gains, it is necessary to design coding schemes that achieve rates for arbitrary distributions as proved in Thms.~\ref{Thm:3CQBCCodingTheorem1}.

\subsection{Key Elements of a Generalized Coding Scheme}
\label{SubSec:KeyElements}
The codebooks of users $2$ and $3$ being cosets of a common linear code is central to the containment of the sum $X_{2}\oplus X_{3}$. In general users $2$ and $3$ may require codebooks of different rates. We therefore enforce that the smaller of the two codes be a sub-coset of the larger code. Secondly, in Ex.~\ref{Ex:GenAdd3CQBCExample}, users $2$ and $3$ could achieve capacity by choosing codewords that were uniformly distributed, i.e., typical with respect to the uniform distribution on $\{0,1\}$. Since all codewords of a random linear code are typical with respect to the uniform distribution, we could achieve capacity for users $2$ and $3$ by utilizing all codewords from the random linear code. In general, to achieve rates corresponding to non-uniform distributions, we have to enlarge the linear code beyond the desired rate and partition it into bins so that each bin contains codeword of the desired type. This leads us to a \textit{partitioned coset code} (PCC) (Defn.~\ref{Defn:PCC}) that is obtained by partitioning a coset code into bins that are individually unstructured. 
\section{Decoding Interference at a Single Rx}
\label{Sec:AchRegionI}
We provide a pedogogical description of our main results. In Sec.~\ref{SubSec:Rx1DecodesSumOfPvtCodebooks}, we present a coding scheme involving three codebooks. Specifically, each user is equipped with one codebook. Rx $1$ decodes the sum of user $2$ and $3$'s codewords. In Sec.~\ref{SubSec:DecSumOfPublicCdwrds}, we equip users $2$ and $3$ with two codebooks each, one of which is a private codebook whose codeword is decoded only by the corresponding Rx.

\subsection{Decoding Sum of Private Codewords}
\label{SubSec:Rx1DecodesSumOfPvtCodebooks}
We present our first inner bound to the capacity region of a $3\mhyphen$CQBC. In the sequel, we adopt the following notation. Suppose $S_{2},S_{3}$ and $B_{1}$ are real numbers, $U_{1},V_{2},V_{3}$ are random variables. For any $A \subseteq \{2,3\}$ and any $D \subseteq \{1\}$, let $S_{A} \define \sum_{a \in A}S_{a}$, $T_{A} \define \sum_{a \in A}T_{a}$, $B_{D} \define \sum_{d \in D}B_{d}$, $H(V_{A}) \define H(V_{a}: a \in A)$, $H(V_{A},U_{D}) = H(V_{a}: a \in A, U_{d}: d \in D)$ with the empty sum defined as $0$ throughout.\footnote{The term $I(U_{1};V_{2}\oplus_{q}V_{3})$ in bound \eqref{Eqn:ChnlCodingBnd3} was not added in our submission. This omission was an error. However, the evaluation of this region for the Examples \ref{Ex:GenAdd3CQBCExample}, \ref{Ex:RotatedExampleWithAdd} does not change since $U_{1},V_{2},V_{3}$ are chosen independent for achieving the rate triples claimed in those example .}
\begin{theorem}
 \label{Thm:3CQBCCodingTheorem1}
 A rate-cost quadruple $(\ulineR,\tau)$ is achievable if there exists a finite set $\CalU_{1}$, a finite field $\CalV_{2}=\CalV_{3}=\CalW=\fieldq$ of size $q$, real numbers $B_{1}\geq 0,S_{j}\geq 0,T_{j}=\frac{R_{j}}{\log q}\geq 0:j=2,3$ and a PMF $p_{XU_{1}V_{2}V_{3}}$ on $\CalX\times\CalU_{1}\times\CalV_{2}\times\CalV_{3}$ wrt which
 \begin{eqnarray}
 \label{Eqn:SrcCodingBound}
(S_{A}-T_{A})\log q + B_{D} &>& \sum_{d \in D}\!H(U_{d})+\sum_{a\in A}\!H(V_{a})-H(V_{A},U_{D})+|A|\log q -\sum_{a\in A}\!H(V_{a})
\\
\label{Eqn:SrcCodingBoundAlg}
\max\left\{S_{2}\log q+B_{D},S_{3}\log q+B_{D}  \right\}
&>& \log q  - \min_{\theta \in 
\fieldq\setminus\{0\}}H(V_{2}\oplus \theta V_{3}|U_{D}),
\\
\label{Eqn:ChnlCodingBnd1}
R_{1}+B_{1} &<& I(Y_{1},V_{2}\oplus_{q}V_{3};U_{1}), \\
\label{Eqn:ChnlCodingBnd2}
\max\left\{S_{2}\log q,S_{3}\log q \right\}
  &<&    I(Y_{1},U_{1};V_{2}\oplus_{q}V_{3})+\log q - H(V_{2}\oplus_{q} V_{3}),\\
   \label{Eqn:ChnlCodingBnd3}
  \max\left\{ 
R_{1}+B_{1}+S_{2}\log q,R_{1}+B_{1}+S_{3}\log q
  \right\}
  &<&   I(Y_{1};V_{2}\oplus_{q}V_{3},U_{1})+I(U_{1};V_{2}\oplus_{q}V_{3})\log q - H(V_{2}\oplus_{q} V_{3})  \\
  \label{Eqn:ChnlCodingBnd4}
  S_{j}\log q &<& I(Y_{j};V_{j})+\log q - H(V_{j}) : j =2,3,
 \end{eqnarray}
 for all $A \subseteq\{2,3\}$, $D \subseteq\{1\}$, $\sum_{x \in \CalX}p_{X}(x)\kappa(x) \leq \tau$, where all the above information quantities are computed wrt state,
\begin{eqnarray}
\label{Eqn:TheoremState}
\sigma^{\ulineY XU_{1}V_{2}V_{3}W} = \sum_{\substack{x,u_{1},v_{2},v_{3},w}}p_{XU_{1}V_{2}V_{3}W}(x,u_{1},v_{1},v_{2},w)\rho_{x}\otimes  \ketbra{x~u_{1}~v_{2}~v_{3}~w}\mbox{ with}\\
p_{XU_{1}V_{2}V_{3}W}\left( {x,u_{1}, v_{2},v_{3},w} \right)=p_{XU_{1}V_{2}V_{3}}\left({x,u_{1}, v_{2},v_{3}}\right)\mathds{1}_{\{\substack{ w=  v_{2}\oplus_{q} v_{3}}\}}\nonumber
\end{eqnarray}
$\forall (x,u_{1},v_{2},v_{3},w) \in \CalX\times\CalU_{1}\times \CalV_{2}\times \CalV_{3}\times \CalW$.
 \end{theorem}
 \begin{prop}
  \label{Prop:OutperformsForQuantExample}
Consider Ex.~\ref{Ex:RotatedExampleWithAdd}. The rate triple $(h_{b}(\tau*\epsilon)-h_{b}(\epsilon),h_{b}(\tilde{\delta_{2}}),h_{b}(\tilde{\delta_{3}}) )$ is achievable if \eqref{Eqn:ConditionForRotatedAdditionExample}$(b)$ holds. This triple is not achievable via IID codes if \eqref{Eqn:ConditionForRotatedAdditionExample}$(a)$ holds.
  \end{prop}
\noindent In order to prove the first statement of Prop.~\ref{Prop:OutperformsForQuantExample}, we identify a choice of parameters in statement of Thm.~\ref{Thm:3CQBCCodingTheorem1} and verify the bounds. Towards that end, consider $\CalU_{1}=\CalX_{1},\CalV_{j}=\CalX_{j}$, $X_{1}=U_{1}$, $X_{j}=V_{j}$ for $j=2,3$, $p_{X_{1}X_{2}X_{3}}(x_{1},x_{2},x_{3})=\frac{1-\tau}{4}$ if $x_{1}=0$  $p_{X_{1}X_{2}X_{3}}(x_{1},x_{2},x_{3})=\frac{\tau}{4}$ if $x_{1}=1$, $S_{j}=T_{j}=h_{b}(\tilde{\delta_{3}})$, $R_{1}=h_{b}(\tau*\epsilon)-h_{b}(\epsilon)$, $B_{1}=0$. We now verify the bounds in \eqref{Eqn:SrcCodingBound} - \eqref{Eqn:ChnlCodingBnd4}. Since $X_{1}=U_{1},X_{2}=v_{2}$ and $X_{3}=V_{3}$ are independent and moreover $H(V_{2})=H(V_{3})=\log 2$ and $q=2$, it can be verified that the lower bounds in \eqref{Eqn:SrcCodingBound} and \eqref{Eqn:SrcCodingBoundAlg} are $0$. This implies our choice $B_{1}=0$ and $S_{j}=T_{j}$ do not violate these bounds. With the choice for $p_{X_{1}}=p_{U_{1}}$, it can be verified that the upper bound on (i) $R_{1}$ in \eqref{Eqn:ChnlCodingBnd1} is $h_{b}(\tau*\epsilon)-h_{b}(\epsilon)$, (ii) $S_{3}\log q$ in \eqref{Eqn:ChnlCodingBnd2} is $1-h_{b}(\epsilon)$, (iii) $R_{1}+S_{3}\log q$ in \eqref{Eqn:ChnlCodingBnd3} is $1-h_{b}(\epsilon)$ and (iv) $S_{j}\log q$ in \eqref{Eqn:ChnlCodingBnd4} is $h_{b}(\tilde{\delta_{j}})$. Since \eqref{Eqn:ConditionForRotatedAdditionExample}$(b)$ holds, all of these bounds are satisfied. The proof of the second statement is based on our proof of Fact \ref{Fact:UnstructuredCodingSchemeSub-Optimal} that is provided in \cite[Sec.~III]{201804TIT_PadPra}. From the similarity of the examples, this can be verified.
\begin{remark}{\rm 
 The inner bound in Thm.~\ref{Thm:3CQBCCodingTheorem1} is characterized via additional code parameters $B_{1},S_{2},S_{3}$. To characterize an inner bound in terms of only $R_{1},R_{2},R_{3}$, we perform a variable elimination. Instead of using the Fourier Motzkin technique, we leverage the technique proposed in \cite{201108ISIT_ChaEmaZamAre} to perform variable elimination. In Corollary \ref{Cor:3CQBCCodingTheorem1}, we state the resulting inner bound obtained by eliminating variables $B_{1},S_{2}$ and $S_{3}$ in Thm.~\ref{Thm:3CQBCCodingTheorem1}.}
\end{remark}
\begin{proof}
We begin by outlining our techniques and identifying the new elements. The main novelty is in 1) the code structure (Sec.~\ref{SubSec:CodeStructure}) that involves \textit{jointly designed PCCs} built over $\fieldq$ for users $2$, $3$, and 2) the decoding rule (Sec.~\ref{SubSec:DecodingPOVMs}) wherein Rx $1$ decodes the \textit{sum} of users $2$ and $3$'s codewords to facilitate interference peeling. We adopt Marton's encoding (Sec.~\ref{SubSec:Encoding}) of identifying a jointly typical triple, point-to-point decoder POVMs for Rxs $2,3$ and the joint decoder proposed in \cite{201206TIT_FawHaySavSenWil} for Rx $1$. Since the random codewords of users $2,3$ are uniformly distributed and statistically correlated, we are forced to go beyond a `standard information theory proof' and carefully craft our analysis steps. We adapt the steps in \cite[Proof of Thm.~2]{201206TIT_FawHaySavSenWil} to analyze Rx $1$'s error and those in \cite[Proof of Thm.~2]{202107ISIT_AnwPadPra3CQIC} to analyze Rx $2,3$'s error. Specifically, a new element in our analysis is the use of `list threshold' event ($\mathcal{E}_{\ulinem}$) that enables us recover a binning exponent (Rem.~\ref{Rem:ListThresholdEvent}).

\subsection{Code Structure}
\label{SubSec:CodeStructure}
We design a coding scheme with parameters $R_{1},B_{1},S_{2},T_{2},S_{3},T_{3}$ that communicates at rates $R_{1}$ to Rx 1 and $R_{j}=T_{j}\log q$ to Rxs $j=2,3$ respectively. Let $\CalU_{1}$ be a finite set and $\CalV_{2}=\CalV_{3}=\fieldq$ be the finite field of cardinality $q$. Rx $1$'s message is communicated via codebook $\CalC_{1} = (u_{1}^{n}(m_{1},b_{1}) : m_{1} \in [\CalM_{1}] = [2^{nR_{1}}], b_{1} \in [2^{nB_{1}}])$ built over $\CalU_{1}$. As discussed in Sec.~\ref{SubSec:KeyElements}, Rx $2$ and $3$'s messages are communicated via PCCs that are characterized below.
 
\begin{definition}
\label{Defn:PCC}
A partitioned coset code (PCC) built over a finite field $\mathcal{V}=\fieldq$ comprises of (i) a generator matrix $g \in \CalV^{s \times n}$, (ii) a shift vector $d^{n} \in \CalV^{n}$ and (iii) a binning map $\iota : \CalV^{s} \rightarrow \CalV^{t}$. We let $v^{n}(a^{s}) \define a^{s}g \oplus_{q} d^{n} : a^{s} \in \CalV^{s}$ denote its codewords, $c(m^{t}) = \{a^{s}\in \CalV^{s} : \iota(a^{s}) = m^{t}\}$ denote the bin corresponding to message $m^{t}$, $S=\frac{s}{n}$ and $T = \frac{t}{n}$. When clear from context, we ignore superscripts $s$ in $a^{s}$ and $t$ in $m^{t}$. We refer to this as PCC $(n,S,T,g,d^{n},\iota)$ or PCC $(n,S,T,g,d^{n},c)$.
\end{definition}
For $j=2,3$, Rx $j$'s message is communicated via PCC $\lambda_{j}$. In Sec.~\ref{SubSec:KeyElements}, we noted that the smaller of these codes must be a sub-coset of the larger code. Without loss of generality assume $\lambda_{3}$ is larger than $\lambda_{2}$, we let $\lambda_{j}$ be the PCC $(n,S_{j},T_{j},g_{j},d_{j}^{n},\iota_{j})$ where $g_{3} = \left[ g_{2}^{T}~ g_{3/2}^{T} \right]^{T}$, $g_{2} \in \CalV^{s_{2} \times n}$ and $g_{3/2}  \in \CalV^{(s_{3}-s_{2}) \times n}$. This guarantees that the collection of vectors obtained by adding all possible pairs of codewords from $\lambda_{2}$ and $\lambda_{3}$ is the collection $\lambdapl \define (w^{n}(\apl) \define \apl g_{3}\oplus_{q}\dpl^{n} : \apl \in \CalV^{s_{3}} )$ where $\dpl \define d_{2}^{n}\oplus_{q} d_{3}^{n}$. For $j=2,3$ let $\CalM_{j} = q^{t_{j}}$ and $[\CalM_{j}] \define [q^{t_{j}}]$ denote Rx $j$ message index set.
 
\subsection{Encoding}
\label{SubSec:Encoding}
The triple $(u_{1}^{n}(m_{1},b_{1}): b_{1} \in [2^{nB_{1}}] ),c_{2}(m_{2}), c_{3}(m_{3})$ of bins correspond to the available choice of codewords that the encoder can use to communicate the message triple $\ulinem \define (m_{1},m_{2},m_{3}) \in [\ulineCalM]$. Let
 \begin{eqnarray}
\CalL(\ulinem) \define \left\{\!\!\!\begin{array}{c} (b_{1},a_{2},a_{3}) :  b_{1} \in [2^{nB_{1}}], a_{j} \in c_{j}(m_{j}): j =2,3 \\ \ni (u_{1}(m_{1},b_{1}),v_{2}(a_{2}),v_{3}(a_{3})) \in T_{\eta}(p_{U_{1}V_{2}V_{3}}),  \end{array}\!\!\!\right\}  
  \nonumber
 \end{eqnarray}
be a list of typical triples among this choice and $\alpha(\ulinem) \define |\CalL(\ulinem)|$. Let $(b_{1}(\ulinem),a_{2}(\ulinem),a_{3}(\ulinem) )$ be a triple chosen from $\CalL(\ulinem)$ if $\alpha(\ulinem) \geq 1$. Otherwise, set $(b_{1}(\ulinem),a_{2}(\ulinem),a_{3}(\ulinem) ) = (1,0^{s_{2}},0^{s_{3}})$. A `fusion map' $f:\CalU_{1}^{n} \times \CalV_{2}^{n} \times \CalV_{3}^{n} \rightarrow \CalX^{n}$ is used to map $( u_{1}^{n}(m_{1},b_{1}(\ulinem)),v_{2}(a_{2}(\ulinem)),v_{3}(a_{3}(\ulinem)) )$ into an input sequence henceforth denoted as $x^{n}(\ulinem) \define (x(\ulinem)_{t} : t \in [n])$. To communicate a message triple $\ulinem$, the encoder prepares the state $\rho_{\ulinem} \define \otimes_{t=1}^{n}\rho_{x(\ulinem)_{t}}$.
\subsection{Decoding POVMs}
\label{SubSec:DecodingPOVMs}
In addition to $m_{1}$, Rx $1$ aims to decode the sum $v_{2}(a_{2}(\ulinem)) \oplus v_{3}(a_{3}(\ulinem))$ of the codewords chosen by the Tx. To reference this sum succinctly, henceforth we let $a_{\oplus}(\ulinem) = a_{2}(\ulinem) 0^{s_{3}-s_{2}} \oplus_{q} a_{3}(\ulinem)$. Recollecting code $\lambdapl$, it can be verified that the sum $v_{2}(a_{2}(\ulinem)) \oplus v_{3}(a_{3}(\ulinem)) =  w^{n}(a_{\oplus}(\ulinem))$. Rx $1$ employs the decoding POVM \cite[Proof of Thm.~2]{201206TIT_FawHaySavSenWil}
\begin{eqnarray}
 \theta_{m_{1},b_{1}}^{\apl} \!\!\define \!\!\left( \sum_{\hatm_{1},\hatb_{1},\hatapl}\!\!\!\! \lambda_{\hatm_{1},\hatb_{1}}^{\hatapl} \!\!\right)^{\!\!-\frac{1}{2}}\!\!\!\!\!\lambda_{m_{1},b_{1}}^{\apl}\!\!\left( \sum_{\hatm_{1},\hatb_{1},\hatapl}\!\!\!\lambda_{\hatm_{1},\hatb_{1}}^{\hatapl} \!\right)^{\!\!-\frac{1}{2}}
 \nonumber
\end{eqnarray}
where $\lambda_{m_{1},b_{1}}^{\apl}\define \pi^{Y_{1}}\pi_{m_{1},b_{1}}\pi_{m_{1},b_{1}}^{\apl}\pi_{m_{1},b_{1}}\pi^{Y_{1}}$,
\begin{eqnarray}
\label{Eqn:Rx1DecPOVMs}
\pi_{m_{1},b_{1}} \define \pi_{u_{1}^{n}(m_{1},b_{1})}^{Y_{1}}  
, ~~~ \pi^{\apl}_{m_{1},b_{1}} \define \pi^{Y_{1}}_{u_{1}^{n}(m_{1},b_{1}),w^{n}(\apl)}. 
\end{eqnarray}
are the C-Typ-Projs. with respect to states to $\rho^{Y_{1}}_{u_{1}} = \tr_{Y_{2}Y_{3}}(\sum_{x}p_{X|U_{1}}(x|u_{1})\rho_{x}^{\ulineY} )$ and 
\begin{eqnarray}
\rho_{u_{1}w}^{Y_{1}} \define \tr_{Y_{2}Y_{3}}\left( \sum_{x}p_{X|U_{1}W}(x|u_{1},w)\rho_{x}^{\ulineY}  \right)\nonumber                                                                                                         \end{eqnarray}
respectively.

Rxs $2$ and $3$ decode only their messages. Rx $j$ employs the POVM $\{ \theta_{m_{j}} = \displaystyle\!\!\!\!\!\!\!\sum_{a_{j} \in c_{j}(m_{j})}\!\!\!\!\!\!\!\!\theta_{a_{j}} : m_{j} \in [\CalM_{j}]\}$ where
\begin{eqnarray}
\label{Eqn:POVMOfDecoder2And3}
 \theta_{a_{j}} \!\define\! \left(\!\! \sum_{\hata_{j}}\pi^{Y_{j}}\pi_{\hata_{j}}\pi^{Y_{j}}\!\! \right)^{\!\!\!-\frac{1}{2}}\!\!\! \pi^{Y_{j}}\pi_{a_{j}}\pi^{Y_{j}}\!\!\left( \sum_{\hata_{j}}\pi^{Y_{j}}\pi_{\hata_{j}}\pi^{Y_{j}} \right)^{\!\!\!-\frac{1}{2}}\!\!\!\!\!\!,
 \nonumber
\end{eqnarray}
$\pi^{Y_{2}},\pi^{Y_{3}}$ are the U-Typ-Proj of $\rho^{Y_{2}},\rho^{Y_{3}}$ respectively,  and for $j=2,3$, $\pi_{a_{j}} \define \pi_{v_{j}^{n}(a_{j})}$ is the C-Typ-Proj wrt state $\rho_{v_{j}}^{Y_{j}} \define \tr_{Y_{1}Y_{\msout{j}}}(\sum_{x}p_{X|V_{j}}(x|v_{j})\rho_{x}^{\ulineY} )$.

\subsection{Probability of Error Analysis}
\label{SubSec:ProbErrorAnalysis}
We first state the distribution of the random code. Codewords of $\CalC_{1}$ are picked IID $\prod p_{U_{1}}$. The generator matrix $G_{3}$, shifts $D_{2}^{n},D_{3}^{n}$, bin indices $(\iota_{j}(a_{j}) : a_{j}\in \fieldq^{s_{j}}):j =2,3$ are all picked mutually independently and uniforSubSec:DecSumOfPublicCdwrdsmly from their respective ranges. If $\alpha(\ulinem) \geq 1$, $(B_{1}(\ulinem),A_{2}(\ulinem),A_{3}(\ulinem))$ is picked uniformly from $\CalL(\ulinem)$. Otherwise, $(B_{1}(\ulinem),A_{2}(\ulinem),A_{3}(\ulinem)) \define (1,0^{s_{2}},0^{s_{3}})$. All of the above objects are also mutually independent. We emphasize that $(B_{1}(\ulinem),A_{2}(\ulinem),A_{3}(\ulinem))$ is \textit{conditionally independent} of $\CalC_{1}, G_{3},D_{2}^{n},D_{3}^{n}$, $(\iota_{j}(a_{j}) : a_{j}\in \fieldq^{s_{j}}):j =2,3$ \textit{given the event} $\{\alpha(\ulinem)\geq 1\}$, a fact ({\textbf{Note 1}}) we shall use at a later point in our analysis. Finally, $X^{n}(\ulinem)$ is picked according to $p_{X|U_{1}V_{2}V_{3}}^{n}(\cdot|V_{1}^{n}(m_{1},b_{1}(\ulinem)),V_{2}(A_{2}(\ulinem)),V_{3}(A_{3}(\ulinem)) )$.

The average probability of error is
\begin{eqnarray}
 \label{Eqn:FirstCodingThmProof_1}
P(\mbox{Err}) \!= \!\frac{1}{|\ulineCalM|}\! \sum_{ \ulinem  \in [\ulineCalM]}\!\!\! {\tr\!\left\{\!\left(\! I_{{\ulineY}}^{\otimes n}\!\!-\!\theta^{\apl(\ulinem)}_{m_{1},b_{1}}\otimes \theta_{m_{2}}\otimes \theta_{m_{3}} \!\right)\!\rho_{\ulinem} \right\}\!}. 
   \end{eqnarray}
Henceforth, we analyze a generic term $\xi_{\ulinem}$ in the sum \eqref{Eqn:FirstCodingThmProof_1}. Define a `list threshold' event $\CalE_{\ulinem} \define \{ \alpha(\ulinem) \geq 2^{n(\tau-\eta)\}}$ where 
\begin{eqnarray}
 \label{Eqn:Threshold}
 \tau \define B_{1}\!+\!\sum_{j=2}^{3}\!(S_{j}-T_{j})\log q\! -\! 2\log q\!+\!H(V_{2},V_{3}|U_{1}\!),
\end{eqnarray}

\noindent is\footnote{In our submission, the last term in \eqref{Eqn:Threshold} appeared erroneously with a negative sign as $-H(V_{2},V_{3}|U_{1})$.} the exponent of the expected number of jointly typical triples in any triple of bins and $\mathds{1}_{\CalE_{\ulinem}}$ be its indicator. Firstly, note that every term  $\xi_{\ulinem}$ in the above sum is at most $1$. Hence $\xi_{\ulinem} \mathds{1}_{\CalE^{c}_{\ulinem}} \leq \mathds{1}_{\CalE^{c}_{\ulinem}}$. Next, note that the operator inequalities $0 \leq \theta^{\apl(\ulinem)}_{m_{1},b_{1}} \leq I_{Y_{1}}$ and $0 \leq \theta_{m_{j}} \leq I_{Y_{j}}$ hold. We therefore have
\begin{eqnarray}
 \label{Eqn:FirstCodingThmProof_2}
 \xi_{\ulinem} \!\define\! \tr\!\left\{\!\!\left(\! I_{{\ulineY}}^{\otimes n}-\theta^{\apl(\ulinem)}_{m_{1},b_{1}}\!\!\otimes\! \theta_{m_{2}}\!\!\otimes\! \theta_{m_{3}}\! \!\right)\!\rho_{\ulinem}\! \right\}\! \leq T_{0}+T_{1}+T_{2} + T_{3}\!\!\!\!\!\!\!\!\!\!\!\!
 \nonumber\\
 \label{Eqn:FirstCodingThmProof_3}
 \mbox{where }T_{0} =\mathds{1}_{\CalE_{\ulinem}^{c}},~ T_{1} \define \tr_{Y_{1}}\!\!\left\{ \!\left( I_{Y_{1}}-\theta_{m_{1},b_{1}}^{\apl(\ulinem)} \right)\!\rho_{\ulinem}^{Y_{1}}\!\right\}\mathds{1}_{\CalE_{\ulinem}},\\
 \label{Eqn:FirstCodingThmProof_4}
  T_{j} \define \tr_{Y_{j}}\!\!\left\{ \left(I_{Y_{j}}-\theta_{m_{j}} \right) \!\rho_{\ulinem}^{Y_{j}}\!\right\}\mathds{1}_{\CalE_{\ulinem}}\mbox{ for }j=2,3,
 \end{eqnarray}
where $\xi_{\ulinem}\mathds{1}_{\CalE_{\ulinem}} \leq T_{1}+T_{2}+T_{3}$ is true because whenever operators $A_{i}$ satisfy $0 \leq A_{i} \leq I_{Y_{i}}$ for $i \in [3]$, we have $I_{Y_{1}}\!\otimes\! I_{Y_{2}}\!\otimes\! I_{Y_{3}} \!-\! A_{1}\!\otimes\! A_{2}\!\otimes\! A_{3} \!\leq\! I_{Y_{1}}\!\otimes\! I_{Y_{2}} \!\otimes\! (I_{Y_{3}}\!-\!A_{3})\!+\! I_{Y_{1}}\!\otimes\! (I_{Y_{2}}\!-\!A_{2})\!\otimes\! I_{Y_{3}}\!  +\! (I_{Y_{1}}\!-\!A_{1})\!\otimes\! I_{Y_{2}}\!\otimes\! I_{Y_{3}}$. This is analogous to a `union bound' \cite[Eqn.~\(78\)]{201512TIT_SavWil} in classical probability. The rest of our proof analyzes each of $T_{0},T_{1},T_{2}$ and $T_{3}$.
 \begin{remark}
  \label{Rem:ListThresholdEvent} We have tagged along $\mathds{1}_{\CalE_{\ulinem}}$ to suppress a binning exponent in our pre-variable-elimination bounds. While this does not enlarge the rate region, it enables our variable elimination to yield a compact description of the inner bound. Specifically, if we had not tagged along this event and not suppressed the binning exponent, then a variable elimination performed on the characterization in Thm.~\ref{Thm:3CQBCCodingTheorem1} would yield far more inequalities post variable elimination as compared to that provided in Corollary \ref{Cor:3CQBCCodingTheorem1}.
 \end{remark}

 \med\textit{Analysis of $T_{0}$} : Since $T_{0}$ involves only classical probabilities, its analysis is identical to that in \cite[App.~1]{201804TIT_PadPra} and we therefore refer the reader to \cite[App.~1]{201804TIT_PadPra} for a proof of Prop.~\ref{Prop:TheClassicalCovering}.
\begin{prop}
 \label{Prop:TheClassicalCovering}
 For any $\eta > 0$, there exists $N_{\eta} \in \naturals$ such that for all $n \geq N_{\eta}$, we have $\mathbb{E}\{ T_{0}\}  \leq \exp\{-n\eta\}$ if the bounds \eqref{Eqn:SrcCodingBound}, \eqref{Eqn:SrcCodingBoundAlg} in the Thm.~\ref{Thm:3CQBCCodingTheorem1} statement holds for all $A \subseteq \{2,3\}$ and every $D \subseteq \{1\}$ with the empty sum being defined as $0$.
 \end{prop}
 To comprehend the above bounds, note that the first $3$ terms on the RHS form the usual lower bound in classical covering. The extra $|A|\log q - \sum_{a \in A}H(V_{a})$ is the penalty in binning rate we pay since the codewords of the linear code are uniformly distributed. Indeed, this is the divergence $D(p_{V_{A}}||\mbox{Unif}^{|A|})$ between the desired distribution and the uniform on $\fieldq^{|A|}$.
\med\textit{Analysis of $T_{1}$} : Let $\pi_{\apl(\ulinem)} = \pi_{w^{n}(\apl(\ulinem))}$ be the conditional typical projector with respect to the state $\rho_{w}^{Y_{1}} \define \tr_{Y_{2}Y_{3}}\left( \sum_{x}p_{X|W}(x|w)\rho_{x}^{\ulineY} \right)$. Since we have fixed an arbitrary $\ulinem$, we let $\apl = \apl(\ulinem)$ and $\pi_{\apl} \define \pi_{\apl(\ulinem)} $in the sequel. We adapt \cite[Proof of Thm.~2]{201206TIT_FawHaySavSenWil} to derive our upper bound on $T_{1}$. From the definition of $T_{1}$ in \eqref{Eqn:FirstCodingThmProof_3} and the fact that $\tr(\lambda\rho) \leq \tr(\lambda\sigma)+\norm{\rho-\sigma}_{1}$ for $0 \leq \rho,\sigma,\lambda \leq I$, we have
\begin{eqnarray}
 \label{Eqn:FirstCodingThmProof_6}
 T_{1} \leq \tr_{Y_{1}}\!\!\left\{ \!\left( I_{Y_{1}}-\theta_{m_{1},b_{1}}^{\apl} \right)\!\pi_{\apl}\rho_{\ulinem}^{Y_{1}}\pi_{\apl}\!\right\}\mathds{1}_{\CalE_{\ulinem}} + T_{10}
\end{eqnarray}
where $T_{10} = \norm{\pi_{\apl}\rho_{\ulinem}^{Y_{1}}\pi_{\apl} - \rho_{\ulinem}^{Y_{1}}}_{1}\mathds{1}_{\CalE_{\ulinem}}$. Since $\alpha(\ulinem) \geq 1$, the chosen triple of codewords is jointly typical and by pinching and the gentle operator lemma, we have $\mathbb{E}\{T_{10}\} \leq \exp\{-n\eta\}$ for sufficiently large $n$. Denoting the first term in \eqref{Eqn:FirstCodingThmProof_6} as $\overline{T}_{11}$ and applying the Hayashi-Nagaoka inequality \cite{200307TIT_HayNag} on $\overline{T}_{11}$, we have $\overline{T}_{11} \leq 2(I-T_{11})\mathds{1}_{\CalE_{\ulinem}} + 4(T_{12}+T_{13} + T_{14})\mathds{1}_{\CalE_{\ulinem}}$, where
\begin{eqnarray}
\label{Eqn:FirstCodingThmProof_7}
T_{11} \define \tr\!\left( \!\lambda_{m_{1},b_{1}}^{\apl} \pi_{\apl}\rho_{\ulinem}^{Y_{1}}\pi_{\apl}\!\right)\!, T_{12}\! \define \!\!\!\!\!\!\sum_{\substack{(\hatm_{1},\hatb_{1}) \neq \\ (m_{1},b_{1}) }} \!\!\!\!\!\!\!\tr(\!\lambda_{\hatm_{1},\hatb_{1}}^{\apl}\!\!\pi_{\apl}\rho_{\ulinem}^{Y_{1}}\pi_{\apl}\!\!)\nonumber\\
T_{13}\! \define \!\!\!\!\sum_{\substack{\hatapl \neq \apl}} \!\!\!\!\!\tr(\!\lambda_{m_{1},b_{1}}^{\hatapl}\!\!\pi_{\apl}\rho_{\ulinem}^{Y_{1}}\pi_{\apl}\!\!), T_{14} \! \define \!\!\!\!\!\!\!\!\!\!\!\!\!\!\sum_{\substack{\hatapl \neq \apl \\ (\hatm_{1},\hatb_{1}) \neq (m_{1},b_{1}) }} \!\!\!\!\!\!\!\!\!\!\!\!\!\!\!\tr(\!\lambda_{\hatm_{1},\hatb_{1}}^{\hatapl}\!\!\pi_{\apl}\rho_{\ulinem}^{Y_{1}}\pi_{\apl}\!\!)
\nonumber
\end{eqnarray}
Analysis of $T_{11}$ is straightforward involving repeated use of the gentle operator lemma. See \cite{202202arXiv_Pad3CQBC} for detailed steps.
\med\textit{Analysis of $T_{12},T_{13}, T_{14}$} : We pull forward our new steps that enable us analyze $T_{12}, T_{13}, T_{14}$ and illustrate them below in a unified manner. Abbreviating $\ulinev^{n}=(v_{2}^{n},v_{3}^{n})$, $\ulineV=(V_{2},V_{3})$, recalling $\apl = \apl(\ulinem)$ and with a slight abuse of notation, let
\begin{eqnarray}
\label{Eqn:Rx1AnalysisF1}
 \CalF_{1}\!\! \define\!\! \left\{\!\!\! \!
 \begin{array}{c} 
 V_{j}^{n}(a_{j})=v_{j}^{n}, \iota_{j}(a_{j}) = m_{j}\!:\!j\!=\!2,3, U_{1}(m_{1},b_{1}) = u_{1}^{n}
 ,W^{n}(\apl) = w^{n}
 \end{array} \!\!\!\!\right\}\!, \CalB\! \define\! {\left\{ \!{(u_{1}^{n},\ulinev^{n},w^{n})  \in T_{\eta}\!(p_{U_{1}\ulineV W}\!) }\! \right\}}
 \\
 \label{Eqn:Rx1AnalysisF2F3}
 \CalF_{2}\! \define\!\left\{ \!\!\! \! \begin{array}{c}B_{1}(\ulinem)=b_{1}, A_{j}(\ulinem)= a_{j}  \end{array}\!\!\! : \!\!\begin{array}{c}j=2,3 \end{array}\!\!\!\right\},~ \CalG_{12} = \left\{ \!\!\!\! \begin{array}{c}U_{1}^{n}(\hatm_{1},\hatb_{1})=\hatu_{1}^{n}, (\hatm_{1},\hatb_{1}) \neq (m_{1},b_{1})\end{array} \!\!\! \right\} 
 \nonumber\\
\CalF_{3}\!\define \!\{X^{n}(\ulinem) = x^{n} \!\},  \CalG_{13}\define \left\{\!\!\!\! \begin{array}{c} W^{n}(\hatapl) = \hatw^{n},\hatapl \neq \apl \end{array}\!\!\!\! \right\} ,
  \label{Eqn:Rx1AnalysisG14}
 \CalG_{14} \define \CalG_{12} \cap \CalG_{13}, \CalA \define \left\{ \!\!\!\! \begin{array}{c} (\hatu_{1}^{n},w^{n})\in T_{\eta}(p_{U_{1}W}) \end{array}\!\!\!\!\right\},
\end{eqnarray}
where the event $\CalF_{1}\cap \CalF_{2} \cap \CalF_{3}$ specifies the realization for the chosen codewords and $\CalG_{12},\CalG_{13},\CalG_{14}$ specify the realization of an incorrect codeword in regards to $T_{12}, T_{13},T_{14}$ respectively. We let $\lambda_{u_{1}^{n},w^{n}}\define \pi^{Y_{1}}\pi_{u_{1}^{n}} \pi_{ u_{1}^{n}w^{n}}\pi_{u_{1}^{n}} \pi^{Y_{1}}$ as in \eqref{Eqn:Rx1DecPOVMs}, $\CalE_{1k} \define \CalF_{1}\cap\CalG_{1k} \cap \CalE_{\ulinem} \cap  \CalF_{2}\cap \CalF_{3}\cap \CalA \cap \CalB$ for $k=2,3,4$ and note that
\begin{eqnarray}
 \label{Eqn:Rx1ErrorEventChars-12}
 \Expectation\{T_{12}\mathds{1}_{\CalE_{\ulinem}} \} =\!\!\!\!\!\! \sum_{\substack{b_{1},a_{2},a_{3}\\u_{1}^{n},v_{2}^{n},v_{3}^{n}}}\sum_{\substack{w^{n},x^{n}\\\hatm_{1},\hatb_{1},\hatu_{1}^{n}}}\!\!\!\! \!\!
 \tr(\lambda_{\hatu_{1}^{n},w^{n}}\!\pi_{w^{n}}\rho_{x^{n}}\pi_{w^{n}} \!)P(\CalE_{12})
 \\
 \label{Eqn:Rx1ErrorEventChars-13}
 \Expectation\{T_{13}\mathds{1}_{\CalE_{\ulinem}}\} = \!\!\!\!\!\!\!\sum_{\substack{b_{1},a_{2},a_{3}\\u_{1}^{n},v_{2}^{n},v_{3}^{n}}} \sum_{\substack{w^{n},x^{n}\\\hatapl, \hatw^{n} }}\!\!\!\tr(\lambda_{u_{1}^{n},\hatw^{n}}\pi_{w^{n}}\rho_{x^{n}}\pi_{w^{n}} \!)P(\CalE_{13})
 \\
 \label{Eqn:Rx1ErrorEventChars-14}
 \Expectation\{T_{14}\mathds{1}_{\CalE_{\ulinem}}\} = \!\!\!\!\!\!\!\!\!\sum_{\substack{b_{1},a_{2},a_{3},u_{1}^{n},v_{2}^{n},v_{3}^{n}\\w^{n},x^{n} ,\hatm_{1},\hatb_{1},\hatapl, \hatu_{1}^{n},\hatw^{n} }} \!\!\!\!\!\!\!\!\!\!\!\!\!\!\!\!\tr(\lambda_{\hatu_{1}^{n},\hatw^{n}}\!\pi_{w^{n}}\rho_{x^{n}}\pi_{w^{n}} \!)P(\CalE_{14}).
\end{eqnarray}
We have
\begin{eqnarray}
 \label{Eqn:ProbabilityofEevents1}
 \lefteqn{P(\CalE_{1k}) =  P(\CalE_{\ulinem}\cap\CalF_{1}\cap\CalG_{1k})P(\CalF_{2}|\CalE_{\ulinem}\cap\CalF_{1}\cap\CalG_{1k})}\nonumber\\
 \!\! \!\!&\!\!\!\!&\times P(\CalF_{3}|\CalF_{2}\cap\CalE_{\ulinem}\cap\CalF_{1}\cap\CalG_{1k})
 \nonumber\\
 \label{Eqn:ProbabilityofEevents2}
 \!\! \!\!&\!\!\leq\!\!&\!\frac{P(\CalF_{1}\!\cap\CalG_{1k})}{2^{n(\tau-\eta)}}\!\!\cdot p^{n}_{X|U_{1}\ulineV W}(x^{n}|u_{1}^{n}\!,\!\ulinev^{n}\!,\!w^{n})\mathds{1}_{\left\{\!\substack{w^{n}=\\v_{2}^{n}\!\oplus_{q}\! v_{3}^{n}}\!\right\}},
\end{eqnarray}
where we have used \textbf{Note 1} stated earlier in substituting the upper bound $2^{-n(\tau-\eta)}$ on $P(\CalF_{2}|\CalE_{\ulinem}\cap\CalF_{1}\cap\CalG_{1k})$ and
\begin{eqnarray}
\label{Eqn:DistributionalReslation}
p_{X|U_{1}V_{2}V_{3}}(x|u_{1},v_{2},v_{3}) = p_{X|U_{1}V_{2}V_{3}W}(x|u_{1},v_{2},v_{3},v_{2}\oplus_{q}\!v_{3})
.\nonumber
\end{eqnarray}
From the distribution of the random code, we have
\begin{eqnarray}
 \label{Eqn:Rx1ErrorEventAssociatedProbs}
 P(\CalF_{1}\!\cap\!\CalG_{12}\!) \!=\! \frac{p_{U_{1}}^{n}\!(u_{1}^{n})p_{U_{1}}^{n}\!(\hatu_{1}^{n})}{q^{n(2+T_{2}+T_{3})}}
\!, P(\!\CalF_{1}\!\cap\!\CalG_{13}\!) \!=\! \frac{p_{U_{1}}^{n}\!(u_{1}^{n})}{\!q^{n(3+T_{2}+T_{3})}\!}
\end{eqnarray}
and $P(\CalF_{1}\!\cap\CalG_{14})=q^{-n}P(\CalF_{1}\!\cap\CalG_{12})$. We note that the above probabilities differ from a conventional information theory analysis wherein the $V_{2},V_{3}-$codebooks are picked IID $p_{V_{2}},p_{V_{3}}$.  The next step involves substituting these probabilities in \eqref{Eqn:Rx1ErrorEventChars-12}-\eqref{Eqn:Rx1ErrorEventChars-14} and evaluating upper bounds on the same.
\begin{prop}
\label{Prop:Rx1Errors}
For all $\eta>0$, there exists $N_{\eta} \in \naturals$ such that $\forall n \geq N_{\eta}$, we have $\Expectation\{T_{1k}\mathds{1}_{\CalE_{\ulinem}} \} \leq \exp\{-n\eta\}$ for $k=2,3,4$ if
bounds \eqref{Eqn:ChnlCodingBnd1}, \eqref{Eqn:ChnlCodingBnd2} and \eqref{Eqn:ChnlCodingBnd3} hold wrt state \eqref{Eqn:TheoremState}.
\end{prop}

\med\textit{Analysis of $T_{2},T_{3}$} : Analysis of $T_{2},T_{3}$ is identical and we state the same in terms of a generic index $j$. From \eqref{Eqn:FirstCodingThmProof_4}, we have
\begin{eqnarray}
 \label{Eqn:FirstCodingThmBoundingTj}
 T_{j} \leq \tr\!\left\{\! \left(\! I_{Y_{j}} - \theta_{a_{j}(\ulinem)}\!\right)\!\rho_{\ulinem}^{Y_{j}} \!\right\}\mathds{1}_{\CalE_{\ulinem}} \leq 2(I\!-\!T_{j0})\mathds{1}_{\CalE_{\ulinem}}\!+4T_{j1}\mathds{1}_{\CalE_{\ulinem}}
 \nonumber\\
 T_{j0} \!\define\! \tr(\!\pi^{Y_{j}}\pi_{a_{j}(\ulinem)}\pi^{Y_{j}}\rho_{\ulinem}^{Y_{j}}\!), T_{j1}\! \define\!\!\!\!\! \sum_{\hata_{j} \neq a_{j}(\ulinem)} \!\!\!\!\!\!\!\tr(\!\pi^{Y_{j}}\pi_{\hata_{j}(\ulinem)}\pi^{Y_{j}}\rho_{\ulinem}^{Y_{j}}\!).
 \nonumber
\end{eqnarray}
Bounding of $T_{j0},T_{j1}$ is similar to the proof of nested coset codes achieving capacity of a CQ channel \cite[Thm.~2]{202107ISIT_AnwPadPra3CQIC}. The change in our proof of Prop.~\ref{Prop:Rx2And3ErrorEvent} is the use of the `list threshold' event that suppresses the binning exponent.The proof of Prop.~\ref{Prop:Rx2And3ErrorEvent} is provided in Appendix \ref{AppSec:ErrorEventAtRx2And3} and this concludes our proof.
\begin{prop}
 \label{Prop:Rx2And3ErrorEvent}
 For all $\eta>0$, there exists $N_{\eta} \in \naturals$ such that $\forall n \geq N_{\eta}$, we have $\Expectation\{T_{j1}\mathds{1}_{\CalE_{\ulinem}} \} \leq \exp\{-n\eta\}$ if
bounds \eqref{Eqn:ChnlCodingBnd4} in Thm.~\ref{Thm:3CQBCCodingTheorem1} statement holds wrt state \eqref{Eqn:TheoremState}.
\end{prop}
\end{proof}

We now state the inner bound in Thm.~\ref{Thm:3CQBCCodingTheorem1} after eliminating variables $B_{1},S_{2}$ and $S_{3}$ via the variable elimination technique proposed in \cite{201108ISIT_ChaEmaZamAre}.

\begin{corollary}
 \label{Cor:3CQBCCodingTheorem1}
 A rate-cost quadruple $(\ulineR,\tau)$ is achievable if there exists a finite set $\CalU_{1}$, a finite field $\CalV_{2}=\CalV_{3}=\CalW=\fieldq$ of size $q$, a PMF $p_{XU_{1}V_{2}V_{3}}$ on $\CalX\times\CalU_{1}\times\CalV_{2}\times\CalV_{3}$ and some $l \in \{2,3\}$ for which
 \begin{eqnarray}
\label{Eqn:3CQBCRateRegion1}
R_{j}&<&\Upsilon_{j}+H(V_{j})
\\
\label{Eqn:3CQBCRateRegion2}
R_{2}+R_{3}&<&\Upsilon_{2}+\Upsilon_{3}+H(V_{2},V_{3})
\\\label{Eqn:3CQBCRateRegion3}
R_{1}&<&\min\left\{ \!\!\!\begin{array}{c} I(U_{1};Y_{1},W),I(U_{1},W;Y_{1})+I(U_{1};W)+\gamma_{12}-H(W),I(U_{1};W,Y_{1})+\gamma_{12}+\Upsilon_{l}\end{array} \!\!\right\}
\\\label{Eqn:3CQBCRateRegion4}
R_{1}+R_{\msout{l}}&<&\min\left\{\!\!\! \begin{array}{c} 
H(V_{\msout{l}})-I(U_{1};V_{\msout{l}})+I(U_{1},W;Y_{1})+I(U_{1};W)-H(W),\\ H(V_{\msout{l}})-I(U_{1};V_{\msout{l}})+I(U_{1};Y_{1},W)+\Upsilon_{\msout{l}}
\end{array}\!\!\!\right\}
\\\label{Eqn:3CQBCRateRegion5}
R_{1}+R_{\msout{l}}&<&\min\left\{\!\!\! \begin{array}{c}
H(V_{\msout{l}})-I(U_{1};V_{\msout{l}})+\gamma+I(U_{1},W;Y_{1})+I(U_{1};W)-H(W)+\Upsilon_{\msout{l}},\\
H(V_{\msout{l}})+\gamma_{12}+I(U_{1},W;Y_{1})+I(U_{1};W)-H(W) +\Upsilon_{{l}}
\end{array}\!\!\!\right\}
\\\label{Eqn:3CQBCRateRegion6}
R_{1}+R_{l}&<&\min \left\{\!\!\! \begin{array}{c}
                                  H(V_{l}|U_{1})+I(U_{1},W;Y_{1})+I(U_{1};W)-H(W), H(V_{l}|U_{1})+I(U_{1};Y_{1},W)+\Upsilon_{l}
                                 \end{array}
\!\!\!\right\}
\\\label{Eqn:3CQBCRateRegion7}
R_{1}+R_{2}+R_{3}&<&\min \left\{\!\!\! \begin{array}{c}I(U_{1},W;Y_{1})+I(U_{1};W)-H(W)+H(V_{2},V_{3}|U_{1})+\min\{\Upsilon_{2},\Upsilon_{3}\},\\ H(V_{2},V_{3}|U_{1})+I(U_{1};Y_{1},W)+\Upsilon_{2}+\Upsilon_{3},\\I(U_{1},W;Y_{1})+I(U_{1};W)-H(W) +H(V_{2},V_{3}|U_{1})+\gamma_{12} + 2\Upsilon_{l},\\H(V_{3})+H(V_{2}|U_{1})+\Upsilon_{2}+I(U_{1},W;Y_{1})+I(U_{1};W)-H(W),\\H(V_{2})+H(V_{3}|U_{1})+\Upsilon_{3}+I(U_{1},W;Y_{1})+I(U_{1};W)-H(W) \end{array}\!\!\!\right\}
\\\label{Eqn:3CQBCRateRegion8}
R_{1}+R_{2}+2R_{3}&<& 2H(V_{3})+H(V_{2})-I(U_{1};V_{3})-I(V_{2};V_{3})+I(U_{1},W;Y_{1})+I(U_{1};W)-H(W)+2\Upsilon_{3}
\\\label{Eqn:3CQBCRateRegion9}
R_{1}+2R_{2}+R_{3}&<&2H(V_{2})+H(V_{3})-I(U_{1};V_{2})+I(U_{1},W;Y_{1})+I(U_{1};W)-H(W)+2\Upsilon_{2}\\
\label{Eqn:3CQBCRateRegion10}
2R_{1}+R_{2}+R_{3}&<& 2I(U_{1},W;Y_{1})+2I(U_{1};W)-2H(W)-I(V_{2},V_{3};U_{1})-I(V_{2};V_{3})+H(V_{2})+H(V_{3})
\end{eqnarray}
where $\Upsilon_{j} \define \min\{ I(V_{j};Y_{j})-H(V_{j}),I(W;Y_{1},U_{1})-H(W) \}$, $\gamma_{12} \define \min_{\theta \in \fieldq\setminus\{0\}}H(V_{2}\oplus_{q}\theta V_{3}|U_{1})$, $\gamma \define \min_{\theta \in \fieldq\setminus\{0\}}H(V_{2}\oplus_{q}\theta V_{3})$, $\sum_{x \in \CalX}p_{X}(x)\kappa(x) \leq \tau$, where all the above information quantities are computed wrt state,
\begin{eqnarray}
\label{Eqn:CorollaryState}
\sigma^{\ulineY XU_{1}V_{2}V_{3}W} = \sum_{\substack{x,u_{1},v_{2},v_{3},w}}p_{XU_{1}V_{2}V_{3}W}(x,u_{1},v_{1},v_{2},w)\rho_{x}\otimes  \ketbra{x~u_{1}~v_{2}~v_{3}~w}\mbox{ with}\\
p_{XU_{1}V_{2}V_{3}W}\left( {x,u_{1}, v_{2},v_{3},w} \right)=p_{XU_{1}V_{2}V_{3}}\left({x,u_{1}, v_{2},v_{3}}\right)\mathds{1}_{\{\substack{ w=  v_{2}\oplus_{q} v_{3}}\}}\nonumber
\end{eqnarray}
$\forall (x,u_{1},v_{2},v_{3},w) \in \CalX\times\CalU_{1}\times \CalV_{2}\times \CalV_{3}\times \CalW$.
\end{corollary}
\begin{proof}
 Follows by eliminating variables $B_{1},S_{2},S_{3}$ in the characterization of the inner bound derived in Thm.~\ref{Thm:3CQBCCodingTheorem1} using the technique proposed in \cite{201108ISIT_ChaEmaZamAre}
\end{proof}

\subsection{Decoding Sum of Public Codewords}
\label{SubSec:DecSumOfPublicCdwrds}
The coding scheme presented in Sec.~\ref{SubSec:Rx1DecodesSumOfPvtCodebooks} does not exploit the technique of message splitting. Recall that both in the Marton's coding scheme and the Han-Kobayashi coding scheme, each Rx faciltates the other Rx to decode a part of its message. This ensures that each Rx decode only that component of the other user's signal that is interefering. Message splitting also ensures that the primary user's code is not rate limited. In this section, we split Rx $2$ and $3$'s transmission into two parts each - $U_{j}$ and $V_{j}$. Rx $1$ as before decodes $U_{1},V_{2}\oplus_{q}V_{3}$. For $j=2,3$, Rx $j$ decodes both $U_{j},V_{j}$ codebooks.

\begin{theorem}
\label{Thm:3CQBCCodingTheorem2}
  A rate-cost quadruple $(\ulineR,\tau)$ is achievable if there exists finite sets $\CalU_{1},\CalU_{2},\CalU_{3}$, a finite field $\CalV_{2}=\CalV_{3}=\CalW=\fieldq$ of size $q$, real numbers $S_{j}\geq 0,T_{j}\geq 0:j=2,3$, $B_{k},L_{k} : k \in [3]$ satisfying $R_{1}=L_{1}$, $R_{j}=L_{j}+T_{j}\log q$ for $j=2,3$ and a PMF $p_{XU_{1}V_{2}V_{3}}$ on $\CalX\times\CalU_{1}\times\CalV_{2}\times\CalV_{3}$ wrt which
 \begin{eqnarray}
 \label{Eqn:Step2SrcCodingBound}
(S_{A}-T_{A})\log q + B_{D} &>&\! \sum_{d \in D}\!\!H(U_{d})+\sum_{a\in A}\!\!H(V_{a})-H(V_{A},U_{D})+|A|\log q -\sum_{a\in A}\!\!H(V_{a})
\\
\begin{array}{c}
\max\left\{S_{2}\log q + B_{D}, S_{3}\log q + B_{D} \right\}
  \end{array}&>& \log q  +\sum_{d \in D}H(U_{d}) \label{Eqn:Step2SrcCodingBoundAlg}- \min_{\theta \in 
\fieldq\setminus\{0\}}H(V_{2}\oplus \theta V_{3},U_{D}),
\\
\label{Eqn:Step2ChnlCodingBnd1}
R_{1}+B_{1} &<& I(Y_{1},V_{2}\oplus_{q}V_{3};U_{1}), \\
\label{Eqn:Step2ChnlCodingBnd2}
\begin{array}{c}
\max\left\{S_{2}\log q, S_{3}\log q  \right\}
  \end{array}
  &<&   \left\{ 
   I(Y_{1},U_{1};V_{2}\oplus_{q}V_{3})+\log q - H(V_{2}\oplus_{q} V_{3})
   \right\},\\
   \label{Eqn:Step2ChnlCodingBnd3}
  \max\left\{
R_{1}+B_{1}+S_{2}\log q,R_{1}+B_{1}+S_{3}\log q \right\}
  &<&  I(Y_{1};V_{2}\oplus_{q}V_{3},U_{1})+I(U_{1};V_{2}\oplus_{q} V_{3})+\log q - H(V_{2}\oplus_{q} V_{3})\\
  \label{Eqn:Step2ChnlCodingBnd4}
  S_{j}\log q &<& I(Y_{j},U_{j};V_{j})+\log q - H(V_{j})\\
  \label{Eqn:Step2ChnlCodingBnd5}
  L_{j}+B_{j} &<& I(Y_{j},V_{j};U_{j})\\
  \label{Eqn:Step2ChnlCodingBnd6}
  S_{j}\log q+L_{j}+B_{j} &<& I(Y_{j};U_{j},V_{j})+I(U_{j};V_{j})+\log q - H(V_{j})
 \end{eqnarray}
 for all $A \subseteq\{2,3\}$, $D \subseteq\{1\}$, $j=2,3$, $\sum_{x \in \CalX}p_{X}(x)\kappa(x) \leq \tau$, where all the above information quantities are computed wrt state,
\begin{eqnarray}
\label{Eqn:Step2TheoremState}
\sigma^{\ulineY X\ulineU V_{2}V_{3}W} = \sum_{\substack{x,\ulineu,v_{2},v_{3},w}}p_{X\ulineU V_{2}V_{3}W}(x,\ulineu,v_{1},v_{2},w)\rho_{x}\otimes  \ketbra{x~u_{1}~u_{2}~u_{3}~v_{2}~v_{3}~w}\mbox{ with}\\
p_{X\ulineU V_{2}V_{3}W}\left( {x,u_{1},u_{2},u_{3},v_{2},v_{3},w} \right)=p_{X\ulineU V_{2}V_{3}W}\left({x,\ulineu,v_{2},v_{3}}\right)\mathds{1}_{\{\substack{ w= v_{2}\oplus_{q} v_{3}}\}}
\end{eqnarray}
$\forall (x,\ulineu,v_{2},v_{3},w) \in \CalX\times\ulineCalU\times \CalV_{2}\times \CalV_{3}\times \CalW$.
 \end{theorem}

In the previous section, we introduced an inner bound for the 3-user CQBC, where only Rx 1 decodes a bivariate function of the interference, while Rx 2 and Rx 3 decodes only their respective codewords. To analyze the decoding error at Rx 1, we employed Fawzi et.~al.’s simultaneous decoding technique \cite{201206TIT_FawHaySavSenWil}. This technique devised by Fawzi et.~al.,  can be used to decode at most two codewords simultaneously. In the next section, we characterize our second and third inner bounds. In deriving these new inner bounds, we leverage Sen's \cite{202103SAD_Sen} recent technique of tilting smoothing and augmentation to design and analyse a simultaneous decoding POVM. To emphasize, the coding strategy devised in the next section enables all three receivers to decode a bivariate function of the interference, peel the same off and decode their respective messages.

\section{Simultaneous Decoding of Coset Codes}
\label{Sec:SimultDecOfCosetCds}

Through Exs.~\ref{Ex:GenAdd3CQBCExample}, \ref{Ex:RotatedExampleWithAdd}, it is evident that decoding is in general more efficient than Tx precoding and (ii) interference on a $3-$CQBC is in general a bivariate function of the interferer signals, Exs.~~\ref{Ex:GenAdd3CQBCExample}, \ref{Ex:RotatedExampleWithAdd} illustrate that by employing jointly designed coset codes and Rxs decoding the bivariate components directly, interference is more efficiently managed. A general coding scheme must facilitate each Rx to efficiently decode \textit{both} univariate and bivariate components of its two interfering signals. Unstructured IID codes and jointly coset codes are efficient at decoding the former and latter components respectively. A general coding strategy must therefore incorporate conventional unstructured IID and coset codebooks to enable each Rx decode both univatiate and bivariate functions of the other Rx's signals.

We are led to the following general coding strategy. See Fig.~\ref{Fig:3CQBCMapOfRVs}. Each Rx $j$'s message is split into six parts $m_{j}=(m_{j}^{W},m_{ji}^{Q},m_{jk}^{Q},m_{ji}^{U},m_{jk}^{U},m_{j}^{V})$ where $i,j,k \in [3]$ are distinct indices i.e., $\{ i,j,k\}=[3]$. $m_{j}^{W},m_{ji}^{Q},m_{jk}^{Q}$ and $m_{j}^{V}$ are encoded using unstructured ID codes built on $\CalW,\CalQ_{ji},\CalQ_{jk},\CalV_{j}$ respectively. $m_{ji}^{U}$ and $m_{jk}^{U}$ are encoded via coset codes built on $\CalU_{ji}=\CalF_{\upsilon_{i}},\CalU_{jk}=\CalF_{\upsilon_{k}}$ respectively. In addition to decoding all these six parts, Rx $j$ also decodes $m_{ij}^{T},m_{kj}^{T}$ and the sum $u_{ij}^{n}(m_{ij}^{U}) \oplus u_{kj}^{n}(m_{kj}^{U})$ of the codewords corresponding to $m_{ij}^{U}$ and  $m_{kj}^{U}$ respectively. We let (i) $\ulineT_{j*}=(T_{ji},T_{jk})$, (ii) $\ulineT_{*j}=(T_{ij},T_{kj})$ and (iii) $U_{j}^{\oplus} = U_{ij}\oplus U_{kj}$ denote (i) RVs encoding semi-private parts $m_{ji}^{T},m_{jk}^{T}$ of Rx $j$, (ii) semi-private parts of Rx $i$,$k$  decoded by Rx $j$ and (iii) bivariate component decoded by Rx $j$ respectively. We present the analysis of general strategy in two steps. In contrast to Sec.~\ref{Sec:AchRegionI} where we designed coding strategies using PCCs, we will use Nested Coset Codes (NCCs) in this section. We define the same below.
\begin{definition}
\label{Defn:NestedCosetCode}
An $(n,k,l,g_{I},g_{o/I}, b^n)$ nested coset code (NCC) over $\CalF_{\Prime}$  consists of (i) generators matrices $ g_{I} \in \CalF_{\Prime}^{k\times n}$ and $ g_{o/I} \in \CalF_{\Prime}^{l \times n}$, 
(ii) bias vector $ b^n \in \CalF_{\Prime}^n$, and (iii) an encoder map $e : \CalF_{\Prime}^l \rightarrow \CalF_{\Prime}^k$. We let $u^n(a,m)= a g_{I}\oplus m g_{o/I} \oplus b^n$, for $(a,m) \in \CalF_{\Prime}^k \times \CalF_{\Prime}^l$ denote elements in its range space.
\end{definition}
\begin{figure}[H]
 \centering
\includegraphics[width=5in]{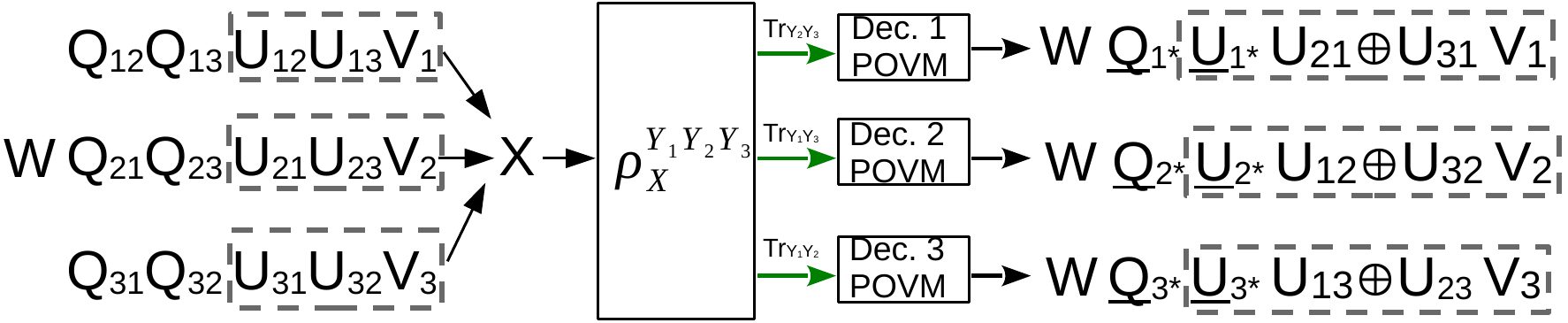}
    \caption{Depiction of all RVs in the full blown coding strategy. In Sec.~\ref{SubSec:StepICodingTheorem} (Step I) only RVs in the gre dashed box are non-trivial, with the rest trivial.}
    \label{Fig:3CQBCMapOfRVs}
\end{figure}

\med\textbf{Notation}: For prime $\Prime \in \integers$, $\CalF_{\Prime}$ will denote finite field of size $\Prime \in \integers$ with $\oplus$ denoting field addition in $\CalF_{\Prime}$ (i.e. mod$-\Prime$). \textcolor{black}{In any context, when used in conjunction }$i,j,k$ will denote distinct indices in $[3]$, hence $\{i,j,k\}=[3]$. $\dulineU \define (\ulineU_{1*},\ulineU_{2*},\ulineU_{3*})=(\ulineU_{*1},\ulineU_{*2},\ulineU_{*3})$, $\dulineT \define (\ulineT_{1*},\ulineT_{2*},\ulineT_{3*})$
\begin{figure}[H]
   \centering
\includegraphics[width=3in]{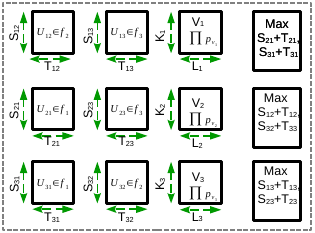}
   \vspace{-0.1in}
    \caption{The $9$ codebooks on the left are used by Tx. The $3$ rightmost cosets are obtained by adding corresp.~cosets. Rx $k$ decodes into $4$ codes in row $k$.}
    \label{Fig:3CQBCStep1SimultDecoding}
        \vspace{-0.1in}
\end{figure}
\subsection{Step I : Simultaneous decoding of Bivariate Components}
\label{SubSec:StepICodingTheorem}
In this first step, we activate only the `bivariate' $\dulineU \define (\ulineU_{1*},\ulineU_{2*},\ulineU_{3*})$ and private parts $\ulineV=V_{1},V_{2},V_{3}$ (Fig.~\ref{Fig:3CQBCMapOfRVs}) are non-trivial, with the rest $\ulineW=\phi,\ulineQ = (\ulineQ_{k*}: k \in [3])=\phi$ trivial (Fig.~\ref{Fig:3CQBCMapOfRVs}). $\ulineV$ being unstructured IID and $\dulineU$ coded via coset codes, Thm.~\ref{Thm:3CQBCStepIInnerBound} possesses all non-trivial elements and meets all objectives.
\begin{theorem}
\label{Thm:3CQBCStepIInnerBound}
Let $\hat{\alpha}_{S} \in [0,\infty)^{4}$ be the set of all rate-cost quadruples $(R_{1},R_{2},R_{3},\tau) \in [0,\infty)^{4}$ for which there exists (i) finite sets $\CalV_{j}: j \in [3]$, (ii) finite fields $\SemiPrivateRVSet_{ij}=\CalF_{\upsilon_{j}}$ for each $ij \in \Xi\define \{12,13,21,23,31,32\}$, (iii) a PMF $p_{\dulineU~\!\!\ulineV X}=p_{\ulineU_{1*}\ulineU_{2*}\ulineU_{3*}V_{1}V_{2}V_{3} X}$ on $\dulineCalU \times \ulineCalV\times\CalX$, (iv) nonnegative numbers $S_{ij},T_{ij}:ij \in \left\{12,13,21,23,31,32 \right\}, K_{j},L_{j}:j\in \left\{ 1,2,3\right\}$, such that
$R_{1}=T_{12}\log \upsilon_{2}+T_{13}\log \upsilon_{3}+L_{1},
R_{2}=T_{21}\log \upsilon_{1}+T_{23}\log \upsilon_{3}+L_{2},
R_{3}=T_{31}\log _{1}+T_{32}\log \upsilon_{2}+L_{3}$, $\Expectation\left\{ \kappa(X) \right\} \leq \tau$,
\begin{eqnarray}
\label{Eqn:ManyToManySourceCodingBounds}
&& S_{A}+M_{B}+K_{C} \stackrel{\boldsymbol{\alpha}}{>}\Theta(A,B,C),\mbox{ where } 
\end{eqnarray}
\begin{eqnarray}
\Theta(A,B,C)\define \!\! \max_{\substack{(\theta_{j}:j \in B) \\\in \underset{j \in B}{\prod} \fieldpij}}  \left\{
\!\!\! \begin{array}{c}\sum_{a \in A}\log |\mathcal{U}_{a}| + \sum_{j \in B}\log \upsilon_{j} +\sum_{c \in C} H(V_{c}) - H(U_{A},U_{ji}\oplus \theta_{j}U_{jk}:j \in B,V_{C})\end{array}\!\!\!
 \right\}
\nonumber
\end{eqnarray}
for all $A \subseteq \left\{12,13,21,23,31,32\right\}, B \subseteq \left\{ 1,2,3 \right\}, C \subseteq \left\{ 1,2,3 \right\}$, that satisfy $A \cap A(B) = \phi$, where $A({B}) = \cup_{j \in B}\{ ji,jk\}$, $U_{A} = (U_{jk}:jk \in A)$, $V_{C}=(V_{c}:c \in C)$, $S_{A} = \sum _{jk \in A}S_{jk}, M_{B}\define \sum_{j \in B} \max\{ S_{ij}+T_{ij},S_{kj}+T_{kj}\}, K_{C} = \sum_{c \in C}K_{c}$, and
\begin{eqnarray}
\label{Eqn:CQBCChannelCodingStep1Bounds2}
S_{\mathcal{A}_{j}}+T_{\mathcal{A}_{j}} \leq \sum_{a \in \mathcal{A}_{j}}\!\!\log |\mathcal{U}_{a}| - H(U_{\mathcal{A}_{j}}|U_{\mathcal{A}_{j}^{c}},U_{ij}\oplus U_{kj},V_{j},Y_{j})
\\\label{Eqn:CQBCChannelCodingStep1Bounds3}
S_{\mathcal{A}_{j}}+T_{\mathcal{A}_{j}}+S_{ij}+T_{ij} \leq \sum_{a \in \mathcal{A}_{j}}\log |\mathcal{U}_{a}| + \log \upsilon_{j} - H(U_{\mathcal{A}_{j}},U_{ij}\oplus
U_{kj}|U_{\mathcal{A}_{j}^{c}},V_{j},Y_{j}) \\
\label{Eqn:CQBCChannelCodingStep1Bounds4}
S_{\mathcal{A}_{j}}+T_{\mathcal{A}_{j}}+S_{kj}+T_{kj} \leq \sum_{a \in \mathcal{A}_{j}}\log |\mathcal{U}_{a}| + \log \upsilon_{j}
- H(U_{\mathcal{A}_{j}},U_{ij}\oplus
U_{kj}|U_{\mathcal{A}_{j}^{c}},V_{j},Y_{j}) \\
\label{Eqn:CQBCChannelCodingStep1Bounds5}
S_{\mathcal{A}_{j}}+T_{\mathcal{A}_{j}}+K_{j}+L_{j} \leq \sum_{a \in \mathcal{A}_{j}}\log |\mathcal{U}_{a}|+H(V_{j})-H(U_{\mathcal{A}_{j}},V_{j}|U_{\mathcal{A}_{j}^{c}},U_{ij}\oplus
U_{kj},Y_{j}) \\
\label{Eqn:CQBCChannelCodingStep1Bounds6}
S_{\mathcal{A}_{j}}+T_{\mathcal{A}_{j}}+K_{j}+L_{j}+S_{ij}+T_{ij} \leq \sum_{a \in \mathcal{A}_{j}}\log |\mathcal{U}_{a}| + \log \upsilon_{j} +H(V_{j}) -
H(U_{\mathcal{A}_{j}},V_{j},U_{ij}\oplus U_{kj}|U_{\mathcal{A}_{j}^{c}},Y_{j}) \\
\label{Eqn:CQBCChannelCodingStep1Bounds7}
S_{\mathcal{A}_{j}}+T_{\mathcal{A}_{j}}+K_{j}+L_{j}+S_{kj}+T_{kj} \leq \sum_{a \in \mathcal{A}_{j}}\!\!\log |\mathcal{U}_{a}| + \log  \upsilon_{j} +H(V_{j}) - H(U_{\mathcal{A}_{j}},V_{j},U_{ij}\oplus
U_{kj}|U_{\mathcal{A}_{j}^{c}},Y_{j}), 
\end{eqnarray}
for every $\mathcal{A}_{j} \subseteq \left\{ ji,jk\right\}$ with
distinct indices $i,j,k$ in $\left\{ 1,2,3 \right\}$, where
$S_{\mathcal{A}_{j}} \define \sum_{a \in \mathcal{A}_{j}}S_{a},
T_{\mathcal{A}_{j}} \define \sum_{a \in \mathcal{A}_{j}}T_{a},
U_{\mathcal{A}_{j}} = (U_{a}:a \in \mathcal{A}_{j})$ and all the information quantities are evaluated wrt state 
\begin{eqnarray}
 \label{Eqn:StageITestChnl}
 \Psi^{\dulineU\!\!~\ulineU^{\oplus}\!\!~\ulineV X\!\!~\ulineY}\define \!\!\!\!\!\!
 \sum_{\substack{\dulineu, \dulinev,x\\u_{1}^{\oplus},u_{2}^{\oplus},u_{3}^{\oplus} }}\!\!\!\!\!p_{\dulineU\ulineV X}(\ulineu_{1*},\ulineu_{2*},\ulineu_{3*},\ulinev,x)\mathds{1}_{\left\{\substack{u_{ji}\oplus u_{jk}\\=u_{j}^{\oplus}:j\in[3]}\right\}}\!\!\ketbra{\ulineu_{1*}~\! \ulineu_{2*}~\! \ulineu_{3*}~\! u_{1}^{\oplus}~\! u_{2}^{\oplus}~\! u_{3}^{\oplus}~\! \ulinev~\! x} \!\otimes\! \rho_{x}. \nonumber
\end{eqnarray}
Let $\alpha_{S}$ denote the convex closure of $\hat{\alpha}_{S}$. Then $\alpha_{S} \subseteq \ScrC(\tau)$ is an achievable rate region.

\end{theorem}
\begin{remark}
 \label{Rem:Step1CodingThmRemarks}
 Inner bound $\alpha_{S}$ (i) subsumes \cite[Thm.~1]{202202arXiv_Pad3CQBC, 202206ISIT_Pad}, (ii) is the CQ analogue of \cite[Thm.~7]{201603TIT_PadSahPra}, (iii) does \textit{not} include a time-sharing RV and \textit{includes} the `dont's care' inequalities \cite{200702ITA_KobHan,200807TIT_ChoMotGarElg} - \eqref{Eqn:CQBCChannelCodingStep1Bounds3}, \eqref{Eqn:CQBCChannelCodingStep1Bounds4} with $A_{j}=\phi$. Thus, $\alpha_{S}$ can be enlarged.
\end{remark}

\med\textit{Discussion of the bounds}: We \textit{explain} how (i) each bound arises and (ii) their role in driving error probability down. All the source coding bounds are captured through lower bound \eqref{Eqn:ManyToManySourceCodingBounds}$\boldsymbol{\alpha}$ by choosing different $A,B,C$. To understand these bounds, lets begin with simple example. Suppose we have $K$ codebooks $C_{k}=(G_{k}^{n}(m_{k}):m_{k} \in 2^{nR_{k}}):k \in [K]$ with all of them mutually independent and picked IID in the standard fashion with distribution $G_{k}^{n}(m_{k})\sim q_{G_{k}}^{n}$ for $k \in [K]$. What must the bounds satify so that we can find \textit{one} among the $2^{n(R_{1}+\cdots +R_{K})}$ $K-$tuples of codewords to be jointly typical wrt a joint PMF $r_{\ulineG}=r_{G_{1}G_{2}\cdots G_{K}}$. Denoting $q_{\ulineG}= \prod_{k=1}^{K}q_{G_{k}}$, we know from standard info.~theory that if $\sum_{l \in S}R_{l} > D(r_{G_{S}}||\otimes_{l\in S}r_{G_{l}})+D(\otimes_{l\in S}r_{G_{l}}||\otimes_{l\in S}q_{G_{l}})$ for all $S\subseteq [K]$, then\footnote{As is standard, here $r_{G_{S}}\!$ isthe marginal of $r_{\ulineG}$ over the RVs indexed in S.} via the standard second moment method \cite{198101TIT_GamMeu} we can find the desired $K-$tuple. The bounds in \eqref{Eqn:ManyToManySourceCodingBounds}$\boldsymbol{\alpha}$ are just these. We summarize here the detailed explanation provided in \cite{202202arXiv_Pad3CQBC}. We have $9$ codebooks (Fig.~\ref{Fig:3CQBCStep1SimultDecoding}), however we should also consider sum of $U_{ji}$ and $U_{jk}$ codebooks, therefore 12 codebooks, hence $K=12$. $r_{G_{1}\cdots G_{6}}=p_{\dulineU}$ and $q_{U_{ji}}=\frac{1}{|\CalU_{ji}||}$ is uniform. Indeed, owing to pairwise independence and uniform distribution of codewords in the random coset code \cite{BkNIT_PraPadShi}. This explains the the first term within the $\max$ defining $\Theta(\cdots)$. $r_{G_{7}G_{8}G_{9}}=p_{\ulineV}$ and the $V-$codewords have been picked independently $\prod p_{V_{j}}$. This explains the third term in definition of $\Theta(\cdots)$. Owing to the coset structure, it turns out that \cite[Bnds.~(45), (46)]{201804TIT_PadPra} the $U_{ij}$-codeword plus $\theta_{j}$ times $U_{kj}$-codeword must be found in the sum of the $U_{ij}, U_{kj}$ cosets (Fig.~\ref{Fig:3CQBCStep1SimultDecoding} rightmost codes). These vectors being uniformly distributed $\sim \frac{1}{\upsilon_{j}}$, we have the second term and the presence of $U_{ji}\oplus \theta_{j}U_{jk}$ in the entropy term in the defn.~of $\Theta(\cdots)$.\footnote{The curious reader may look at the two bounds in \cite[Bnds.~(72)]{201804TIT_PadPra} with the $\max_{\theta \neq 0}$ which in fact yields \cite[Bnds.~(45), (46)]{201804TIT_PadPra}.} Due to page limit constraints, we are unable to explain the channel coding bounds. However, these are identical to the ones obtained on a $3-$user CQ interference channel ($3$-CQIC). Our companion paper \cite{202503arXiv_GouPad} explains these bounds vividly and further sketches the elements of tilting. The explanation of the channel coding bounds and in fact a full proof of Thm.~\ref{Thm:3CQBCStepIInnerBound} is provided in \cite{202202arXiv_Pad3CQBC}.
\vspace{-0.1in}
\med\textbf{Note} : We have chosen to discuss the source coding bounds in this article since the source coding bounds are different for a broadcast channel as against to an interference channel wherein the source coding bounds are trivial. The channel coding bounds are identical to both and hence provided in the IC paper.
\begin{proof}
Throughout this detailed sketch of proof, $i,j,k$ will represent distinct indices in $[3]$, i.e., $\{i,j,k\} = [3]$ and we let $\dbrackthree \define \{12,13,21,23,31,32\}$. Consider (i) $\SemiPrivateRVSet_{ji} = \CalF_{\Prime_{i}}$ for $ji \in \llbracket 3 \rrbracket$ , (ii) PMF $p_{\dulineU \ulineV X}=p_{\ulineU_{1*}\ulineU_{2*}\ulineU_{3*} V_1 V_2 V_3 X}$ on $\dulineCalU\times\ulineCalV\times\CalX$ as described in the hypothesis. Let $(R_1,R_2,R_3)$ be a rate triple for which there exists non negative numbers $S_{ji},T_{ji},K_{j},L_{j}$ for $ji \in \llbracket 3 \rrbracket$, $j \in [3]$ such that $\Expectation\{\kappa(X)\}\leq \tau$, $R_{j}=T_{ji}+T_{jk}+L_{j}$
and the bounds in (\ref{Eqn:CQBCChannelCodingStep1Bounds2})-(\ref{Eqn:CQBCChannelCodingStep1Bounds7}) hold. We begin our proof by describing the code structure, encoding and decoding.\\

\noindent\textit{Code structure:} Message $m_j$ intended for Rx $j$ splits $\ulinem_j=(m_{ji},m_{jk},m_{jj})$ into three parts - two semi-private parts $m_{ji},m_{jk}$ and one private part $m_{jj}$, see Fig. ~\ref{Fig:FigCodeStructureCQBC}. The semi-private parts $m_{ji},m_{jk}$ are encoded using coset codes built over $\CalU_{ji},\CalU_{jk}$ respectively, and $m_{jj}$ is encoded using a conventional IID random code $\CalC_j$ built over $\CalV_{j}$. RVs $\ulineU_{j*}\define (U_{ji},U_{jk})$ represent information of Tx $j$ encoded in the semi-private parts and RV $V_{j}$ represents the information of Tx $j$ encoded in the private part $m_{jj}$. Note $\ulineU_{j*},V_{j} \in \CalF_{\Prime_{i}} \times \CalF_{\Prime_{k}} \times \CalV_j$, where $\CalV_{j}$ is an arbitrary finite set and $\CalF_{\Prime_{i}}$ is the finite field of size $\Prime_{i}$. Specifically, let $\CalC_{j} =\{v_{j}^{n}(b_{j},m_{jj}) : 1 \leq b_{j} \leq 2^{nK_{j}}, 1 \leq m_{jj}\leq 2^{nL_{j}} \}$ be built on $\CalV_{j}$ with the marginal $p_{V_{j}}^{n}$. For $i \in [3] \setminus \{j\}$, the message $m_{ji}$ is encoded via NCC denoted as $(n,r_{ji},t_{ji},g^{ji}_{I},g^{ji}_{o/I}, b_{ji}^n)$ over $\CalF_{\Prime_{i}}$.
\begin{figure}[H]
\centering
\includegraphics[width=4in]{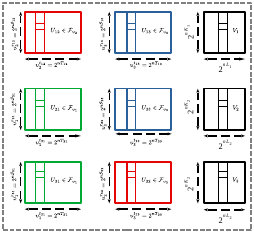}
   \vspace{-0.1in}
    \caption{$9$ codes employed in the proof of Thm.~\ref{Thm:3CQBCStepIInnerBound}. The $U_{ji} : ji \in \dbrackthree$ are coset codes built over finite fields. Codes with the same color are built over the same finite field, and the smaller of the two is a sub-coset of the larger. The black codes are built over auxiliary finite sets $\CalV_{j}$ using the conventional IID random code structure. Row $j$ depicts the codes of Tx, Rx $j$. Rx $j$, in addition to decoding into codes depicted in row $j$ also decodes $U_{ij}\oplus U_{kj}$.}
    \label{Fig:FigCodeStructureCQBC}
        \vspace{-0.1in}
\end{figure} 
Let $c(m_{ji}^{t_{ji}})$ be the coset corresponding to message $m_{ji}^{t_{ji}}$ defined as 
\[
c(m_{ji}^{t_{ji}})= \{ u_{ji}^n(a^{r_{ji}},m_{ji}^{t_{ji}}) : a^{r_{ji}} \in  \CalF_{\Prime_{i}}^{r_{ji}} \} \]
where,  $ u_{ji}^n(a^{r_{ji}},m_{ji}^{t_{ji}})= a^{r_{ji}} g^{ji}_{I}+ m_{ji}^{t_{ji}} g^{ji}_{o/I} + b_{ji}^n$
denote the generic code in the coset $c(m_{ji}^{t_{ji}})$. Henceforth, denoting $\underaccent{\tilde}{m} \define(\ulinem_{1},\ulinem_{2},\ulinem_{3})$ , let
\begin{eqnarray}
    \CalD(\underaccent{\tilde}{m})\!\!=\!\left\{ \left( \!\!\!\begin{array}{c}( u_{ji}^{n} (a^{r_{ji}},m_{ji}^{t_{ji}})\! :\! ji \in \dbrackthree ),\\(v_{j}^{n}(b_{j},m_{jj})\! :\! j \in [3])\end{array}\!\!\!\right) \!:\!  \left( \!\!\!\begin{array}{c}( u_{ji}^{n} (a^{r_{ji}},m_{ji}^{t_{ji}})\! :\! ji \in \dbrackthree ),\\(v_{j}^{n}(b_{j},m_{jj})\! :\! j \in [3])\end{array}\!\!\!\right)  \in T_{\delta}^n(p_{\dulineU \ulineV}) \right\} \nonumber 
\end{eqnarray}
be the list of the codewords that are jointly typical according to $p_{\dulineU \ulineV}$.\\ 
\noindent \textit{Encoding Rule:} Note that every element of $\CalD(\underaccent{\tilde}{m})$ is a collection of $9$ codewords chosen from the respective $9$ codebooks. For every message $\underaccent{\tilde}{m}=(\ulinem_1,\ulinem_2,\ulinem_3)$, the encoder selects one of the possible 
$9-$codeword collection from $\CalD(\underaccent{\tilde}{m})$. This selected collection of codewords, denoted $\left(( u_{ji}^{n}(a_{ji}^{*},m_{ji}^{t_{ji}}): ji \in \dbrackthree ),(v_{j}^{n}(b_{j}^{*},m_{jj}) : j \in [3])\right) $ is referred to as the \textit{chosen collection} of codewords. If $\CalD(\underaccent{\tilde}{m})$ is empty, a default collection $\left(( u_{ji}^{n}(a_{ji}^{*},m_{ji}^{t_{ji}}): ji \in \dbrackthree ),(v_{j}^{n}(b_{j}^{*},m_{jj}) : j \in [3])\right) $ from the codebooks is used instead and referred to as the \textit{chosen collection}. The symbols input on the channel is then obtained by evaluating the function $f$ letterwise on the \textit{chosen collection}
 \[
 x^n(\underaccent{\tilde}{m}) := f^n\left(\begin{array}{c}(  u_{ji}^{n}(a_{ji}^{*},m_{ji}^{t_{ji}}): ji \in \dbrackthree  ),(v_{j}^{n}(b_{j}^{*},m_{jj}) : j \in [3])\end{array} \right).
 \] 
The analysis of this encoder error event is identical to the encoder error event analyzed in \cite{201804TIT_PadPra}. Specifically, in \cite[App.~A, Proof of Lemma 8]{201804TIT_PadPra} we employ a standard second moment method \cite{BkNIT_PraPadShi}, \cite[App.~A]{201710TIT_PadPra} \cite[Lemma 8.1 in App.~8A]{BkNITElGamalKim_2011} \cite[App.~E]{2024MMTIT_Pad} to upper bound the error event therein. An identical series of steps, i.e. employing the same second moment method, it is straight forward to derive an exponentially decaying bound on the probability of $\CalD(\underaccent{\tilde}{m})$ being empty if
\begin{eqnarray}
\label{Eqn:ManyToManySourceCodingBounds}
&&
S_{A}+M_{B}+K_{C} \stackrel{\boldsymbol{\alpha}}{>}\Theta(A,B,C) \nonumber \mbox{ where } \nonumber
\\
&& \Theta(A,B,C)\define \!\!\!\max_{\!\!\!\substack{(\theta_{j}:j \in B) \in \underset{j \in B}{\prod} \fieldpij}}  \left\{ 
 \begin{array}{c} \!\!\!\!\sum_{a \in A}\!\log |\mathcal{U}_{a}|\!\! +\!\! \sum_{j \in B}\!\log \upsilon_{j}\!\! +\!\!\sum_{c \in C} H(V_{c})\!\! -\!\! H(U_{A},U_{ji}\!\oplus \!\theta_{j}U_{jk}:j \!\in \! B,V_{C})\!\!\!\end{array}
\! \right\}
\nonumber
\end{eqnarray}
We now proceed to the decoding rule.\\
\noindent \textit{Decoding POVM:} Rx \( j \) decodes a bivariate component \( U_{ij}(m_{ij}) \oplus U_{kj}(m_{kj}) \) in addition to the codewords \( U_{ji}(m_{ji}), U_{jk}(m_{jk}), V_j(m_{jj}) \) corresponding to its message. Specifically, Rx \( j \) decodes the quartet  
\[
(U_{ji}(m_{ji}), U_{jk}(m_{jk}), V_j(m_{jj}), U_{ij}(m_{ij}) \oplus U_{kj}(m_{kj}))
\]  
using a simultaneous decoding POVM. In particular, the construction of our simultaneous decoding POVM adopts the \textit{tilting, smoothing, and augmentation} approach of Sen \cite{202103SAD_Sen}. Towards describing this, we begin by characterizing the effective CQ MAC state, 

\begin{eqnarray}
\label{Eqn:EffectCQMACOrigNotation} 
\xi^{\ulineU_{j*}V_{j}U_{j}^{\oplus}Y_{j}} &=& \sum_{\ulineu_{j*},v_{j},u_{j}^{\oplus}} p_{\ulineU_{j*}V_{j}}(\ulineu_{j*},v_{j}) p_{U_{j}^{\oplus}}(u_{j}^{\oplus}) \ketbra{\ulineu_{j*}~v_{j}~u_{j}^{\oplus}} \otimes \xi_{\ulineu_{j*}v_{j}u_{j}^{\oplus}} 
\label{decoding}.
\end{eqnarray}
Rx \( j \)'s error analysis is essentially the error analysis of the above MAC, where the Rx attempts to decode \( U_{ji}, U_{jk}, V_{j} \), and \( U_{ij} \oplus U_{kj} \). We may, therefore, restrict our attention to the effective five-transmitter MAC specified via \eqref{Eqn:EffectCQMACOrigNotation}. To reduce clutter and simplify notation, we rename $Z_{j1}\define U_{ji}$
,$Z_{j2} \define U_{jk}$, $Z_{j3} \define V_{j}$,$Z_{j4} \define U_{ij}\oplus U_{kj}$
,$ \ulineZ_{j} \define Z_{j1}, Z_{j2}, Z_{j3} , Z_{j4}$, and a PMF
 $p_{\ulineZ_{j}}(\ulinez_{j})= \sum_{u_{ij}} \sum_{u_{kj}} p_{U_{ji} U_{jk} V_j U_{ij} U_{kj}} (z_{j1}, z_{j2},z_{j3}, u_{ij},u_{kj}) \mathds{1} \{z_{j4}= u_{ij} \oplus u_{kj}\}
 $. Furthermore, to simplify notation, we drop the subscript $j$ too. The construction of the decoder POVM and the analysis below must be carried out identically for each of the three decoders $j=1,2,3$. Through the rest of our proof, we restrict attention to analyzing the error probability of the above five transmitter MAC on which the decoder attempts to decode $z_1, z_2, z_3$ and $z_4$.
Towards that end, we relabel the state in (\ref{decoding}) as 
\begin{eqnarray}
 \xi^{\ulineZ Y} \define \sum_{z_{1},z_{2},z_{3},z_{4}} p_{\ulineZ}(\ulinez)\ketbra{z_{1}~z_{2}~z_{3}~z_{4}}\otimes \xi_{z_{1}z_{2}z_{3}z_{4}}. \nonumber 
\end{eqnarray}
\begin{figure}
\centering
\includegraphics[width=5in]{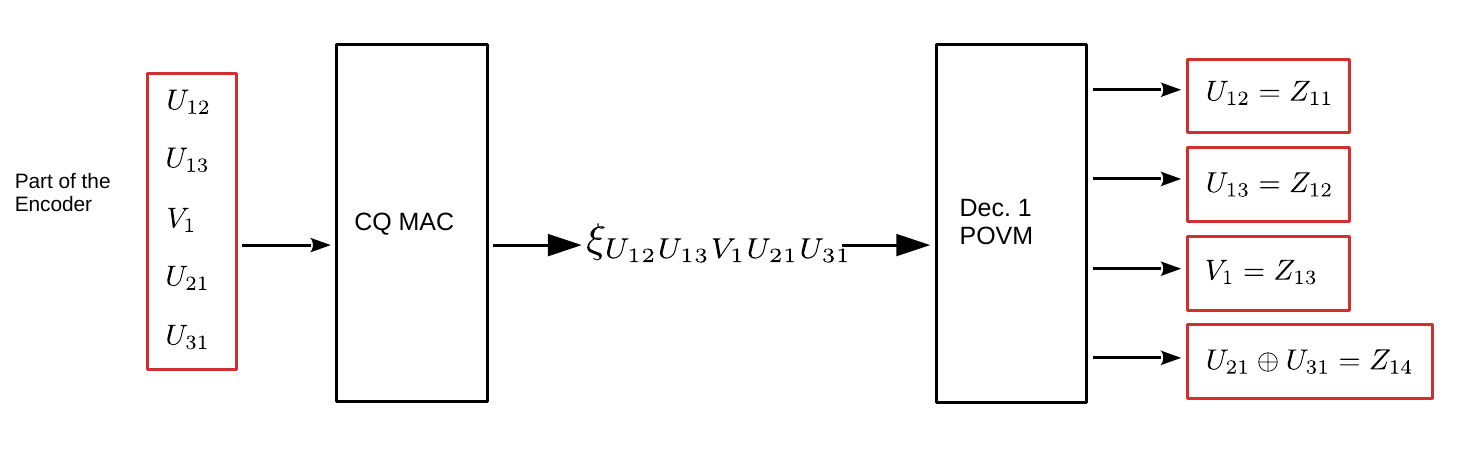}
   \vspace{-0.1in}
    \caption{Communication over the effective CQMAC : The five RVs or codebooks on the left form only a sub-collection of RVs or codes handled by the encoder.}
    \label{Fig:Fig3CQICMAC}
        \vspace{-0.1in}
\end{figure} 
Consider four auxiliary finite sets  $\CalA_{i}$ for $i \in [4]$ along with corresponding four auxiliary Hilbert spaces ${\CalA_{i}}$ of dimension $|\CalA_{i}|$ for each $i \in [4]$.\footnote{$\CalA_{i}$ denotes both the finite set and the corresponding auxiliary Hilbert space. The specific reference will be clear from context.} Define the enlarged space $\CalH_{Y}^{e}$ as follows 
\[
 \CalH_{Y}^{e} \define \CalH_{Y_{G}}
\oplus \bigoplus_{S \subseteq  [4] \setminus [4]} \left(\CalH_{Y_{G}}
 \otimes \CalA_{S} \right)\mbox{ where }\CalH_{Y_{G}} =\CalH_{Y} \otimes \complex^{2}, \CalA_{S} = \bigotimes_{s \in S} \CalA_{s}
\] 
and the symbol $ \oplus $ denotes orthogonal direct sum of space. Note that $\dim(\CalA_{S}) = \prod_{s \in S}\dim(\CalA_{s})$. For $ S \subseteq [4]$, let  $\ket{a_{S}}=\otimes_{s \in S} \ket{a_s}$ be a computational basis vector of  ${\CalA_{S}}$, $Z_{S}=(Z_{s} : s \in S ),$ and $ 0 \leq \eta^n \leq \frac{1}{10}$. In the sequel, we let $\boldsymbol{\CalH_{Y}}= \CalH_{Y}^{\otimes n}$,         $\boldsymbol{\CalH_{Y_G}}=\CalH_{Y_G}^{\otimes n}$,
$\boldsymbol{\CalH_{Y}^{e}}=(\CalH_{Y}^{\otimes n})^{e}$ and $\boldsymbol{\CalA_{S}}= \CalA_{S}^{\otimes n}$. We define the tilting map $\CalT^{S}_{ \ulinea_{S}^{n}, \eta^{n}} : \boldsymbol{\CalH_{Y_G}}\rightarrow
\boldsymbol{\CalH_{Y_G}}  \otimes   \boldsymbol{\CalA_{S}}
$ as 
\begin{eqnarray}
\CalT^{S}_{ \ulinea_{S}^{n}, \eta^{n}}(\ket{h}) = \frac{1}{\sqrt{\Omega(S,\eta)} } (\ket{h} +\eta^{n |S| } \ket{h} \otimes  \ket{a_{S}^{n}})
\nonumber
\end{eqnarray}
where $\Omega(S,\eta) = 1+ \eta^{2n |S|}$. And, $   \CalT_{ \ulinea^{n}, \eta^{n}} : 
\boldsymbol{\CalH_{Y_G}}  \rightarrow \boldsymbol{\CalH_{Y}^{e}}
$ as
\begin{eqnarray}
\CalT_{ \ulinea^{n}, \eta^{n}}(\ket{h}) = \frac{1}{\sqrt{\Omega(\eta)} } (\ket{h} + \sum_{S \subseteq [4]  \setminus [4]} \eta^{n |S| } \ket{h} \otimes \ket{a_{S}^{n}}), \nonumber 
\end{eqnarray}
where, $\Omega(\eta) = 1+16 \eta^{2n}+ 36 \eta^{4n}+ 16 \eta^{6n}$. Let  
 \begin{eqnarray}
    \theta^{ \otimes n}_{\ulinez^n\ulinea^n} = \CalT_{ \ulinea^{n}, \eta^n} \left\{\xi^{\otimes n}_{\ulinez^n}\otimes \ketbra{0} \right\}  \nonumber \label{Eqn:Thetiltedstate}
 \end{eqnarray} 
be the state tilted with components in all appended auxiliary spaces. Here, the map $\CalT_{ \ulinea^{n}, \eta^{n}}$ acts on a mixed state by acting on each pure state in mixture individually. In Appendix \ref{AppSec:ClosenessOfStates}, we show that, 
\begin{eqnarray}
\label{Eqn:ClosenessOfTiltedState}
     \norm{\theta_{\ulinez^{n}, \ulinea^{n}}^{\otimes n} - \xi_{\ulinez^{n}}^{\otimes n} \otimes \ket{0}\bra{0} }_{1} \leq 90 \eta^{n}. 
\end{eqnarray}
Next, we construct the decoding POVM. For partition $S,S^{c}$ of [4] when $S \neq \phi$,  
\begin{eqnarray}
   G^{S}_{\ulinez^{n}}\define \Pi_{z_{S^{c}}^{n}} \Pi_{\ulinez^{n}}\Pi_{z_{S^{c}}^{n}}, \nonumber 
\end{eqnarray}
where $\Pi_{\ulinez^{n}}$ is the conditional typical projector (C-Typ-Proj) of $\otimes_{t=1}^{n}\xi_{z_{1t}z_{2t}z_{3t}z_{4t}}$ 
and $\Pi_{z_{S}^{n} }$ is C-Typ-Proj of  $\otimes_{t=1}^{n}\xi_{z_{st}:s \in S}$.
By Gelfand-Naimark’s thm.~\cite[Thm.~3.7]{BkHolevo_2019}, there exists a projector  $\olineG^{S}_{\ulinez^{n}}\in \CalP(\boldsymbol{\CalH_{Y_G}} )$ that yields identical measurement statistic as  $G^{S}_{\ulinez^{n}} $on states in $\boldsymbol{\CalH_{Y}} $. Let $\olineB_{\ulinez^{n}}^{S}=I_{\boldsymbol{\CalH_{Y_G}}}-\olineG^{S}_{\ulinez^{n}}$  be the complement projector. Now, consider
\begin{eqnarray}
 \beta^{S}_{\ulinez^{n},\ulinea^{n}} =  \CalT_{\ulinea^{n}_{S},\eta^n}^{S} ( \olineB_{\ulinez^{n}}^{S})\mbox{, for }\: S \neq [4] \mbox{ and } \beta^{[4]}_{\ulinez^{n},\ulinea^{n}} =  \olineB_{\ulinez^{n}}^{[4]}, \label{Eqn:Tiltingofprojectors}
\end{eqnarray}
where $\beta^{S}_{\ulinez^{n},a_{S}^{n}} $
represent the tilted projector along direction $\ulinea_{S}^{n}$.
Next, we define  $\beta^{*}_{\ulinez^{n},\ulinea^{n}}$ as the projector in $\boldsymbol{\CalH_{Y}^{e}}$ whose support is the union of the supports of $\beta^{S}_{\ulinez^{n},a_{S}^{n}}$ for all $S \subseteq [4]$, $S \neq \phi$.
Let $\Pi_{\boldsymbol{\CalH_{Y_G}}}$, be the orthogonal projector in $\boldsymbol{\CalH_{Y}^{e}}$ onto $\boldsymbol{\CalH_{Y_G}}$ . We define the POVM element $\mu_{\ulinem}$ as
\begin{eqnarray}
\mu_{\ulinem} \define \left(\sum_{\uline{\widehat{m}}} \gamma^{*}_{(\ulinez^{n},\ulinea^{n})(\widehat{\ulinem})}\right)^{-\frac{1}{2}}\!\!\!\!\!\!
\gamma^{*}_{(\ulinez^{n},\ulinea^{n})(\ulinem)} \left(\sum_{\uline{\widehat{m}}} \gamma^{*}_{(\ulinez^{n},\ulinea^{n})(\widehat{\ulinem})}\right)^{-\frac{1}{2}}\!\!\!\!\!\!\!, \mbox{ where }\label{POVM}
\gamma^{*}_{\ulinez^{n},\ulinea^{n}} = (I_{\boldsymbol{\CalH_{Y}^{e}}}-  \beta^{*}_{\ulinez^{n},\ulinea^{n}}) \Pi_{\boldsymbol{\CalH_{Y_G}} }(I_{\boldsymbol{\CalH_{Y}^{e}}}-  \beta^{*}_{\ulinez^{n},\ulinea^{n}}).
\end{eqnarray}
Having stated the decoding POVM, we now return to the error analysis of the CQBC. Our notation is summarized in Table \ref{Table:SummaryNotation} for ease of reference.\\
\begin{table}[H]
\huge
\begin{center}\resizebox{1.09\textwidth}{!}{
\begin{tabular}{|c|c|c|c|}
\hline
{ Symbol} & Description & Symbol & Description \\\hline
 $\underaccent{\tilde}{m}$ &$(\ulinem_{1},\ulinem_{2},\ulinem_{3})$ & $\ulinem_{j}$&$(m_{ji},m_{jk},m_{jj})$\\
 \hline
 $\dulineCalU$& $\ulineCalU_{1*} \times \ulineCalU_{2*} \times \ulineCalU_{3*}$ & $\ulineCalU_{j*}$& $\CalU_{ji} \times \CalU_{jk} $ \\  \hline
 $\ulineCalV$ & $\CalV_1 \times \CalV_2 \times \CalV_3$ &  $U_{j}^{\oplus} $ & $U_{ij}\oplus U_{kj}$ \\
        \hline
        $ g^{ji}_{I} \in \mathcal{F}_{\Prime_{i}}^{r_{ji}\times n}$, 
        $ g^{ji}_{o/I} \in \mathcal{F}_{\Prime_i}^{t_{ji} \times n}$  
        & Generator matrices of user $(i,k) : (i,k) \neq j$  
        & $c(m_{ji}^{t_{ji}})$  
        & $\{ u_{ji}^n(a^{r_{ji}},m_{ji}^{t_{ji}}) : a^{r_{ji}} \in  \mathcal{F}_{\Prime_i}^{r_{ji}} \}$ \\ 
        \hline
        $ u_{ji}^n(a^{r_{ji}},m_{ji}^{t_{ji}})= a^{r_{ji}} g^{ji}_{I}+ m_{ji}^{t_{ji}} g^{ji}_{o/I} + b_{ji}^n$  
        & A generic codeword in bin/coset indexed by message $m_{ji}^{t_{ji}}$
        & $ \CalD(\underaccent{\tilde}{m})$  
        & Codeword list jointly typical wrt $p_{\dulineU \ulineV}$ \\ 
        \hline 
$\left( \!\!\!\begin{array}{c}( u_{ji}^{n} (a^{r_{ji}},m_{ji}^{t_{ji}})\! :\! ji \in \dbrackthree ),\\(v_{j}^{n}(b_{j},m_{jj})\! :\! j \in [3])\end{array}\!\!\!\right)$
        & Codewords chosen from $ \CalD(\underaccent{\tilde}{m})$ for message $\underaccent{\tilde}{m}$
        &        $f : \dulineCalU\times \uline{\CalV} \rightarrow \CalX$  
        & Mapping function \\ 
        \hline
        $p_{\dulineU \uline{V} X} \triangleq p_{\uline{U}_{1*} \uline{U}_{2*} \uline{U}_{3*} V_1 V_2 V_3 X}$  
        & Chosen test channel  &
         $x^n(\underaccent{\tilde}{m})$  
&$f^n\left(\begin{array}{c}(  u_{ji}^{n}(a_{ji}^{*},m_{ji}^{t_{ji}}): ji \in \dbrackthree  ),(v_{j}^{n}(b_{j}^{*},m_{jj}) : j \in [3])\end{array} \right)$ 
\\ \hline
       $\xi_{\uline{u}_{j*} v_{j} u_{j}^{\oplus}}$  
        & $\sum_{(\Bar{u}_{j*},\Bar{v}_{j}, x)} p_{\Bar{U}_{j*} \Bar{V}_{j} X  | \uline{U}_{j*} V_{j} U_{j}^{\oplus}} (\Bar{u}_{j*},\Bar{v}_{j}, x  | \uline{u}_{j*}, v_{j},u_{j}^{\oplus}) \operatorname{Tr}_{\mathcal{Y}_i \mathcal{Y}_k }(\rho_{x})$
        &$( \mu_{\underaccent{\tilde}{m}} : \underaccent{\tilde}{m} \in \CalM)$  
        & Decoding POVM as defined in (\ref{POVM}) 
\\\hline
        $\CalH_{Y_{G}}$ &$\CalH_{Y} \otimes \complex^{2}$&
 $\CalH_{Y}^{e}$ & $ \CalH_{Y_{G}}
\oplus \bigoplus_{S \subseteq  [4] \setminus [4]} \left( \CalH_{Y_{G}}
 \otimes \CalA_{S}\right) $
\\ \hline
 $\boldsymbol{\CalH_{Y}}$& $\CalH_{Y}^{\otimes n}$ &
         $\boldsymbol{\CalH_{Y_G}}$& $\CalH_{Y_G}^{\otimes n}$
\\  \hline
        $\boldsymbol{\CalH_{Y}^{e}}$& $(\CalH_{Y}^{\otimes n})^{e}$&
        $\boldsymbol{\CalA_{S}}$ &
        $\CalA_{S}^{\otimes n}$ \\ 
        \hline
        $\CalT^{S}_{ \ulinea_{S}^{n}, \eta^{n}}$& Map that tilts  along direction $ \ulinea_{S}^n$ &
        $\CalT_{ \ulinea^{n}, \eta^{n}}$ & Map that tilts into all subspaces\\
        \hline
\end{tabular}}
\end{center}
    \caption{Description of elements that constitute the coding scheme}
     \label{Table:SummaryNotation}
\end{table}
\noindent \textit{Error Analysis:} We now analyze the error probability for the 3-user CQBC. 
First, we begin with a simple union bound to separate the error analysis of the three users. From the Appendix \ref{AppSec:SimpleVerifications}, we have 
\begin{eqnarray}
\label{seperation of the error analysis of the three users}
     I \otimes I \otimes I - \mu_{\ulinem_{1}} \otimes \mu_{\ulinem_{2}} \otimes \mu_{\ulinem_3} \leq  (I-\mu_{\uline{m}_1}) \otimes I \otimes I   + I \otimes (I - \mu_{\uline{m}_2}  ) \otimes I +  I \otimes I \otimes (I - \mu_{\uline{m}_3}  ). \nonumber
\end{eqnarray}
From this, we have 
\begin{eqnarray}
&&P(e^n,\underaccent{\tilde}{\lambda}^{n}) = \frac{1}{\mathcal{M}} \sum_{\underaccent{\tilde}{m} \in \CalM}\tr\left((I-\mu_{\underaccent{\tilde}m} )\rho_{e(\underaccent{\tilde}m)}^{\otimes n} \right) \nonumber \\ 
& = & \frac{1}{\mathcal{M}} \sum_{\underaccent{\tilde}{m} \in \CalM } \tr\left((I-\mu_{\ulinem_1} \otimes \mu_{\ulinem_2} \otimes \mu_{\ulinem_3} )\rho_{e(\underaccent{\tilde}m)}^{\otimes n} \right) \nonumber \\ 
 &\leq& \frac{1}{\mathcal{M}} \sum_{\underaccent{\tilde}{m} \in \CalM}  \tr\left(((I -\mu_{\ulinem_1}) \otimes I \otimes I   + I \otimes (I - \mu_{\ulinem_2} )\otimes I +  I \otimes I \otimes (I - \mu_{\ulinem_3} ))\rho_{e(\underaccent{\tilde}m)}^{\otimes n} \right) \nonumber \\
&=& \! \! \! \! \! \! \frac{1}{\mathcal{M}} \sum_{\underaccent{\tilde}{m} \in \CalM}  \! \left[ \tr\left( ((I -\mu_{\ulinem_1}  ) \otimes \! I \otimes \! I) \rho_{e(\underaccent{\tilde}m)}^{\otimes n} \right) \!  + \! \tr\left((I \otimes \! (I-\mu_{\ulinem_2}) \otimes \! I ) \rho_{e(\underaccent{\tilde}m)}^{\otimes n} \right) \! +\!  \tr\left(( I \otimes \! I \otimes \! (I - \mu_{\ulinem_3})) \rho_{e(\underaccent{\tilde}m)}^{\otimes n} \right) \right] \nonumber \\
&=& \frac{1}{\mathcal{M}} \sum_{\underaccent{\tilde}{m} \in \CalM}  \left[\tr\left( (I - \mu_{\ulinem_1}  ) \tr_{Y_2,Y_3}\{\rho_{e(\underaccent{\tilde}m)}^{\otimes n}\}\right) 
\!\!+\! \tr\left( (I -\mu_{\ulinem_2}) \tr_{Y_1,Y_3}\{\rho_{e(\underaccent{\tilde}m)}^{\otimes n} \}\right) \!\!+\! \tr\left((I - \mu_{\ulinem_3}) \tr_{Y_1,Y_2}\{\rho_{e(\underaccent{\tilde}m)}^{\otimes n}\} \right) \right] . \nonumber
\end{eqnarray}
 The above three terms being identical, henceforth, we analyze just one of them, i.e., the first decoder's error event. To make the notation shorter, we relabel $ \uline{U}_{\bar{1}*} = (\ulineU_{2*},\ulineU_{3*}), V_{\bar{1}} =(V_2,V_3), \uline{u}_{\bar{1}*}^n = (\ulineu_{2*}^n,\ulineu_{3*}^n), v_{\bar{1}}^n =(v_2^n,v_3^n), \uline{m}_{\bar{1}}=(\ulinem_2,\ulinem_3)$. Observe
\begin{eqnarray}
&&\frac{1}{\mathcal{M}} \sum_{\underaccent{\tilde}{m} \in \CalM}  \tr\left( (I - \mu_{\ulinem_1}  ) \tr_{Y_2,Y_3}\{\rho_{e(\underaccent{\tilde}m)}^{\otimes n}\}\right) = \frac{1}{\mathcal{M}} \sum_{\underaccent{\tilde}{m} \in \CalM} 
\sum_{(\uline{u}_{\bar{1}*}^n, v_{\bar{1}}^n, x^n)}
 \tr\left( (I - \mu_{\ulinem_1}  ) \tr_{Y_2,Y_3}\{\rho_{x^n}^{\otimes n}\}\right) \nonumber \\ &&
\mathds{1}_{\Big\{ X^n(\underaccent{\tilde}{m}) = x^n, \uline{U}_{\bar{1}*}^n( \uline{m}_{\bar{1}})=\uline{u}_{\bar{1}*}^n,V_{\bar{1}}^n( \uline{m}_{\bar{1}}) =v_{\bar{1}}^n \Big\}}.\nonumber 
\end{eqnarray}
We evaluate the conditional expectation given $\{
\ulineU_{1*}^n=\ulineu_{1*}^n , V_{1}^n=v_{1}^n , U_{21}^n\oplus U_{31}^n= u_{21}^n \oplus u_{31}^n \}$
\begin{eqnarray}
&&\frac{1}{\mathcal{M}} \sum_{\underaccent{\tilde}{m} \in \CalM}
\sum_{(\uline{u}_{\bar{1}*}^n, v_{\bar{1}}^n, x^n)}
 p^{n}_{X, \uline{U}_{\bar{1}*}^n, V_{\bar{1}}^n | \ulineU_{1*}V_{1} U_{21}\oplus U_{31} }
 \left( x^n,\uline{u}_{\bar{1}*}^n, v_{\bar{1}}^n,|\ulineu_{1*}^n, v_{1}^n,u_{21}^n \oplus u_{31}^n \right)  \tr\left( (I - \mu_{\ulinem_1}  )
\tr_{Y_2,Y_3}\{\rho_{x}^{\otimes n}\}\right) 
\nonumber \\
 &=& \! \! \! \frac{1}{\mathcal{M}} \sum_{\underaccent{\tilde}{m} \in \CalM}  \! \tr\left(\! \!(I - \mu_{\ulinem_1}) \! \left[ \sum_{(\uline{u}_{\bar{1}*}^n, v_{\bar{1}}^n, x^n)} \! \! \! p^{n}_{X, \uline{U}_{\bar{1}*}^n, V_{\bar{1}}^n | \ulineU_{1*}V_{1} U_{21}\oplus U_{31} } \left( x^n,\uline{u}_{\bar{1}*}^n, v_{\bar{1}}^n,|\ulineu_{1*}^n, v_{1}^n,u_{21}^n \oplus u_{31}^n \right)  \tr_{Y_2,Y_3}\{\rho_{x}^{\otimes n}\} \right] \right) 
\nonumber \\
&=&  \! \! \! \frac{1}{\mathcal{M}_{1}} \sum_{\ulinem_{1} \in \mathcal{M}_{1}} \tr((I -\mu_{\ulinem_{1}} ) \xi_{\ulinez_{1}^n(\ulinem_{1})}^{\otimes n}), \nonumber 
\end{eqnarray}
this term corresponds to the probability of error of the CQMAC, which we shall analyze now. Using the Hayashi-Nagaoka inequality \cite{200307TIT_HayNag}, we have
\begin{eqnarray}
   && \frac{1}{\mathcal{M}_{1}} \sum_{\ulinem_{1} \in \mathcal{M}_{1}} \tr((I -\mu_{\ulinem_{1}} ) \xi_{\ulinez_{1}^n(\ulinem_{1})}^{\otimes n})  
    =  \frac{1}{\mathcal{M}_{1}} \sum_{\ulinem_{1} \in \mathcal{M}_{1}} \tr((I -\mu_{\ulinem_{1}} ) \xi_{\ulinez_{1}^n(\ulinem_{1})}^{\otimes n} \otimes \ketbra{0}) 
 \nonumber  \\
 &\leq & \! \frac{1}{\mathcal{M}_{1}} \! \sum_{\ulinem_{1} \in \mathcal{M}_{1}} \! \tr((I -\mu_{\ulinem_{1}} ) \theta_{(\ulinez_{1}^{n}, \ulinea_{1}^{n})(\ulinem_1)}^{\otimes n}) \! +\frac{1}{\mathcal{M}_{1}} \!\sum_{\ulinem_{1} \in \mathcal{M}_{1}} \!
 \norm{\theta_{(\ulinez_{1}^{n}, \ulinea_{1}^{n})(\ulinem_1)}^{\otimes n} - \xi_{\ulinez_{1}^{n}(\ulinem_1)}^{\otimes n} \otimes \ket{0}\bra{0} }_{1} \!.
 \nonumber \\
&\leq &  \frac{2}{\mathcal{M}_{1}} \sum_{\ulinem_{1} \in \mathcal{M}_{1}} 
\tr\left( (I -\gamma^{*}_{(\ulinez_{1}^{n},\ulinea_{1}^{n})(\ulinem_1)} ) 
\theta_{(\ulinez_{1}^{n}, \ulinea_{1}^{n})(\ulinem_1)}^{\otimes n} \right) 
+ \frac{4}{\mathcal{M}_{1}} \sum_{\ulinem_{1} \in \mathcal{M}_{1}} 
\sum_{\uline{\tilde{m}}_1 \neq \ulinem_1} 
\tr\left( \gamma^{*}_{(\ulinez_{1}^{n},\ulinea_{1}^{n})(\uline{\tilde{m}}_1)}  
\theta_{(\ulinez_{1}^{n}, \ulinea_{1}^{n})(\ulinem_1)}^{\otimes n} \right) 
\nonumber \\
&+& \frac{1}{\mathcal{M}_{1}} \sum_{\ulinem_{1} \in \mathcal{M}_{1}} 
\norm{\theta_{(\ulinez_{1}^{n}, \ulinea_{1}^{n})(\ulinem_1)}^{\otimes n} - 
\xi_{\ulinez_{1}^{n}(\ulinem_1)}^{\otimes n} \otimes \ket{0}\bra{0} }_{1} 
\nonumber \end{eqnarray}
 \begin{eqnarray}
&\leq & \! \!\!\! \frac{1}{\mathcal{M}_{1}} \sum_{\ulinem_{1} \in \mathcal{M}_{1}} 
\norm{\theta_{(\ulinez_{1}^{n}, \ulinea_{1}^{n})(\ulinem_1)}^{\otimes n} - 
\xi_{\ulinez_{1}^{n}(\ulinem_1)}^{\otimes n} \otimes \ket{0}\bra{0} }_{1}
+\frac{2}{\mathcal{M}_{1}} \!\!\! \sum_{\ulinem_{1} \in \mathcal{M}_{1}} \!\!\!\tr\left( (I -\gamma^{*}_{(\ulinez_{1}^{n},\ulinea_{1}^{n})(\ulinem_1)} ) 
\theta_{(\ulinez_{1}^{n}, \ulinea_{1}^{n})(\ulinem_1)}^{\otimes n} \right)
\nonumber \\
\label{Eqn:TheFinalThreeTermsForEff5TxMAC}
&+& \!\!\! \!\!\! \frac{4}{\mathcal{M}_{1}} \!\! \sum_{\ulinem_{1} \in \mathcal{M}_{1}}  
\sum_{S \subseteq [4]} \sum_{(\tilde{m}_{1s} : s \in S)} \!\!\!\!\!\!
\tr\left( \gamma^{*}_{(z_{1S}^n,a_{1S}^n)(\tilde{m}_{1S}),(z_{1S^{c}}^n,a_{1S^{c}})(m_{1S^{c}})} 
\theta_{(\ulinez_{1}^{n}, \ulinea_{1}^{n})(\ulinem_1)}^{\otimes n} \right).
\end{eqnarray}
As is standard, we derive an upper bound on the error probability by averaging over an ensemble of codes. Let us now specify the distribution of the random codebook with respect to which we average the error probability. Specifically, each of the generator matrices is picked independently and uniformly from their respective range space. In other words, all the entries of every generator matrices are picked independently and uniformly from their respective finite fields. Furthermore, the bias vector is also picked independently and uniformly from their respective range space. The codewords of the conventional IID random $\CalV_j$ codebook are mutually independent and identically distributed wrt to PMF $ \prod p_{V_j}$. We emphasize, that the coset codes built over $\CalU_{ji}$, $\CalU_{ki}$ intersect. In other words the smaller of these two codes is a sub-coset of the larger, this is accomplished by ensuring that the rows of the generator matrix of the larger code contains all the rows of the generator matrix of the smaller code. For specific details, the reader is referred to the distribution of the random code used to prove \cite[Thm.~7 in Sec.~IV.D]{201804TIT_PadPra} or its simpler version \cite[Thm.~6 in Sec.~IV.C]{201804TIT_PadPra}. Having described the distribution of the random code, we now continue with the analysis of the terms in \eqref{Eqn:TheFinalThreeTermsForEff5TxMAC}.
In regard to the first two terms on the RHS of \eqref{Eqn:TheFinalThreeTermsForEff5TxMAC}, we have
\begin{eqnarray}
\frac{1}{\mathcal{M}_{1}} \sum_{\ulinem_{1} \in \mathcal{M}_{1}} \norm{\theta_{(\ulinez_{1}^{n}, \ulinea_{1}^{n})(\ulinem_1)}^{\otimes n} - \xi_{\ulinez_{1}^{n}(\ulinem_1)}^{\otimes n} \otimes \ket{0}\bra{0} }_{1} &\leq&90 \sqrt{\epsilon} + \epsilon \nonumber \\
\frac{2}{\mathcal{M}_{1}} \sum_{\ulinem_{1} \in \mathcal{M}_{1}} 
\tr\left( (I -\gamma^{*}_{(\ulinez_{1}^{n},\ulinea_{1}^{n})(\ulinem_1)} ) 
\theta_{(\ulinez_{1}^{n}, \ulinea_{1}^{n})(\ulinem_1)}^{\otimes n} \right) 
&\leq& 10808 \sqrt{\epsilon},\nonumber 
\end{eqnarray}
which are proven in Appendix \ref{AppSec:bounds}. We are left only with the third term in the RHS of \eqref{Eqn:TheFinalThreeTermsForEff5TxMAC}. We elaborate now on the analysis of this third term on the RHS of \eqref{Eqn:TheFinalThreeTermsForEff5TxMAC}, and in doing so we leverage the \textit{smoothing} and \textit{augmentation} \cite{202103SAD_Sen} property of the tilting maps we have defined. As is standard, this will be broken into fifteen terms - four singleton errors, six pair errors, four triple errors and one quadruple error. The analysis of the quadruple error is slightly different from the analysis of the rest fourteen terms. To analyze the quadruple error, we first prove the difference between the average of the original state and the average of the tilted state\footnote{When we say `average of $\cdots$' we mean the expectation of the tilted state over the random codebook.} \eqref{Eqn:Thetiltedstate} is small. This is done analogous to the steps used in \cite{202103SAD_Sen} to derive an upper bound on the $\mathbb{L}_{\infty}$-norm of $N_{\delta}$ in \cite{202103SAD_Sen}. Having done this, the rest of the analysis of the quadruple error is straight forward since the corresponding projector has not been tilted and we are left with it operating on the average of the original state.

Having indicated analysis of the quadruple error term, we are left with fourteen terms. The analysis of each of these fourteen terms is essentially identical. To analyze each of the terms, our first step is similar to above. Specifically, we first prove the difference between the partially averaged original state and the partially averaged tilted state\eqref{Eqn:Thetiltedstate} is small in the $\mathbb{L}_{\infty}$-norm. This is done analogous to the steps used in \cite{202103SAD_Sen} to derive an upper bound on the $\mathbb{L}_{\infty}$-norm of $N_{X;l_{x}\delta},N_{Y;l_{y}\delta}$ in \cite{202103SAD_Sen}. Following this, we are left with both the state and the projector being tilted. Extracting out the corresponding isometry, as done in \cite[First column of Pg.~20]{202103SAD_Sen}, we can reduce this term analysis to the standard analysis involving corresponding typicality projectors acting on the original state. Having mentioned the similarities in our analysis, we now highlight the new elements owing to coset codes.
\begin{remark}
The use of coset codes results in codewords being uniformly distributed. In particular, the codewords are \textit{not} automatically typical wrt the desired distribution. This forces us to certain new analysis steps beyond those mentioned above, which maybe viewed as being standard extension of \cite{202103SAD_Sen}. Specifically, the codewords being uniformly distributed results in an enlargement of the bins. This enlargement of the bins will result in a more careful analysis of the above terms beyond the steps mentioned above. These additional steps will be fleshed out in detail in due course in an enlarged version of this manuscript. Essentially, we incorporate the same sequence of steps as employed in \cite[Appendix H, I]{2024MMTIT_Pad} to derive upper bounds on $\zeta_{4}(\ulinem)$, $\zeta_{5}(\ulinem)$ therein. In particular, the use of a `list error event' with $L_{j}>1$ assists us in our pursuit. See \cite[Discussion related to $L_{j}$ found between (10) and (11)]{2024MMTIT_Pad} for more details regarding these sequence of steps.
\end{remark}

Collating all of these steps, an exponentially decaying upper bound on the third term in the RHS of \eqref{Eqn:TheFinalThreeTermsForEff5TxMAC} can be proven if

\begin{eqnarray}
\label{Eqn:3CQICStep1ChnlBnd1}
S_{A_{j}} < \sum_{a \in A_{j}} \log |\mathcal{U}_{a}| - H(U_{A_{j}}|U_{A_{j}^{c}},U_{j}^{\oplus},V_{j},Y_{j}),
\\
S_{A_{j}}+S_{ij} < \sum_{a \in A_{j}} \log |\mathcal{U}_{a}| + \log \Prime_{j} 
\label{Eqn:Bnd2}
- H(U_{A_{j}},U_{j}^{\oplus}|U_{A_{j}^{c}},V_{j},Y_{j}),  \\
\label{Eqn:Bnd3}
S_{A_{j}}+S_{kj} < \sum_{a \in A_{j}}\log |\mathcal{U}_{a}| + \log \Prime_{j} 
- H(U_{A_{j}},U_{j}^{\oplus}|U_{A_{j}^{c}},V_{j},Y_{j}),\\
S_{A_{j}} +K_{j}+L_{j} <  \sum_{a \in A_{j}} \log |\mathcal{U}_{a}|+H(V_{j})
\label{Eqn:Bnd4}
- H(U_{A_{j}},V_{j}|U_{A_{j}^{c}},U_{j}^{\oplus},Y_{j}), \\
\nonumber
S_{A_{j}}+K_{j}+L_{j}+S_{ij} < \sum_{a \in A_{j}} \log |\mathcal{U}_{a}| + \log \Prime_{j}+H(V_{j})
\label{Eqn:Bnd5}
- H(U_{A_{j}},V_{j},U_{j}^{\oplus} |U_{A_{j}^{c}},Y_{j}), \mbox{ and }\\
S_{A_{j}}+K_{j}+L_{j}+S_{kj} < \sum_{a \in A_{j}}\log |\mathcal{U}_{a}| +\log \Prime_{j}+H(V_{j})
\label{Eqn:Bnd6}
- H(U_{A_{j}},V_{j},U_{j}^{\oplus}|U_{A_{j}^{c}},Y_{j}).
\end{eqnarray}
This completes our detailed sketch of proof. Additional steps filling in for some of the standard arguments above will be provided in an enlarged version of this manuscript.
\end{proof}
\subsection{Step II: Enlarging $\alpha_{S}$ via Unstructured IID codes to $\alpha_{US}$}
\label{Sec:Step2}
The inner bound proven in Thm.~\ref{Thm:3CQBCStepIInnerBound} only enables Rxs to decode a bivariate component of the interference. As Marton's common codebook \cite{197905TIT_Mar} strategy proves, there is import to decoding a univariate part of the other user's codewords. In other words, by combining both unstructured codebook based strategy and coset code strategy, we can derive an inner bound that subsumes all current known inner bounds for the three user CQBC. The following inner bound is obtained by combining the conventional unstructured IID code based strategy due to Marton \cite{197905TIT_Mar} and the above coset code based strategy (Thm.~\ref{Thm:3CQBCStepIInnerBound}). \textcolor{black}{A proof of this inner bound is involved, but essentially is a generalization of the proof provided for Thm.~\ref{Thm:3CQBCStepIInnerBound}. We indicate steps in the proof in an enlarged version of this article. In regards to the code structure, refer to Fig.~\ref{Fig:3CQBCStepIICodeStructure}.}
 \begin{figure}[H]
    \centering
    \includegraphics[width=4.5in]{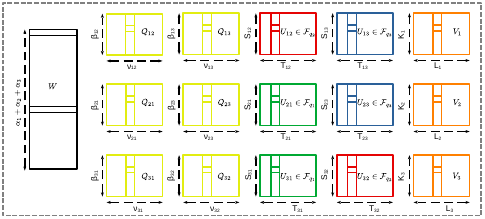}
    \caption{$16$ codes employed in the proof of Thm.~\ref{Thm:Step2}. Row $j$ depicts the codes of Rx $j$. The $U_{ji} : ji \in \dbrackthree$ are coset codes built over finite fields. Codes with the same color are built over the same finite field and the smaller of the two is a sub-coset of the larger. The black, yellow and the orange codes are built over auxiliary finite sets  $\CalW$, $\CalQ_{ji} : ji \in \dbrackthree$ and $ \CalV_{j} : j \in [3] $ respectively using the conventional IID random code structure.}
    \label{Fig:3CQBCStepIICodeStructure}
\end{figure}
 \med\textbf{Notation}: As before, $i,j,k \in [3]$ denote distinct indices, i.e., $\{i,j,k\} = [3]$. We let $\dulineQ \define (\ulineQ_{1*},\ulineQ_{2*},\ulineQ_{3*})=(\ulineQ_{*1},\ulineQ_{*1},\ulineQ_{*1})$, where $\ulineQ_{j*}=(Q_{ji},Q_{jk})$ and similarly $\dulineCalQ=\times_{j=1}^{3}\ulineCalQ_{j*}$, $\dulineq = (\ulineq_{1*},\ulineq_{2*},\ulineq_{3*})$, where $\ulineq_{j*}=(q_{ji},q_{jk})$ and analogously the rest of the objects.
 \begin{theorem}
 \label{Thm:Step2}
Let $\hat{\alpha}_{US} \in [0,\infty)^{4}$ be the set of all rate-cost quadruples $(R_{1},R_{2},R_{3},\tau) \in [0,\infty)^{4}$ for which there exists (i) finite sets $\CalW,\CalQ_{\xi},\CalV_{j}$ for $\xi \in \Xi\define \{12,13,21,23,31,32\}$ and $j \in [3]$, (ii) finite fields $\SemiPrivateRVSet_{ij}=\CalF_{\upsilon_{j}}$ for each $ij \in \Xi$, (ii) a PMF $p_{W\dulineQ\dulineU~\!\!\ulineV X}=p_{W\ulineQ_{1*}\ulineQ_{2*}\ulineQ_{3*}\ulineU_{1*}\ulineU_{2*}\ulineU_{3*}V_{1}V_{2}V_{3} X}$ on $\dulineCalU \times \dulineQ \times \ulineCalV\times\CalX$ (iii) nonnegative numbers $\alpha,\alpha_{j},\beta_{ij},\nu_{ij},S_{ij},T_{ij}:ij \in \Xi, K_{j},L_{j}:j\in \left\{ 1,2,3\right\}$, such that $\alpha > \alpha_{1},\alpha_{2},\alpha_{3}$,
$R_{1}=\alpha_{1}+\nu_{12}+\nu_{13}+T_{12}\log \upsilon_{2}+T_{13}\log \upsilon_{3}+L_{1},
R_{2}=\alpha_{2}+\nu_{21}+\nu_{23}+T_{21}\log \upsilon_{1}+T_{23}\log \upsilon_{3}+L_{2},
R_{3}=\alpha_{3}+\nu_{31}+\nu_{32}+T_{31}\log \upsilon_{1}+T_{32}\log \upsilon_{2}+L_{3}$, $\Expectation\left\{ \kappa(X) \right\} \leq \tau$, $S_{A}+M_{B}+K_{C}+\beta_{D} >\Theta(A,B,C,D)$ where $\Theta(A,B,C,D)\define $
\begin{eqnarray}
 \!\!\!\max_{\substack{(\theta_{j}:j \in B) \in \underset{j \in B}{\prod} \fieldpij}}\!\! \left\{\!\!\!\!
 \begin{array}{c}\sum_{a \in A}\log |\mathcal{U}_{a}| + \sum_{j \in B}\log \upsilon_{j}+\sum_{c \in C} H(V_{c}|W)+\sum_{d \in D}H(Q_{D}|W) - \nonumber \\H(U_{A},U_{ji}\oplus \theta_{j}U_{jk}\!:j \in B,V_{C}|W)\end{array}
 \!\!\!\right\}\!\!\!\!\!
 \nonumber
\end{eqnarray}
for all $A,D \subseteq \left\{12,13,21,23,31,32\right\}, B \subseteq \left\{ 1,2,3 \right\}, C \subseteq \left\{ 1,2,3 \right\}$, that satisfy $A \cap A(B) = \phi$, where $A({B}) = \cup_{j \in B}\{ ji,jk\}$, $U_{A} = (U_{jk}:jk \in A)$, $V_{C}=(V_{c}:c \in C)$, $Q_{D}=(Q_{d}:d \in D)$, $S_{A} = \sum _{jk \in A}S_{jk}, M_{B}\define \sum_{j \in B} \max\{ S_{ij}+T_{ij},S_{kj}+T_{kj}\}, K_{C} = \sum_{c \in C}K_{c}$, $\beta_{D} = \sum _{jk \in D}\beta_{jk}$ and 
\begin{eqnarray}
\label{Eqn:Chnlbnd1}
S_{\mathcal{A}_{j}} + T_{\mathcal{A}_{j}} + \beta_{\CalC_{j}} + \nu_{\CalC_{j}} &\leq& \sum_{a \in \mathcal{A}_{j}} \log |\mathcal{U}_{a}| + H(Q_{\CalC_{j}}|Q_{\CalC_{j}^{C}},W) \nonumber \\ && -H(U_{\mathcal{A}_{j}},Q_{\CalC_{j}}|U_{\mathcal{A}_{j}^{c}},Q_{\CalC_{j}^{C}},U_{ij} \oplus U_{kj}, V_{j}, Y_{j}, W), \nonumber 
\end{eqnarray}
\begin{eqnarray}
\label{Eqn:Chnlbnd2}
S_{\mathcal{A}_{j}} + T_{\mathcal{A}_{j}} + \beta_{\CalC_{j}} + \nu_{\CalC_{j}} + S_{ij} + T_{ij} &\leq& \sum_{a \in \mathcal{A}_{j}} \log |\mathcal{U}_{a}| + \log \upsilon_{j} + H(Q_{\CalC_{j}}|Q_{\CalC_{j}^{C}},W) \nonumber \\ &&- H(U_{\mathcal{A}_{j}}, Q_{\CalC_{j}}, U_{ij} \oplus U_{kj}|U_{\mathcal{A}_{j}^{c}},Q_{\CalC_{j}^{C}},V_{j},Y_{j},W), \nonumber 
\\
\label{Eqn:Chnlbnd4}
S_{\mathcal{A}_{j}}+T_{\mathcal{A}_{j}}+\beta_{\CalC_{j}}+\nu_{\CalC_{j}}+S_{kj}+T_{kj} &\!\leq\!& \sum_{a \in \mathcal{A}_{j}}\!\!\log |\mathcal{U}_{a}| + \log \upsilon_{j}
+H(Q_{\CalC_{j}}|Q_{\CalC_{j}^{C}},W) \nonumber \\
&&- H(U_{\mathcal{A}_{j}}\!,\!Q_{\CalC_{j}},\!U_{ij}\!\oplus\!
U_{kj}|U_{\mathcal{A}_{j}^{c}},Q_{\CalC_{j}^{C}},V_{j},Y_{j},W), \nonumber 
\\
\label{Eqn:Chnlbnd6}
 S_{\mathcal{A}_{j}} +T_{\mathcal{A}_{j}} +\beta_{\CalC_{j}} + \nu_{\CalC_{j}} + K_{j} +\L_{j} &\leq& \sum_{a \in \mathcal{A}_{j}} \log |\mathcal{U}_{a}|+ H(V_{j}|W) + H(Q_{\CalC_{j}}|Q_{\CalC_{j}^{C}},W) \nonumber \\ && -H(U_{\mathcal{A}_{j}}, Q_{\CalC_{j}},\!V_{j}|U_{\mathcal{A}_{j}^{c}},Q_{\CalC_{j}^{C}},U_{ij} \oplus
U_{kj},Y_{j},W), \nonumber 
\\
\label{Eqn:Chnlbnd8}
S_{\mathcal{A}_{j}} + T_{\mathcal{A}_{j}} +\beta_{\CalC_{j}} +\nu_{\CalC_{j}} +K_{j}+L_{j}+S_{ij}+T_{ij} &\leq& \sum_{a \in \mathcal{A}_{j}} \log |\mathcal{U}_{a}| + \log \upsilon_{j} +H(V_{j}|W) +H(Q_{\CalC_{j}}|Q_{\CalC_{j}^{C}},W) \nonumber
\\
\label{Eqn:Chnlbnd9}
&& -H(U_{\mathcal{A}_{j}}, Q_{\CalC_{j}}, V_{j},U_{ij}\oplus U_{kj}|U_{\mathcal{A}_{j}^{c}}, Q_{\CalC_{j}^{C}}, Y_{j},W)
\nonumber
\\
\label{Eqn:Chnlbnd10}
S_{\mathcal{A}_{j}} + T_{\mathcal{A}_{j}} + \beta_{\CalC_{j}} + \nu_{\CalC_{j}} + K_{j}+L_{j}+S_{kj}+T_{kj} &\leq& \sum_{a \in \mathcal{A}_{j}}\log |\mathcal{U}_{a}| + \log  \upsilon_{j}+H(V_{j}|W) +H(Q_{\CalC_{j}}|Q_{\CalC_{j}^{C}},W)
\nonumber\\ 
\label{Eqn:Chnlbnd11}
&& - H(U_{\mathcal{A}_{j}}, Q_{\CalC_{j}},V_{j},U_{ij}\oplus
U_{kj}|U_{\mathcal{A}_{j}^{c}},Q_{\CalC_{j}^{C}},Y_{j},W), \nonumber 
\\
\label{Eqn:Chnlbnd12}
\alpha+\olineS+\olineT+K_{j}+L_{j}+\overline{\beta}+\overline{\nu} &\leq& \log(\upsilon_{i}\upsilon_{k})+H(Q_{ji},Q_{jk}|W)+H(V_{j}|W)\nonumber \\ && -H(\ulineU_{j*},\ulineQ_{j*},W,V_{j}|Y_{j},U_{ij}\oplus U_{kj}), \nonumber 
\\
\label{Eqn:Chnlbnd12}
\alpha+ \olineS + \olineT + K_{j} + L_{j} + \overline{\beta}+ \overline{\nu}+ S_{ij} + K_{ij} & \leq & \log(\upsilon_{i}\upsilon_{k})+H(Q_{ji},Q_{jk}|W)+H(V_{j}|W) \nonumber \\ && -H(\ulineU_{j*}, \ulineQ_{j*}, W, V_{j}, U_{ij}\oplus U_{kj}|Y_{j}), \nonumber 
\\
\label{Eqn:Chnlbnd12}
\alpha+ \olineS + \olineT + K_{j} + L_{j} + \overline{\beta} + \overline{\nu} + S_{ij} + K_{ij} & \leq &  \log(\upsilon_{i}\upsilon_{k})+H(Q_{ji}, Q_{jk}|W)+H(V_{j}|W) \nonumber \\ && -H(\ulineU_{j*}, \ulineQ_{j*}, W, V_{j}, U_{ij}\oplus U_{kj}|Y_{j}), \nonumber 
\end{eqnarray}
 for every $\mathcal{A}_{j} \subseteq \left\{ ji,jk\right\},\CalC_{j} \subseteq \left\{ ji,jk\right\}$ with
distinct indices $i,j,k$ in $\left\{ 1,2,3 \right\}$, where
$S_{\mathcal{A}_{j}} \define \sum_{a \in \mathcal{A}_{j}}S_{a}$,
$T_{\mathcal{A}_{j}} \define \sum_{a \in \mathcal{A}_{j}}T_{a}$, $\beta_{\CalC_{j}} = \sum _{jk \in \CalC_{j}}\beta_{jk}$, $\nu_{\CalC_{j}} = \sum _{jk \in \CalC_{j}}\nu_{jk}$,
$U_{\mathcal{A}_{j}} = (U_{a}:a \in \mathcal{A}_{j})$, $Q_{\mathcal{C}_{j}} = (Q_{c}:c \in \mathcal{C}_{j})$ with the notation $\olineS = S_{ji}+S_{jk}\olineT = T_{ji}+T_{jk}, \overline{\beta} = \beta_{ji}+\beta_{jk},\overline{\nu}= \nu_{ji}+\nu_{jk}$ and all the information quantities are evaluated wrt state 
\begin{eqnarray}
\label{Eqn:StageIITestChnl}
 \Phi^{W\dulineQ\!\!~\dulineU\!\!~\ulineU^{\oplus}\!\!~\ulineV X\!\!~\ulineY}\define \!\!\!\!\!\!\!
\sum_{\substack{w,\dulineq,\dulineu,\dulinev,x\\u_{1}^{\oplus},u_{2}^{\oplus},u_{3}^{\oplus} }}\!\!\!\!\!\!p_{W\dulineQ~\!\!\dulineU~\!\!\ulineV X}(w,\dulineq,\dulineu,\ulinev,x)\mathds{1}_{\left\{\substack{u_{ji}\oplus u_{jk}\\=u_{j}^{\oplus}:j\in[3]}\right\}}\!\ketbra{w~\dulineq~\dulineu~  u_{1}^{\oplus}~ u_{2}^{\oplus}~ u_{3}^{\oplus}~ \ulinev~ x} \!\otimes\! \rho_{x}\nonumber
\end{eqnarray}
Let $\alpha_{US}$ denote the convex closure of $\hat{\alpha}_{US}$. Then $\alpha_{US} \subseteq \ScrC(\tau)$ is an achievable rate region.
\end{theorem}
We shall provide a description of the code structure, encoding and decoding in an enlarged version of this manuscript. Specifically, Fig.~\ref{Fig:3CQBCStepIICodeStructure} depicts the code structure with the rate parameters as specified in the theorem statement. We shall explain details of this coding scheme and sketch elements of the proof in enlarged version of this manuscript.
\appendices
\section{Error Events at Rx $1$}
\label{AppSec:Rx1ErrorEvents}
\begin{prop}
 \label{Prop:BoundingT11}
 For every $\eta > 0$, there exists an $N_{\eta} \in \naturals$ such that for all $n \geq N_{\eta}$, we have $\Expectation\{ T_{11}\mathds{1}_{\CalE_{\ulinem}}\} \leq \exp\{ -n\eta\}$.
\end{prop}
\begin{proof}
We leverage the inequality $\tr(\lambda\sigma) \geq \tr(\lambda\rho) -\norm{\rho-\sigma}_{1}$ for sub-normalized states $0 \leq \rho,\sigma,\lambda \leq I$ to assert that
\begin{eqnarray}
 T_{11}\mathds{1}_{\CalE_{\ulinem}}  \geq \tr(\pi_{m_{1},b_{1}}^{\apl}\rho_{\ulinem}^{Y_{1}})\mathds{1}_{\CalE} - \norm{ \rho_{\ulinem}^{Y_{1}} - \pi_{m_{1},b_{1}}\rho_{\ulinem}^{Y_{1}}\pi_{m_{1},b_{1}} }_{1}\mathds{1}_{\CalE} \nonumber\\ - \norm{\rho_{\ulinem}^{Y_{1}} - \pi_{\apl}\rho_{\ulinem}^{Y_{1}}\pi_{\apl} }_{1}\mathds{1}_{\CalE} -\norm{\rho_{\ulinem}^{Y_{1}} - \pi^{Y_{1}}\rho_{\ulinem}^{Y_{1}}\pi^{Y_{1}} }_{1}\mathds{1}_{\CalE} 
  \nonumber
\end{eqnarray}
where we have abbreviated $\mathds{1}_{\CalE_{\ulinem}}$ as $\mathds{1}_{\CalE}$. From repeated use of pinching for non-commuting operators, we have $\min\{ \tr(\pi^{Y_{1}}\rho_{\ulinem}^{Y_{1}}), \tr(\pi_{\apl}\rho_{\ulinem}^{Y_{1}}),\tr(\pi_{m_{1},b_{1}}\rho_{\ulinem}^{Y_{1}}) \} \geq 1-6\sqrt{\eta}$ and the gentle operator lemma guarantees that $\Expectation\{T_{11}\mathds{1}_{\CalE_{\ulinem}}\} \geq 1-10\sqrt{\eta}$ for any $\eta> 0$ and sufficiently large $n$.
\end{proof}

\med\textbf{Proof of Prop.~\ref{Prop:Rx1Errors}} : We begin by analyzing $T_{12}$.
\med\textit{Analysis of $T_{12}$} : From \eqref{Eqn:ProbabilityofEevents2} and \eqref{Eqn:Rx1ErrorEventAssociatedProbs} we have
\begin{eqnarray}
\label{Eqn:MeasureRelationshipsT12}
 P(\CalE_{12}) \leq p_{XU_{1}\ulineV |W}^{n}(x^{n},u_{1}^{n},\ulinev^{n}|w^{n})p^{n}_{U_{1}}(\hatu_{1}^{n})\Phi_{12}\\
 \mbox{where }\Phi_{12} = \frac{2^{-n[(2+T_{2}+T_{3})\log q]}\mathds{1}_{\CalA \cap \CalB}}{2^{n(\tau-\eta -H(\ulineV,U_{1}|W)+H(U_{1}) )}}\nonumber
 \end{eqnarray}
and $\mathcal{A}$ and $\mathcal{B}$ are as defined in \eqref{Eqn:Rx1AnalysisG14} and three equations prior to it. The above can be derived using standard bounds on typical sequences. When we substitute \eqref{Eqn:MeasureRelationshipsT12} in a generic term of \eqref{Eqn:Rx1ErrorEventChars-12}, we obtain
\begin{eqnarray}
 \label{Eqn:Thm1Rx1ErrorT12Analysis1}
  \tr(\lambda_{\hatu_{1}^{n},w^{n}}\!\pi_{w^{n}}\rho_{x^{n}}\pi_{w^{n}} \!)p_{XU_{1}\ulineV |W}^{n}(u_{1}^{n},\ulinev^{n}|w^{n})p^{n}_{U_{1}}(\hatu_{1}^{n})\Phi_{12}.
 \end{eqnarray}
Substituting \eqref{Eqn:Thm1Rx1ErrorT12Analysis1} for a generic term in \eqref{Eqn:Rx1ErrorEventChars-12}, recognizing the terms other than $\rho_{x^{n}}$ do not depend on $u_{1}^{n},v_{2}^{n},v_{3}^{n}$, summing over the latter and recalling that all our information quantities are evaluated with respect to the state \eqref{Eqn:TheoremState}, we obtain
\begin{eqnarray}
\label{Eqn:Thm1Rx1ErrorT12Analysis2}
\lefteqn{~\Expectation\{T_{12}\mathds{1}_{\CalE_{\ulinem}} \} \leq \!\!\!\!\!\! \sum_{\substack{b_{1},a_{2},a_{3},w^{n}\\\hatm_{1},\hatb_{1},\hatu_{1}^{n} }}\!\!\!\! \!\!
 \tr(\lambda_{\hatu_{1}^{n},w^{n}}\!\pi_{w^{n}}\rho_{w^{n}}\pi_{w^{n}} \!)p_{U_{1}}^{n}(\hatu_{1}^{n})\Phi_{12}}\nonumber\\
 \label{Eqn:Thm1Rx1ErrorT12Analysis3}
 \lefteqn{~\leq \!\!\!\!\!\! \sum_{\substack{b_{1},a_{2},a_{3},w^{n}\\\hatm_{1},\hatb_{1},\hatu_{1}^{n} }}\!\!\!\! \!\!
 \tr(\lambda_{\hatu_{1}^{n},w^{n}}\!\pi_{w^{n}}\!)p_{U_{1}}^{n}(\hatu_{1}^{n})2^{-nH(Y_{1}|W)}\Phi_{12}}\\
 \label{Eqn:Thm1Rx1ErrorT12Analysis4}
 &&\!\!\!\!\!\!\!\!\!= \!\!\!\!\!\!\! \sum_{\substack{b_{1},a_{2},a_{3},w^{n}\\\hatm_{1},\hatb_{1},\hatu_{1}^{n} }}\!\!\!\! \!\!\!\!
 \tr(\pi^{Y_{1}}\pi_{\hatu_{1}^{n}} \pi_{ \hatu_{1}^{n}w^{n}}\pi_{\hatu_{1}^{n}} \pi^{Y_{1}}\!\pi_{w^{n}}\!)p_{U_{1}}^{n}(\hatu_{1}^{n})2^{-nH(Y_{1}|W)}\Phi_{12}
 \\
 \label{Eqn:Thm1Rx1ErrorT12Analysis5}
 &&\!\!\!\!\!\!\!\!\!= \!\!\!\!\!\!\! \sum_{\substack{b_{1},a_{2},a_{3},w^{n}\\\hatm_{1},\hatb_{1},\hatu_{1}^{n} }}\!\!\!\! \!\!\!\!
 \tr( \pi_{ \hatu_{1}^{n}w^{n}}\pi_{\hatu_{1}^{n}} \pi^{Y_{1}}\!\pi_{w^{n}}\!\pi^{Y_{1}}\pi_{\hatu_{1}^{n}})p_{U_{1}}^{n}(\hatu_{1}^{n})2^{-nH(Y_{1}|W)}\Phi_{12}
 \\
 \label{Eqn:Thm1Rx1ErrorT12Analysis6}
 \lefteqn{~\leq  \!\!\!\!\!\! \sum_{\substack{b_{1},a_{2},a_{3},w^{n}\\\hatm_{1},\hatb_{1},\hatu_{1}^{n} }}\!\!\!\! \!\!\!\!
 \tr( \pi_{ \hatu_{1}^{n}w^{n}}I)p_{U_{1}}^{n}(\hatu_{1}^{n})2^{-nH(Y_{1}|W)}\Phi_{12} }
 \\
 \label{Eqn:Thm1Rx1ErrorT12Analysis7}
  &&\!\!\!\!\!\!\!\!\!\leq  \!\!\!\!\!\! \sum_{\substack{b_{1},a_{2},a_{3},w^{n}\\\hatm_{1},\hatb_{1},\hatu_{1}^{n} }}\!\!\!\! \!\!\!\!
 p_{U_{1}}^{n}\!(\hatu_{1}^{n})2^{-n[I(Y_{1};U_{1}|W)-2\eta]}\Phi_{12} \leq 2^{-n[I(Y_{1},W;U_{1})-4\eta-R_{1}-B_{1} ]}
 \end{eqnarray}
where (a) we have used the operator inequality $\pi_{w^{n}}\rho_{w^{n}}\pi_{w^{n}} \leq 2^{-nH(Y_{1}|W)}\pi_{w^{n}}$ which holds since $w^{n}\in T_{\eta}^{n}(p_{V_{2}\oplus_{q}V_{3}})=T_{\eta}^{n}(p_{W})$ in \eqref{Eqn:Thm1Rx1ErrorT12Analysis3}, (b) \eqref{Eqn:Thm1Rx1ErrorT12Analysis4} follows by substitution for $\lambda_{\hatu_{1}^{n},w^{n}}$ from the definition prior to \eqref{Eqn:Rx1DecPOVMs}, (c) \eqref{Eqn:Thm1Rx1ErrorT12Analysis5} follows from cyclicity of the trace, (d) \eqref{Eqn:Thm1Rx1ErrorT12Analysis6} follows from the operator inequality $\pi_{\hatu_{1}^{n}} \pi^{Y_{1}}\!\pi_{w^{n}}\!\pi^{Y_{1}}\pi_{\hatu_{1}^{n}} \leq I$, (e) \eqref{Eqn:Thm1Rx1ErrorT12Analysis7} follows from the typicality bound $\tr( \pi_{ \hatu_{1}^{n}w^{n}}I) \leq 2^{nH(Y_{1}|U_{1},W)}$ which holds for $(u_{1}^{n},w^{n}) \in T_{\eta}^{n}(p_{U_{1}W})$.
\med\textit{Analysis of $T_{13}$} : From \eqref{Eqn:ProbabilityofEevents2} and \eqref{Eqn:Rx1ErrorEventAssociatedProbs} we have
\begin{eqnarray}
 \label{Eqn:MeasureRelationshipsE13}
 P(\CalE_{13}) \leq p_{X\ulineV |U_{1}W}^{n}(x^{n},\ulinev^{n}|u_{1}^{n},w^{n})p^{n}_{W|U_{1}}(\hatw^{n}|u_{1}^{n})\Phi_{13}
 \mbox{where }\Phi_{13}=\frac{2^{-n[(3+T_{2}+T_{3})\log q+H(U_{1})]}\mathds{1}_{\CalA \cap \CalB}}{2^{n(\tau-\eta -H(\ulineV|U_{1},W)-H(W|U_{1}) )}}
\end{eqnarray}
and $\mathcal{A}$ and $\mathcal{B}$ are as defined in \eqref{Eqn:Rx1AnalysisG14} and \eqref{Eqn:Rx1AnalysisF1}. The above can be derived using standard bounds on typical sequences. Substituting the upper bound \eqref{Eqn:MeasureRelationshipsE13} in \eqref{Eqn:Rx1ErrorEventChars-13}, we have
\begin{eqnarray}
 \Expectation\{T_{13}\mathds{1}_{\CalE_{\ulinem}}\} &\leq& \sum_{\substack{b_{1},a_{2},\hata_{\oplus}\\a_{3},u_{1}^{n},w^{n}}}\!\!\tr\!\left(\! \sum_{\hatw^{n}}p^{n}_{W|U_{1}}(\hatw^{n}|u_{1}^{n})\pi^{Y_{1}}\pi_{u_{1}^{n}} \pi_{ u_{1}^{n}\hatw^{n}}
  \label{Eqn:Thm1Rx1ErrorT13Analysis2}
 \pi_{u_{1}^{n}} \pi^{Y_{1}}\pi_{w^{n}}\!\!\!\!\!\!\sum_{x^{n},v_{2}^{n},v_{3}^{n}}\!\!\!\!\!p_{X\ulineV |U_{1}W}^{n}(x^{n},\ulinev^{n}|u_{1}^{n},w^{n})\rho_{x^{n}}\pi_{w^{n}} \!\!\right)\!\!\Phi_{13}\nonumber\\
 &\leq& 2^{n[H(Y_{1}|U_{1},W)+2\eta]}\!\!\!\!\!\! \sum_{\substack{b_{1},a_{2},\hata_{\oplus}\\a_{3},u_{1}^{n},w^{n}}}\!\!\!\!\tr\!\left(\! \sum_{\hatw^{n}}p^{n}_{W|U_{1}}(\hatw^{n}|u_{1}^{n})\pi^{Y_{1}}\pi_{u_{1}^{n}} \rho_{ u_{1}^{n}\hatw^{n}}
  \label{Eqn:Thm1Rx1ErrorT13Analysis4}
 \pi_{u_{1}^{n}} \pi^{Y_{1}}\pi_{w^{n}}\rho_{u_{1}^{n}w^{n}}\pi_{w^{n}} \!\right)\Phi_{13}
 \\
 \label{Eqn:Thm1Rx1ErrorT13Analysis5}
 &\leq& 2^{n[H(Y_{1}|U_{1},W)+2\eta]}\!\!\!\!\!\! \sum_{\substack{b_{1},a_{2},\hata_{\oplus}\\a_{3},u_{1}^{n},w^{n}}}\!\!\!\!\!\tr\!\left(\pi^{Y_{1}}\pi_{u_{1}^{n}} \rho_{ u_{1}^{n}}\pi_{u_{1}^{n}} \pi^{Y_{1}}\pi_{w^{n}}\rho_{u_{1}^{n}w^{n}}\pi_{w^{n}} \!\right)\Phi_{13}\nonumber\\
  \label{Eqn:Thm1Rx1ErrorT13Analysis6}
 &\leq& 2^{-n[I(Y_{1};W|U_{1})-4\eta]}\!\!\!\!\!\! \sum_{\substack{b_{1},a_{2},\hata_{\oplus}\\a_{3},u_{1}^{n},w^{n}}}\!\!\!\!\!\tr\!\left(\pi^{Y_{1}}\pi_{u_{1}^{n}}  \pi^{Y_{1}}\pi_{w^{n}}\rho_{u_{1}^{n}w^{n}}\pi_{w^{n}} \!\right)\Phi_{13}\\
  \label{Eqn:Thm1Rx1ErrorT13Analysis7}
 &\leq& 2^{-n[I(Y_{1};W|U_{1})-4\eta]}\!\!\!\!\!\! \sum_{\substack{b_{1},a_{2},\hata_{\oplus}\\a_{3},u_{1}^{n},w^{n}}}\!\!\!\!\!\tr\!\left(I_{Y_{1}}^{\otimes n}\pi_{w^{n}}\rho_{u_{1}^{n}w^{n}}\pi_{w^{n}} \!\right)\Phi_{13}\\
\label{Eqn:Thm1Rx1ErrorT13Analysis8}
 &\leq& 2^{-n[I(Y_{1};W|U_{1})-4\eta]}\!\!\!\!\!\! \sum_{\substack{b_{1},a_{2},\hata_{\oplus}\\a_{3},u_{1}^{n},w^{n}}}\!\!\!\!\!\tr\!\left(\rho_{u_{1}^{n}w^{n}} \!\right)\Phi_{13}
 \\
 \label{Eqn:Thm1Rx1ErrorT13Analysis9}
 &\leq& 2^{-n[I(Y_{1},U_{1};W)+\log q-H(W) - \max\{S_{2}\log q,S_{3}\log q\}-6\eta]}
 \end{eqnarray}
 where (a) \eqref{Eqn:Thm1Rx1ErrorT13Analysis4} follows from the fact that for $(u_{1}^{n},\hatw^{n}) \in T_{\eta}^{n}(p_{U_{1}W})$, we have 
 \begin{eqnarray}
 \pi_{u_{1}^{n}\hatw^{n}} \leq 2^{n[H(Y_{1}|U_{1},W)+2\eta]}\pi_{u_{1}^{n}\hatw^{n}} \rho_{u_{1}^{n}\hatw^{n}} \pi_{u_{1}^{n}\hatw^{n}}  = 2^{n[H(Y_{1}|U_{1},W)+2\eta]}\sqrt{\rho_{u_{1}^{n}\hatw^{n}}}\pi_{u_{1}^{n}\hatw^{n}} \sqrt{\rho_{u_{1}^{n}\hatw^{n}}}  \nonumber\\
 \leq 2^{n[H(Y_{1}|U_{1},W)+2\eta]}\sqrt{\rho_{u_{1}^{n}\hatw^{n}}} \sqrt{\rho_{u_{1}^{n}\hatw^{n}}} = 2^{n[H(Y_{1}|U_{1},W)+2\eta]}\rho_{u_{1}^{n}\hatw^{n}}
 \nonumber
 \end{eqnarray}
 which is based on the analysis of the second term in \cite[Eqn.~\(21\)]{201206TIT_FawHaySavSenWil}
 (b) \eqref{Eqn:Thm1Rx1ErrorT13Analysis6} follows from $2^{-n[H(Y_{1}|U_{1})-2\eta]}\pi_{u^{n}} \leq \pi_{u_{1}^{n}}\rho_{u_{1}^{n}}\pi_{u_{1}^{n}}$ for $u_{1}^{n} \in T_{\eta}^{n}(p_{U_{1}})$, (c) \eqref{Eqn:Thm1Rx1ErrorT13Analysis7} follows from the operator inequality $\pi^{Y_{1}}\pi_{u_{1}^{n}}  \pi^{Y_{1}} \leq I_{Y_{1}^{\otimes n}}$, (d) \eqref{Eqn:Thm1Rx1ErrorT13Analysis8} follows from cyclicity of the tracem $\pi_{w^{n}}\pi_{w^{n}}=\pi_{w^{n}} \leq I_{Y_{1}}^{\otimes n}$, and (e) the last inequality \eqref{Eqn:Thm1Rx1ErrorT13Analysis9} follows from substituting $\Phi_{13}$ from \eqref{Eqn:MeasureRelationshipsE13}.
\med\textit{Analysis of $T_{14}$} : From \eqref{Eqn:ProbabilityofEevents2} and \eqref{Eqn:Rx1ErrorEventAssociatedProbs} we have
\begin{eqnarray}
\label{Eqn:MeasureRelationships14}
 P(\CalE_{14}) \leq p_{XU_{1}\ulineV|W}^{n}(x^{n},u_{1}^{n},\ulinev^{n}|w^{n})p^{n}_{W}(w^{n})\Phi_{14}
 \mbox{ where }\Phi_{14} = \frac{2^{-n[(3+T_{2}+T_{3})\log q]}\mathds{1}_{\CalA \cap \CalB}}{2^{n(\tau-\eta -H(\ulineV,U_{1})+2H(U_{1}) )}}
 \end{eqnarray}
 and $\mathcal{A}$ and $\mathcal{B}$ are as defined in \eqref{Eqn:Rx1AnalysisG14} and three equations prior to it. The above can be derived using standard bounds on typical sequences. When we substitute \eqref{Eqn:MeasureRelationships14} in a generic term of \eqref{Eqn:Rx1ErrorEventChars-14}, we obtain
\begin{eqnarray}
 \label{Eqn:Thm1Rx1ErrorT14Analysis1}
 \Expectation\{T_{14}\mathds{1}_{\CalE_{\ulinem}}\} &\leq&\!\!\!\!\!\! \sum_{\substack{b_{1},a_{2},\hata_{\oplus},\hatm_{1},\hatb_{1}\\a_{3},\hatu_{1}^{n},\hatw^{n},w^{n}}}\!\!\!\!\!\!p_{W}^{n}(w^{n})\tr\!\left(\! \pi^{Y_{1}}\pi_{\hatu_{1}^{n}} \pi_{ \hatu_{1}^{n}\hatw^{n}}\pi_{\hatu_{1}^{n}} \pi^{Y_{1}}\pi_{w^{n}}
  \label{Eqn:Thm1Rx1ErrorT13Analysis2}
 \sum_{x^{n},v_{2}^{n},v_{3}^{n},u_{1}^{n}}\!\!\!\!\!\!\!p_{X\ulineV U_{1}|W}^{n}(x^{n},\ulinev^{n},u_{1}^{n}|w^{n})\rho_{x^{n}}\pi_{w^{n}} \!\!\right)\!\!\Phi_{14}\nonumber\\
\label{Eqn:Thm1Rx1ErrorT14Analysis3}
 &\leq&\!\!\!\!\!\!\!\!\!\!\!\! \sum_{\substack{b_{1},a_{2},\hata_{\oplus},\hatm_{1},\hatb_{1}\\a_{3},\hatu_{1}^{n},\hatw^{n},w^{n}}}\!\!\!\!\!\!\!\!\!\!\!p_{W}^{n}(w^{n})\tr\!\left(\! \pi^{Y_{1}}\pi_{\hatu_{1}^{n}} \pi_{ \hatu_{1}^{n}\hatw^{n}}\pi_{\hatu_{1}^{n}} \pi^{Y_{1}}\pi_{w^{n}}\rho_{w^{n}}\pi_{w^{n}}\right)\!\Phi_{14}\nonumber\\
 \label{Eqn:Thm1Rx1ErrorT14Analysis4}
 &\leq&\!\!\!\!\!\!\!\!\!\!\! \sum_{\substack{b_{1},a_{2},\hata_{\oplus},\hatm_{1},\hatb_{1}\\a_{3},\hatu_{1}^{n},\hatw^{n},w^{n}}}\!\!\!\!\!\!\!\!\!p_{W}^{n}(w^{n})\tr\!\left(\! \pi^{Y_{1}}\pi_{\hatu_{1}^{n}} \pi_{ \hatu_{1}^{n}\hatw^{n}}\pi_{\hatu_{1}^{n}} \pi^{Y_{1}}\rho_{w^{n}}\right)\!\Phi_{14}\\
 \label{Eqn:Thm1Rx1ErrorT14Analysis5}
 &\leq&\!\!\!\! \sum_{\substack{b_{1},a_{2},\hata_{\oplus},\hatm_{1}\\\hatb_{1},a_{3},\hatu_{1}^{n},\hatw^{n}}}\!\!\!\!\!\tr\!\left(\! \pi^{Y_{1}}\pi_{\hatu_{1}^{n}} \pi_{ \hatu_{1}^{n}\hatw^{n}}\pi_{\hatu_{1}^{n}} \pi^{Y_{1}}\rho^{\otimes n}\right)\!\Phi_{14}\\
 \label{Eqn:Thm1Rx1ErrorT14Analysis6}
 &=& \sum_{\substack{b_{1},a_{2},\hata_{\oplus},\hatm_{1}\\\hatb_{1},a_{3},\hatu_{1}^{n},\hatw^{n}}}\!\!\!\!\!\tr\!\left(\! \pi_{\hatu_{1}^{n}} \pi_{ \hatu_{1}^{n}\hatw^{n}}\pi_{\hatu_{1}^{n}} \pi^{Y_{1}}\rho^{\otimes n}\pi^{Y_{1}}\right)\!\Phi_{14}\nonumber\\
 \label{Eqn:Thm1Rx1ErrorT14Analysis7}
 &\leq& 2^{-nH(Y_{1})}\sum_{\substack{b_{1},a_{2},\hata_{\oplus},\hatm_{1}\\\hatb_{1},a_{3},\hatu_{1}^{n},\hatw^{n}}}\!\!\!\!\!\tr\!\left(\! \pi_{\hatu_{1}^{n}} \pi_{ \hatu_{1}^{n}\hatw^{n}}\pi_{\hatu_{1}^{n}} \pi^{Y_{1}}\right)\!\Phi_{14}\\
 \label{Eqn:Thm1Rx1ErrorT14Analysis8}
 &\leq& 2^{-n[H(Y_{1})-2\eta]}\sum_{\substack{b_{1},a_{2},\hata_{\oplus},\hatm_{1}\\\hatb_{1},a_{3},\hatu_{1}^{n},\hatw^{n}}}\!\!\!\!\!\tr\!\left(\! \pi_{\hatu_{1}^{n}} \pi_{ \hatu_{1}^{n}\hatw^{n}}\pi_{\hatu_{1}^{n}}\right)\!\Phi_{14}\nonumber\\
 \label{Eqn:Thm1Rx1ErrorT14Analysis9}
 &\leq& 2^{-n[H(Y_{1})-H(Y_{1}|U_{1},W)-4\eta]}\Phi_{14}\!\!\sum_{\substack{b_{1},a_{2},\hata_{\oplus}\\\hatm_{1}\hatb_{1},a_{3}}}\!1
  \nonumber\\
  \label{Eqn:Thm1Rx1ErrorT14Analysis10}
  &\leq& 2^{-n\left[{I(Y_{1};U_{1},W)+I(U_{1};W)+\log q  -H(W)-R_{1}-B_{1}-\max\{S_{2},S_{3}\}\log q }-6\eta\right]}
 \end{eqnarray}
 where (a) \eqref{Eqn:Thm1Rx1ErrorT14Analysis4} follows since $\pi_{w^{n}}\rho_{w^{n}}\pi_{w^{n}} =\sqrt{\rho_{w^{n}}}\pi_{w^{n}}\pi_{w^{n}}\sqrt{\rho_{w^{n}}} = \sqrt{\rho_{w^{n}}}\pi_{w^{n}}\sqrt{\rho_{w^{n}}} \leq \sqrt{\rho_{w^{n}}}I\sqrt{\rho_{w^{n}}} = \rho_{w^{n}}$, (b) \eqref{Eqn:Thm1Rx1ErrorT14Analysis5} follows from $\sum_{w^{n} \in T_{\eta}^{n}(p_{W})}p_{W}^{n}(w^{n})\rho_{w^{n}} \leq \rho^{\otimes n}$, (c) \eqref{Eqn:Thm1Rx1ErrorT14Analysis7} follows from  $\pi^{Y_{1}}\rho^{\otimes n}\pi^{Y_{1}} \leq 2^{-n[H(Y_{1})-2\eta] }\pi^{Y_{1}}$ and \eqref{Eqn:Thm1Rx1ErrorT14Analysis10} follows from substituting $\Phi_{14}$ from \eqref{Eqn:MeasureRelationships14}.
\section{Error Event at Receivers $2$ and $3$}
\label{AppSec:ErrorEventAtRx2And3}
\med\textbf{Proof of Prop.~\ref{Prop:Rx2And3ErrorEvent}} : Towards deriving an upper bound on $T_{j1} : j =2,3$, we define events analogous to those defined in \eqref{Eqn:Rx1AnalysisG14}. Let
\begin{eqnarray}
 \label{Eqn:Thm1Rx2ErrorEvents1}
 \tilde{\CalF}_{1}\! \define\! \left\{\!\!\! \!
 \begin{array}{c} 
 V_{j}^{n}(a_{j})=v_{j}^{n}, \iota_{j}(a_{j}) = m_{j}:j= 2,3,U_{1}(m_{1},b_{1}) = u_{1}^{n} \end{array} \!\!\!\!\right\}, \tilde{\CalB}\! \define\! {\left\{ \!{(u_{1}^{n},\ulinev^{n}) \in T_{\eta}(p_{U_{1}\ulineV }\!) }\! \right\}}
 \nonumber\\
  {\CalF}_{2} \define\left\{ \!\!\! \! \begin{array}{c}B_{1}(\ulinem)=b_{1},A_{j}(\ulinem)= a_{j}  \end{array}\!\!\! : \!\!\begin{array}{c}j=2,3 \end{array}\!\!\!\right\},~ \CalG_{j} = \left\{ \!\!\!\! \begin{array}{c}V_{j}^{n}(\hata_{j})=\hatv_{j}^{n}, \hata_{j} \neq a_{j}\end{array} \! \right\} ,~
 \CalF_{3}\!\define \!\{X^{n}(\ulinem) = x^{n} \!\},~~\tilde{\CalA}\define \{\hatv_{j}^{n} \in T_{\eta}(p_{V_{j}})\}
 \nonumber
\end{eqnarray}
where $\tilde{\CalF}_{1}\cap \CalF_{2}\cap\CalF_{3}$ specifies the realization for the chosen codewords and $\CalG_{j}$ specify the realization of an incorrect codeword in regards to $T_{j}$. Let $\CalE_{j} \define \tilde{\CalF}_{1}\cap\CalG_{j}\cap\CalE_{\ulinem} \cap \CalF_{2}\cap\CalF_{3}\cap\tilde{\CalA}\cap\tilde{\CalB}$. With this, we have
\begin{eqnarray}
 \label{Eqn:Thm1Rx2ErrorTerm}
&&\!\!\!\!\! \!\!\!\!\!\!\!\!\mathbb{E}\{T_{j1}\mathds{1}_{\CalE_{\ulinem}}\} = \!\!\!\!\!\! \sum_{\substack{b_{1},a_{2},a_{3}\\u_{1}^{n},v_{2}^{n},v_{3}^{n}}}\sum_{\substack{\hata_{j}^{n},x^{n}\\\hatv_{j}^{n}}}\!\!
 \tr(\pi^{Y_{j}}\pi_{\hatv_{j}^{n}}\pi^{Y_{j}}\rho^{Y_{j}}_{x^{n}} \!)P(\CalE_{j}),
 \\
 \label{Eqn:Thm1Rx2ErrorTermProb1}
  P(\CalE_{j}) &=&  P(\CalE_{\ulinem}\cap\tilde{\CalF}_{1}\cap\CalG_{j})P(\CalF_{2}|\CalE_{\ulinem}\cap\tilde{\CalF}_{1}\cap\CalG_{j})\nonumber
  \label{Eqn:Thm1Rx2ErrorTermProb2}
  \times P(\CalF_{3}|\CalF_{2}\cap\CalE_{\ulinem}\cap\tilde{\CalF}_{1}\cap\CalG_{j})
 \nonumber\\
 \label{Eqn:Thm1Rx2ErrorTermProb3}
 &\leq& p^{n}_{XU_{1}\ulineV}(x^{n},u_{1}^{n},\ulinev^{n})\Phi_{j}\mbox{ where }
 \Phi_{j}\define \frac{2^{-n[H(U_{1})-H(U_{1},\ulineV)]}\cdot \mathds{1}_{\tilde{\CalA}\cap\tilde{\CalB}}}{2^{n(\tau-\eta)}2^{n(3+T_{2}+T_{3})\log q} }.
   \end{eqnarray}
The above upper bound is obtained analogous to our analysis in \eqref{Eqn:ProbabilityofEevents2}, \eqref{Eqn:Rx1ErrorEventAssociatedProbs}, \eqref{Eqn:MeasureRelationshipsE13}, \eqref{Eqn:MeasureRelationships14}. As done in our analysis of $\mathbb{E}\{T_{1k}\mathds{1}_{\CalE_{\ulinem}}\}$ for $k=2,3$ earlier, we substitute the upper bound \eqref{Eqn:Thm1Rx2ErrorTermProb3} in \eqref{Eqn:Thm1Rx2ErrorTerm} to obtain
\begin{eqnarray}
 \label{Eqn:Thm1Rx2MainErrorEventAnalysis1}
 \mathbb{E}\{T_{j1}\mathds{1}_{\CalE_{\ulinem}}\} \leq \!\!\!\!\! \sum_{\substack{b_{1},a_{2},a_{3}\\\hata_{j}^{n}\hatv_{j}^{n}}}\!
 \!\!\!\!\tr(\!\!\pi^{Y_{j}}\pi_{\hatv_{j}^{n}}\pi^{Y_{j}}\!\!\!\!\!\!\!\sum_{\substack{u_{1}^{n},x^{n},\ulinev^{n}}}\!\!\!\!\!\!p^{n}_{XU_{1}\ulineV}(x^{n},u_{1}^{n},\ulinev^{n})\rho^{Y_{j}}_{x^{n}} \!\!)\Phi_{j}\!\!\!\!\!\!\!\!\!\!\!\!\!\!\!\!\!
 \nonumber\\
 \label{Eqn:Thm1Rx2MainErrorEventAnalysis2}
 \leq \!\!\!\!\! \sum_{\substack{b_{1},a_{2},a_{3}\\\hata_{j}^{n}\hatv_{j}^{n}}}\!
 \!\!\!\!\tr(\pi^{Y_{j}}\pi_{\hatv_{j}^{n}}\pi^{Y_{j}}\rho^{\otimes n} )\Phi_{j} =\!\!\!\!\! \sum_{\substack{b_{1},a_{2},a_{3}\\\hata_{j}^{n}\hatv_{j}^{n}}}\!\!\!\!
 \tr(\pi_{\hatv_{j}^{n}}\pi^{Y_{j}}\rho^{\otimes n}\pi^{Y_{j}} )\Phi_{j}
 \nonumber\\
 \label{Eqn:Thm1Rx2MainErrorEventAnalysis3}
 =2^{-n[H(Y_{j})-2\eta]}\!\!\!\!\! \sum_{\substack{b_{1},a_{2},a_{3}\\\hata_{j}^{n}\hatv_{j}^{n}}}\!\!\!\!
 \tr(\pi_{\hatv_{j}^{n}}\pi^{Y_{j}} )\Phi_{j} \leq 2^{-n[I(Y_{j};V_{j})-4\eta]}\!\!\!\!\! \sum_{\substack{b_{1},a_{2},a_{3}\\\hata_{j}^{n}\hatv_{j}^{n}}}\!\!\!\!
 \Phi_{j}\\
 \label{Eqn:Thm1Rx2MainErrorEventAnalysis4}
 \leq 2^{-n[I(Y_{j};V_{j})+\log q-H(V_{j})-S_{j}\log q -6\eta]}
 \end{eqnarray}
where (a) \eqref{Eqn:Thm1Rx2MainErrorEventAnalysis3} follows from the operator inequality $\pi^{Y_{j}}\rho^{\otimes n}\pi^{Y_{j}} \leq 2^{-n[H(Y_{j})-2\eta]}\pi^{Y_{j}}$ and $\tr(\pi_{\hatv_{j}^{n}}\pi^{Y_{j}} ) \leq \tr(\pi_{\hatv_{j}^{n}} ) \leq 2^{n[H(Y_{j}|V_{j})+2\eta]}$ and (b) the last inequality \eqref{Eqn:Thm1Rx2MainErrorEventAnalysis4} follows from substituting $\Phi_{j}$ in \eqref{Eqn:Thm1Rx2ErrorTermProb3}.
\section{Verification of \eqref{Eqn:ClosenessOfTiltedState}}
\label{AppSec:ClosenessOfStates}
To verify
\begin{eqnarray}
     \norm{\theta_{\ulinez^{n}, \ulinea^{n}}^{\otimes n} - \xi_{\ulinez^{n}}^{\otimes n} \otimes \ket{0}\bra{0} }_{1} \leq 90 \eta^{n}, \nonumber 
\end{eqnarray}
it is sufficient to show that
\begin{eqnarray}
     \norm{ \CalT_{ \ulinea^{n}, \eta^{n}} ( \ket{h}\bra{h}) - \ket{h}\bra{h} }_{1} \leq 90 \eta^{n}.
     \nonumber 
\end{eqnarray}
Towards the same, note that
\begin{eqnarray}
    &&\CalT_{ \ulinea^{n}, \eta^{n}} ( \ket{h}\bra{h}) = \frac{1}{\Omega(\eta)}  
    \left(\ket{h} + \sum_{S \subseteq [4]\setminus [4]} \eta^{n |S| } \ket{h} \otimes \ket{a_{s}^{n}} \right)  \times \left(\bra{h} + \sum_{\Tilde{S} \subseteq [4]\setminus [4]} \eta^{n |\Tilde{S}| } \bra{h} \otimes \bra{a_{\Tilde{S}}^{n}} \right) \nonumber 
    \\&=&  \! \! \frac{1}{\Omega(\eta)} \ket{h} \bra{h} \! \!+\! \! \frac{1}{\Omega(\eta)} \ket{h}  \left( \sum_{\Tilde{S} \subseteq [4] \setminus [4]} \! \! \!  \eta^{n |\Tilde{S}| } \bra{h} \otimes \bra{a_{\Tilde{S}}^{n}} \right) \! \!  + \! \! \frac{1}{\Omega(\eta)} \left( \sum_{S \subseteq [4]\setminus [4]} \! \! \!  \eta^{n |S| } \ket{h} \otimes \ket{a_{s}^{n}} \right)  \bra{h} \nonumber \\
    &&+ \frac{1}{\Omega(\eta)} \! \! \left( \sum_{S \subseteq [4] \setminus [4]} \! \! \! \eta^{n |S| } \ket{h} \otimes \ket{a_{s}^{n}} \right)\! \! \!  \left( \sum_{\Tilde{S} \subseteq [4] \setminus [4]} \! \! \! \eta^{n |\Tilde{S}| } \bra{h} \otimes \bra{a_{\Tilde{S}}^{n}} \right).\nonumber
\end{eqnarray}
Next, note that
\begin{eqnarray}
\left\| \CalT_{ \ulinea^{n}, \eta^{n}} ( \ket{h}\bra{h}) - \ket{h}\bra{h} \right\|_{1}  \leq \frac{1-\Omega(\eta)}{\Omega(\eta)} \left\|\ket{h} \bra{h} \right\|_{1}   + \frac{1}{\Omega(\eta)} \left\| \ket{h} \left( \sum_{\Tilde{S} \subseteq [4]\setminus [4]} \eta^{n |\Tilde{S}| } \bra{h} \otimes \bra{a_{\Tilde{S}}^{n}} \right) \right\|_{1} +~~~~~~~~~~~~~~~\nonumber\\
\!\!\!\!\!\!\!\!\!\!\!\!\frac{1}{\Omega(\eta)} \left\| \left( \sum_{S \subseteq [4]\setminus [4]} \eta^{n |S| } \ket{h} \otimes \ket{a_{s}^{n}} \right)  \!\!\bra{h} \right\|_{1} \! \! \!+\! \frac{1}{\Omega(\eta)} \bigg\| \left( \sum_{S \subseteq [4]\setminus [4]} \!\!\!\! \eta^{n |S| } \ket{h} \otimes \ket{a_{s}^{n}} \right)  \!\!\!\left( \sum_{\Tilde{S} \subseteq [4]\setminus [4]} \!\!\!\! \eta^{n |\Tilde{S}| } \bra{h} \otimes \bra{a_{\Tilde{S}}^{n}} \right) \bigg\|_{1} \label{eq:line4}.
\end{eqnarray}
In the following, we analyze each term on the RHS of \eqref{eq:line4} separately. Towards the same, note that
\begin{eqnarray}
    \frac{1-\Omega(\eta)}{\Omega(\eta)} \left\|\ket{h} \bra{h} \right\|_{1} = \frac{16 \eta^{2n} + 36 \eta^{4n} + 16 \eta^{6n}}{\Omega(\eta)}, \nonumber
   \end{eqnarray}
\begin{eqnarray}
    \left\| \ket{h} \left( \sum_{\Tilde{S} \subseteq [4] \setminus [4]} \eta^{n |\Tilde{S}|} \bra{h} \otimes \bra{a_{j\Tilde{S}}^{n}} \right) \right\|_{1} \! \! \!\!\!\! \!&=& \! \!\! \! \! \tr \left( \sqrt{
    \left( \sum_{S \subseteq [4] \setminus [4]} \eta^{n |S|} \ket{h} \otimes \ket{a_{jS}^{n}} \right) \bra{h} \ket{h} \left( \sum_{\Tilde{S} \subseteq [4] \setminus [4]} \eta^{n |\Tilde{S}|} \bra{h} \otimes \bra{a_{j\Tilde{S}}^{n}} \right)
} \right) \nonumber \\
\! \! \!\!\!\! \!&=& \! \!\! \! \! \tr \left( \sqrt{
    \left( \sum_{S \subseteq [4] \setminus [4]} \eta^{n |S|} \ket{h} \otimes \ket{a_{jS}^{n}} \right)  
    \left( \sum_{\Tilde{S} \subseteq [4] \setminus [4]} \eta^{n |\Tilde{S}|} \bra{h} \otimes \bra{a_{j\Tilde{S}}^{n}} \right)
} \right) \nonumber \\
\! \! \!\!\!\! \!&=& \! \!\! \! \!\sqrt{ \left( \sum_{S \subseteq [4] \setminus [4]} \eta^{n |S|} \bra{h} \otimes \bra{a_{jS}^{n}} \right)  
    \left( \sum_{\Tilde{S} \subseteq [4] \setminus [4]} \eta^{n |\Tilde{S}|} \ket{h} \otimes \ket{a_{j\Tilde{S}}^{n}} \right) } \nonumber  \\
\! \! \!\!\!\! \!&=& \! \!\! \! \! \sqrt{ \sum_{S \subseteq [4] \setminus [4]} \sum_{\Tilde{S} \subseteq [4] \setminus [4]}
    \eta^{n |S|} \eta^{n |\Tilde{S}|} \bra{h} \ket{h} \otimes \bra{a_{jS}^{n}} \ket{a_{j\Tilde{S}}^{n}} } \nonumber \\
\! \! \!\!\!\! \!&=& \! \!\! \! \!\sqrt{ \sum_{S \subseteq [4] \setminus [4]} \eta^{2n |S|}} \nonumber = \sqrt{ 4 \eta^{2n} + 6 \eta^{4n} + 4 \eta^{6n}}, \nonumber 
\end{eqnarray}
\begin{eqnarray}
     \left\| \left( \sum_{S \subseteq [4]\setminus [4]} \eta^{n |S| } \ket{h} \otimes \ket{a_{S}^{n}} \right)  \bra{h} \right\|_{1} \! \! \!\!\!\! \!&=& \! \!\! \! \! \tr \left( \sqrt{ \ket{h}   \left( \sum_{\Tilde{S} \subseteq [4]\setminus [4]} \eta^{n |\Tilde{S}| } \bra{h} \otimes \bra{a_{S}^{n}} \right) \left( \sum_{S \subseteq [4]\setminus [4]} \eta^{n |S| } \ket{h} \otimes \ket{a_{S}^{n}} \right) \bra{h}} \right) \nonumber  \\
\! \! \!\!\!\! \!&=& \! \!\! \! \!\sqrt{4 \eta^{2n} + 6 \eta^{4n} + 4 \eta^{6n}},\mbox{ and} 
\nonumber   
\end{eqnarray}
\begin{eqnarray}
            \bigg\| \left( \sum_{S \subseteq [4] \setminus [4]} \! \! \eta^{n |S|} \ket{h} \otimes \ket{a_{S}^{n}} \right) \! \!
         \left( \sum_{\Tilde{S} \subseteq [4] \setminus [4]} \! \! \eta^{n |\Tilde{S}|} \bra{h} \otimes \bra{a_{\Tilde{S}}^{n}} \right) \bigg\|_{1} \!\!
      \! \! \!\! \! \!&=& \! \!\! \! \! \left( \sum_{S \subseteq [4] \setminus [4]} \! \! \!\eta^{n |S|} \bra{h} \otimes \bra{a_{S}^{n}} \! \! \right) \!\! \left( \sum_{\Tilde{S} \subseteq [4] \setminus [4]} \! \! \! \eta^{n |\Tilde{S}|} \ket{h} \otimes \ket{a_{\Tilde{S}}^{n}} \! \! \right) \nonumber  \\
\! \! \!\!\!\! \!&=& \! \!\! \! \!4 \eta^{2n} + 6 \eta^{4n} + 4 \eta^{6n}. \nonumber
\end{eqnarray}
Hence,
\begin{equation}
    \begin{aligned}
        \left\| \CalT_{ \ulinea^{n}, \eta^{n}} ( \ket{h}\bra{h}) - \ket{h}\bra{h} \right\|_{1} 
        &\leq \frac{16 \eta^{2n} + 36 \eta^{4n} + 16 \eta^{6n}}{\Omega(\eta)}  + 2 \left(\frac{\sqrt{4 \eta^{2n} + 6 \eta^{4n} + 4 \eta^{6n}}}{\Omega(\eta)}\right)  + \frac{4 \eta^{2n} + 6 \eta^{4n} + 4 \eta^{6n}}{\Omega(\eta)} \\
        &= \frac{20 \eta^{2n} + 42 \eta^{4n} + 20 \eta^{6n}}{\Omega(\eta)} + \frac{2\sqrt{4 \eta^{2n} + 6 \eta^{4n} + 4 \eta^{6n}}}{\Omega(\eta)} \\
        & \leq \frac{82 \eta^{2n}+ 2 \sqrt{14} \eta^{n}}{\Omega(\eta)} \\
        &\leq 90 \eta^{n} .\nonumber 
    \end{aligned}
\end{equation}
\section{Upper bound on first two terms of \eqref{Eqn:TheFinalThreeTermsForEff5TxMAC}}
\label{AppSec:bounds}
We begin by proving 
\begin{eqnarray}
    \frac{1}{\mathcal{M}_{1}} \sum_{\ulinem_{1} \in \mathcal{M}_{1}} 
\norm{\theta_{(\ulinez_{1}^{n}, \ulinea_{1}^{n})(\ulinem_1)}^{\otimes n} - 
\xi_{\ulinez_{1}^{n}(\ulinem_1)}^{\otimes n} \otimes \ket{0}\bra{0} }_{1} 
\leq 90 \sqrt{\epsilon} + \epsilon. \nonumber 
\end{eqnarray}
We have, 
\begin{eqnarray}
\frac{1}{\mathcal{M}_{1}} \! \! \sum_{\ulinem_{1} \in \mathcal{M}_{1}} \! \! \!
\norm{\theta_{(\ulinez_{1}^{n}, \ulinea_{1}^{n})(\ulinem_1)}^{\otimes n} \! \! - 
\xi_{\ulinez_{1}^{n}(\ulinem_1)}^{\otimes n} \otimes \ket{0}\bra{0} }_{1} 
\! \! = \! \! \frac{1}{\mathcal{M}_{1}} \! \! \! \sum_{\ulinem_{1} \in \mathcal{M}_{1}} \sum_{\ulinez_{1}^n} \sum_{\ulinea_{1}^n}
\norm{\theta_{(\ulinez_{1}^{n}, \ulinea_{1}^{n})(\ulinem_1)}^{\otimes n} \! \! - 
\xi_{\ulinez_{1}^{n}(\ulinem_1)}^{\otimes n} \otimes \ket{0}\bra{0} }_{1} \nonumber \\ \mathds{1}
 \{\ulinez_{1}^{n}(\ulinem_1)=\ulinez_{1}^{n}, \ulinea_{1}^{n}(\ulinem_{1}) = \ulinea_{1}^{n}\},  \nonumber 
\end{eqnarray}
by evaluating the expectation

\begin{eqnarray}
   && \frac{1}{\mathcal{M}_{1}} \frac{1}{|\uline{\CalA_{1}|}^n} \sum_{\ulinem_{1} \in \mathcal{M}_{1}} \sum_{\ulinez_{1}^n} \sum_{\ulinea_{1}^n} p_{\ulineZ_{1}}^n (\ulinez_{1}^{n})
\norm{\theta_{(\ulinez_{1}^{n}, \ulinea_{1}^{n})}^{\otimes n} - 
\xi_{\ulinez_{1}^{n}}^{\otimes n} \otimes \ket{0}\bra{0} }_{1}  \nonumber \\
&=& \frac{1}{\mathcal{M}_{1}} \frac{1}{|\uline{\CalA_{1}|}^n} \sum_{\ulinem_{1} \in \mathcal{M}_{1}} \sum_{\ulinez_{1}^n  \in T^{n}_{\delta}(p_{\ulineZ_{1}} )} \sum_{\ulinea_{1}^n} p_{\ulineZ_{1}}^n (\ulinez_{1}^{n})
\norm{\theta_{(\ulinez_{1}^{n}, \ulinea_{1}^{n})}^{\otimes n} - 
\xi_{\ulinez_{1}^{n}}^{\otimes n} \otimes \ket{0}\bra{0} }_{1}
\nonumber \\
&+& \frac{1}{\mathcal{M}_{1}} \frac{1}{|\uline{\CalA_{1}|}^n} \sum_{\ulinem_{1} \in \mathcal{M}_{1}} \sum_{\ulinez_{1}^n  \notin T^{n}_{\delta}(p_{\ulineZ_{1}} )} \sum_{\ulinea_{1}^n} p_{\ulineZ_{1}}^n (\ulinez_{1}^{n})
\norm{\theta_{(\ulinez_{1}^{n}, \ulinea_{1}^{n})}^{\otimes n} - 
\xi_{\ulinez_{1}^{n}}^{\otimes n} \otimes \ket{0}\bra{0} }_{1}
\nonumber \\
& \leq & 90 \eta^n + \epsilon. \nonumber 
\end{eqnarray}
For $\eta ^n = \sqrt{\epsilon}$,
$\frac{1}{\mathcal{M}_{1}} \sum_{\ulinem_{1} \in \mathcal{M}_{1}} \norm{\theta_{(\ulinez_{1}^{n}, \ulinea_{1}^{n})(\ulinem_1)}^{\otimes n} - \xi_{\ulinez_{1}^{n}(\ulinem_1)}^{\otimes n} \otimes \ket{0}\bra{0} }_{1}  \leq 90 \sqrt{\epsilon} + \epsilon$.
Now, we proceed to prove 
\begin{eqnarray}
      \frac{2}{\mathcal{M}_{1}} \sum_{\ulinem_{1} \in \mathcal{M}_{1}} 
\tr\left( (I -\gamma^{*}_{(\ulinez_{1}^{n},\ulinea_{1}^{n})(\ulinem_1)} ) 
\theta_{(\ulinez_{1}^{n}, \ulinea_{1}^{n})(\ulinem_1)}^{\otimes n} \right) 
\leq 10808 \sqrt{\epsilon}.
\nonumber
\end{eqnarray}
We have, 
\begin{eqnarray}
&&E \left\{ \frac{2}{\mathcal{M}_{1}} \!\! \sum_{\ulinem_{1} \in \mathcal{M}_{1}} \!\!\!
\tr\left( (I -\gamma^{*}_{(\ulinez_{1}^{n},\ulinea_{1}^{n})(\ulinem_1)} ) 
\theta_{(\ulinez_{1}^{n}, \ulinea_{1}^{n})(\ulinem_1)}^{\otimes n} \right)\ \right\} \nonumber \\
&=& E \left\{\frac{2}{\mathcal{M}_{1}} \sum_{\ulinem_{1} \in \mathcal{M}_{1}}  \!\sum_{\ulinez_{1}^n} \sum_{\ulinea_{1}^n} \tr\left( (I  \gamma^{*}_{(\ulinez_{1}^{n},\ulinea_{1}^{n})} ) 
\theta_{(\ulinez_{1}^{n}, \ulinea_{1}^{n})}^{\otimes n} \right) \mathds{1}\{\ulinez_{1}^{n}(\ulinem_1)=\ulinez_{1}^{n}, \ulinea_{1}^{n}(\ulinem_{1}) = \ulinea_{1}^{n}\} \right\} \nonumber \\
 &=& \frac{2}{\mathcal{M}_{1}} \frac{1}{|\uline{\CalA_{1}|}^n} \sum_{\ulinem_{1} \in \mathcal{M}_{1}}  \sum_{\ulinez_{1}^n} \sum_{\ulinea_{1}^n} p_{\ulineZ_{1}}^n (\ulinez_{1}^{n})
    \tr\left( (I -\gamma^{*}_{(\ulinez_{1}^{n},\ulinea_{1}^{n})} ) 
\theta_{(\ulinez_{1}^{n}, \ulinea_{1}^{n})}^{\otimes n} \right)  \nonumber 
\\
&=& \frac{2}{\mathcal{M}_{1}} \frac{1}{|\uline{\CalA_{1}|}^n} \sum_{\ulinem_{1} \in \mathcal{M}_{1}}  \sum_{\ulinez_{1}^n} \sum_{\ulinea_{1}^n} p_{\ulineZ_{1}}^n (\ulinez_{1}^{n})
\left[\tr\left(\theta_{(\ulinez_{1}^{n}, \ulinea_{1}^{n})}^{\otimes n} \right) -
    \tr\left( \gamma^{*}_{(\ulinez_{1}^{n},\ulinea_{1}^{n})}  
\theta_{(\ulinez_{1}^{n}, \ulinea_{1}^{n})}^{\otimes n} \right)\right]  \nonumber \\
&=& 
\frac{2}{\mathcal{M}_{1}} \frac{1}{|\uline{\CalA_{1}|}^n} \sum_{\ulinem_{1} \in \mathcal{M}_{1}}  \sum_{\ulinez_{1}^n} \sum_{\ulinea_{1}^n} p_{\ulineZ_{1}}^n (\ulinez_{1}^{n})\left[\tr\left(\theta_{(\ulinez_{1}^{n}, \ulinea_{1}^{n})}^{\otimes n} \right) -\tr\left( (I_{\boldsymbol{\CalH_{Y_{1}}^{e}}}-  \beta^{*}_{\ulinez_1^{n},\ulinea_1^{n}}) \Pi_{\boldsymbol{\CalH_{Y_{1G}}} }(I_{\boldsymbol{\CalH_{Y_1}^{e}}}-  \beta^{*}_{\ulinez_1^{n},\ulinea_1^{n}})
\theta_{(\ulinez_{1}^{n}, \ulinea_{1}^{n})}^{\otimes n} \right)\right]  \nonumber \\
\!\!\!\!\!&=&\!\!\!\!\! \frac{2}{\mathcal{M}_{1}} \frac{1}{|\uline{\CalA_{1}|}^n} \!\!\sum_{\ulinem_{1} \in \mathcal{M}_{1}}  \sum_{\ulinez_{1}^n} \sum_{\ulinea_{1}^n} p_{\ulineZ_{1}}^n (\ulinez_{1}^{n}) \left[\tr\left(\theta_{(\ulinez_{1}^{n}, \ulinea_{1}^{n})}^{\otimes n} \right) \!\! -\!\!\tr\left(\Pi_{\boldsymbol{\CalH_{Y_{1G}}} }(I_{\boldsymbol{\CalH_{Y_1}^{e}}} \! \! \! - \beta^{*}_{\ulinez_1^{n},\ulinea_1^{n}})  \theta_{(\ulinez_{1}^{n}, \ulinea_{1}^{n})}^{\otimes n}(I_{\boldsymbol{\CalH_{Y_1}^{e}}} \! \! \!- \beta^{*}_{\ulinez_1^{n},\ulinea_1^{n}})  \Pi_{\boldsymbol{\CalH_{Y_{1G}}} }\right)\right] \nonumber\\
&\leq& \frac{8}{\mathcal{M}_{1}} \frac{1}{|\uline{\CalA_{1}|}^n}\sum_{\ulinem_{1} \in \mathcal{M}_{1}}  \sum_{\ulinez_{1}^n} \sum_{\ulinea_{1}^n} p_{\ulineZ_{1}}^n (\ulinez_{1}^{n}) \bigg[\tr\left(  (I -\Pi_{\boldsymbol{\CalH_{Y_{1G}}} }) \theta_{(\ulinez_{1}^{n}, \ulinea_{1}^{n})}  \right) + \tr\left(  \beta^{*}_{\ulinez_{1}^{n},\ulinea_{1}^{n}}   \theta_{(\ulinez_{1}^{n}, \ulinea_{1}^{n})}  \right) \bigg] \label{eq:noncommutativebound} \\
&=& \frac{8}{\mathcal{M}_{1}} \frac{1}{|\uline{\CalA_{1}|}^n} \sum_{\ulinem_{1} \in \mathcal{M}_{1}}  \sum_{\ulinez_{1}^n} \sum_{\ulinea_{1}^n} p_{\ulineZ_{1}}^n (\ulinez_{1}^{n}) \bigg[\tr\left(  (I - \Pi_{\boldsymbol{\CalH_{Y_{1G}}} }) \theta_{(\ulinez_{1}^{n}, \ulinea_{1}^{n})}  \right) + \tr\left( \beta^{*}_{\ulinez_{1}^{n},\ulinea_{1}^{n}} \theta_{(\ulinez_{1}^{n}, \ulinea_{1}^{n})}  \right)\bigg] \nonumber  \end{eqnarray}
\begin{eqnarray}
&=& \frac{8}{\mathcal{M}_{1}} \frac{1}{|\uline{\CalA_{1}|}^n} \sum_{\ulinem_{1} \in \mathcal{M}_{1}} \sum_{\ulinez_{1}^n} \sum_{\ulinea_{1}^n} p_{\ulineZ_{1}}^n (\ulinez_{1}^{n}) \tr\left(  (I - \Pi_{\boldsymbol{\CalH_{Y_{1G}}} } + \beta^{*}_{\ulinez_{1}^{n},\ulinea_{1}^{n}}) \theta_{(\ulinez_{1}^{n}, \ulinea_{1}^{n})}  \right) \nonumber 
\\ & \leq & \frac{8}{\mathcal{M}_{1}} \frac{1}{|\uline{\CalA_{1}|}^n} \sum_{\ulinem_{1} \in \mathcal{M}_{1}} \sum_{\ulinez_{1}^n} \sum_{\ulinea_{1}^n} p_{\ulineZ_{1}}^n (\ulinez_{1}^{n}) \tr\left(  (I -\Pi_{\boldsymbol{\CalH_{Y_{1G}}} } + \beta^{*}_{\ulinez_{1}^{n},\ulinea_{1}^{n}}) \xi_{\ulinez_{j}^{n}}^{\otimes n} \otimes \ket{0}\bra{0} \right) \nonumber \\
&+& \frac{8}{\mathcal{M}_{1}} \frac{1}{|\uline{\CalA_{1}|}^n} \sum_{\ulinem_{1} \in \mathcal{M}_{1}}  \sum_{\ulinez_{1}^n} \sum_{\ulinea_{1}^n} p_{\ulineZ_{1}}^n (\ulinez_{1}^{n}) \norm{\theta_{\ulinez_{j}^{n}, \ulinea_{j}^{n}}^{\otimes n} - \xi_{\ulinez_{j}^{n}}^{\otimes n} \otimes \ket{0}\bra{0} }_{1}. \label{eq:line5}
\end{eqnarray}
Where inequality \eqref{eq:noncommutativebound} follows from the non-commutative bound \cite[Fact 3]{202103SAD_Sen}. In the following, we analyze each term on the RHS of \eqref{eq:line5} separately. The following steps essentially mimic those employed by Sen in establishing \cite[Bound in (7)]{202103SAD_Sen} - an upper bound on the first term therein. Specifically, by [Corollary 1]\cite{202103SAD_Sen} and Gelfand-Naimark thm.~\cite[Thm.~3.7]{BkHolevo_2019}, we have
\begin{eqnarray}
&&\frac{8}{\mathcal{M}_{1}} \frac{1}{|\uline{\CalA_{1}|}^n} \sum_{\ulinem_{1} \in \mathcal{M}_{1}} \sum_{\ulinez_{1}^n} \sum_{\ulinea_{1}^n} p_{\ulineZ_{1}}^n (\ulinez_{1}^{n}) \tr\left(  (I -\Pi_{\boldsymbol{\CalH_{Y_{1G}}} } + \beta^{*}_{\ulinez_{1}^{n},\ulinea_{1}^{n}}) \xi_{\ulinez_{j}^{n}}^{\otimes n} \otimes \ket{0}\bra{0} \right) \nonumber \\ 
&=& \frac{8}{\mathcal{M}_{1}} \frac{1}{|\uline{\CalA_{1}|}^n} 
\sum_{\ulinem_{1} \in \mathcal{M}_{1}} \sum_{\ulinez_{1}^n} \sum_{\ulinea_{1}^n} 
p_{\ulineZ_{1}}^n (\ulinez_{1}^{n}) \left[ \tr\left(  (I - \Pi_{\boldsymbol{\CalH_{Y_{1G}}} }) \xi_{\ulinez_{1}^{n}}^{\otimes n} \otimes \ket{0}\bra{0}   \right) + \tr\left(  \beta^{*}_{\ulinez_{1}^{n},\ulinea_{1}^{n}} \xi_{\ulinez_{1}^{n}}^{\otimes n} \otimes \ket{0}\bra{0}   \right) \right]\nonumber \\
&=&\frac{8}{\mathcal{M}_{1}} \frac{1}{|\uline{\CalA_{1}|}^n} \sum_{\ulinem_{1} \in \mathcal{M}_{1}} \sum_{\ulinez_{1}^n} \sum_{\ulinea_{1}^n}  p_{\ulineZ_{1}}^n (\ulinez_{1}^{n}) \tr\left(  \beta^{*}_{\ulinez_{1}^{n},\ulinea_{1}^{n}} \xi_{\ulinez_{1}^{n}}^{\otimes n} \otimes \ket{0}\bra{0}   \right) \nonumber \\&\leq& \frac{336}{\mathcal{M}_{1}} \frac{1}{\eta^n} \frac{1}{|\uline{\CalA_{1}|}^n} \sum_{\ulinem_{1} \in \mathcal{M}_{1}}  \sum_{\ulinez_{1}^n} \sum_{\ulinea_{1}^n} 
p_{\ulineZ_{1}}^n (\ulinez_{1}^{n}) \sum_{S \subseteq [4]}  
\tr\left( \olineB_{\ulinez_{1}^{n}}^{S} \xi_{\ulinez_{1}^{n}}^{\otimes n} \otimes \ket{0}\bra{0}   \right) \nonumber \\
&=& \frac{336}{\mathcal{M}_{1}} \frac{1}{\eta^n} \frac{1}{|\uline{\CalA_{1}|}^n} 
\sum_{\ulinem_{1} \in \mathcal{M}_{1}} \sum_{\ulinez_{1}^n} \sum_{\ulinea_{1}^n} \sum_{S \subseteq [4]} p_{\ulineZ_{1}}^n (\ulinez_{1}^{n})
\left[1-\tr\left( \olineG_{\ulinez_{1}^{n}}^{S} \xi_{\ulinez_{1}^{n}}^{\otimes n} \otimes \ket{0}\bra{0}   \right)\right]\nonumber \\
&=&\frac{336}{\mathcal{M}_{1}} \frac{1}{\eta^n} \frac{1}{|\uline{\CalA_{1}|}^n} 
\sum_{\ulinem_{1} \in \mathcal{M}_{1}} \sum_{\ulinez_{1}^{n} } \sum_{\ulinea_{1}^n} \sum_{S \subseteq [4]} p_{\ulineZ_{1}}^n (\ulinez_{1}^{n}) \left[1-\tr\left( G_{\ulinez_{1}^{n}}^{S} \xi_{\ulinez_{1}^{n}}^{\otimes n}   \right)\right]\nonumber \\
&=&\frac{336}{\mathcal{M}_{1}} \frac{1}{\eta^n} \frac{1}{|\uline{\CalA_{1}|}^n} 
\sum_{\ulinem_{1} \in \mathcal{M}_{1}} \sum_{\ulinez_{1}^{n} \in T^{n}_{\delta}(p_{\ulineZ_{1}} )} \sum_{\ulinea_{1}^n} \sum_{S \subseteq [4]} p_{\ulineZ_{1}}^n (\ulinez_{1}^{n}) \left[1-\tr\left( G_{\ulinez_{1}^{n}}^{S} \xi_{\ulinez_{1}^{n}}^{\otimes n}   \right)\right]\nonumber \\
&+& \frac{336}{\mathcal{M}_{1}} \frac{1}{\eta^n} \frac{1}{|\uline{\CalA_{1}|}^n} 
\sum_{\ulinem_{1} \in \mathcal{M}_{1}} \sum_{\ulinez_{1}^{n} \notin T^{n}_{\delta}(p_{\ulineZ_{1}} )} \sum_{\ulinea_{1}^n} \sum_{S \subseteq [4]} p_{\ulineZ_{1}}^n (\ulinez_{1}^{n}) \left[1-\tr\left( G_{\ulinez_{1}^{n}}^{S} \xi_{\ulinez_{1}^{n}}^{\otimes n}   \right)\right]\nonumber \\
&=&  \frac{336}{\mathcal{M}_{1}} \frac{1}{\eta^n} \frac{1}{|\uline{\CalA_{1}|}^n} 
\sum_{\ulinem_{1} \in \mathcal{M}_{1}}  
\sum_{\ulinez_{1}^{n} \in T^{n}_{\delta}(p_{\ulineZ_{1}} )} \sum_{\ulinea_{1}^n} \sum_{S \subseteq [4]} p_{\ulineZ_{1}}^n (\ulinez_{1}^{n})
\left[1-\tr\left( G_{\ulinez_{1}^{n}}^{S} \xi_{\ulinez_{1}^{n}}^{\otimes n}   \right)\right]\nonumber \\
&+& \frac{336}{\mathcal{M}_{1}} \frac{1}{\eta^n} \frac{1}{|\uline{\CalA_{1}|}^n} 
\sum_{\ulinem_{1} \in \mathcal{M}_{1}} \sum_{\ulinez_{1}^{n} \notin T^{n}_{\delta}(p_{\ulineZ_{1}} )} \sum_{\ulinea_{1}^n} \sum_{S \subseteq [4]} p_{\ulineZ_{1}}^n (\ulinez_{1}^{n}) \nonumber \\
& \leq & \frac{5040}{\eta^n}(1- (1-\epsilon)) + \frac{5040}{\eta^n}\epsilon \nonumber \\
&=& \frac{10080}{\eta^n} \epsilon, \nonumber  \end{eqnarray}
\begin{eqnarray}
&&\frac{8}{\mathcal{M}_{1}} \frac{1}{|\uline{\CalA_{1}|}^n} \sum_{\ulinem_{1} \in \mathcal{M}_{1}} \sum_{\ulinez_{1}^n} \sum_{\ulinea_{1}^n}  p_{\ulineZ_{1}}^n (\ulinez_{1}^{n}) \norm{\theta_{\ulinez_{1}^{n}, \ulinea_{1}^{n}}^{\otimes n} - \xi_{\ulinez_{1}^{n}}^{\otimes n} \otimes \ket{0}\bra{0} }_{1} \nonumber \\
&=& \frac{8}{\mathcal{M}_{1}} \frac{1}{|\uline{\CalA_{1}|}^n} \sum_{\ulinem_{1} \in \mathcal{M}_{1}}  \sum_{\ulinez_{1}^n \in T^{n}_{\delta}(p_{\ulineZ_{1}} )} \sum_{\ulinea_{1}^n} p_{\ulineZ_{1}}^n (\ulinez_{1}^{n})
\norm{\theta_{\ulinez_{1}^{n}, \ulinea_{1}^{n}}^{\otimes n} - \xi_{\ulinez_{1}^{n}}^{\otimes n} \otimes \ket{0}\bra{0} }_{1} \nonumber \\ &+& \frac{8}{\mathcal{M}_{1}} \frac{1}{|\uline{\CalA_{1}|}^n} \sum_{\ulinem_{1} \in \mathcal{M}_{1}}  \sum_{\ulinez_{1}^n \notin T^{n}_{\delta}(p_{\ulineZ_{1}} )} \sum_{\ulinea_{1}^n} p_{\ulineZ_{1}}^n (\ulinez_{1}^{n}) \norm{\theta_{\ulinez_{1}^{n}, \ulinea_{1}^{n}}^{\otimes n} - \xi_{\ulinez_{1}^{n}}^{\otimes n} \otimes \ket{0}\bra{0} }_{1} \nonumber \\ &\leq&  720 \eta^n + 8 \epsilon. \nonumber 
\end{eqnarray}
Hence, 
\begin{eqnarray}
      \frac{2}{\mathcal{M}_{1}} \sum_{\ulinem_{1} \in \mathcal{M}_{1}} 
\tr\left( (I -\gamma^{*}_{(\ulinez_{1}^{n},\ulinea_{1}^{n})(\ulinem_1)} ) 
\theta_{(\ulinez_{1}^{n}, \ulinea_{1}^{n})(\ulinem_1)}^{\otimes n} \right) 
\leq \frac{10080 \epsilon}{ \eta^n} + 720 \eta^n + 8 \epsilon. 
\nonumber
\end{eqnarray}
For $ \eta^n = \sqrt{\epsilon}$,$\frac{2}{\mathcal{M}_{1}} \sum_{\ulinem_{1} \in \mathcal{M}_{1}} 
\tr\left( (I -\gamma^{*}_{(\ulinez_{1}^{n},\ulinea_{1}^{n})(\ulinem_1)} ) 
\theta_{(\ulinez_{1}^{n}, \ulinea_{1}^{n})(\ulinem_1)}^{\otimes n} \right) 
\leq 10808 \sqrt{\epsilon}.$

\section{Verification of Standard Inequalities}
\label{AppSec:SimpleVerifications}

\subsection{Verification of \eqref{seperation of the error analysis of the three users}}
\begin{eqnarray}
    && I \otimes I \otimes I - \mu_{\ulinem_1} \otimes \mu_{\ulinem_2} \otimes \mu_{\ulinem_3}  =  I \otimes I \otimes I- \mu_{\ulinem_1} \otimes \mu_{\ulinem_2} \otimes \mu_{\ulinem_3} - \mu_{\ulinem_1} \otimes I \otimes I + \mu_{\ulinem_1} \otimes I \otimes I \nonumber 
    \\
    &=& (I-\mu_{\ulinem_1}) \otimes I \otimes I- \mu_{\ulinem_1} \otimes \mu_{\ulinem_2} \otimes \mu_{\ulinem_3} + \mu_{\ulinem_1} \otimes I \otimes I \nonumber \\
    &=& (I-\mu_{\ulinem_1}) \otimes I \otimes I +  \mu_{\ulinem_1} \otimes (I- \mu_{\ulinem_2})  \otimes (I- \mu_{\ulinem_3}) \nonumber \\
    &\leq & (I-\mu_{\ulinem_1}) \otimes I \otimes I +  I \otimes (I- \mu_{\ulinem_2})  \otimes (I- \mu_{\ulinem_3}) \label{a} \\
    &=& (I-\mu_{\ulinem_1}) \otimes I \otimes I   +  I \otimes (I- \mu_{\ulinem_2})  \otimes I - 
I \otimes  (I- \mu_{\ulinem_2}) \otimes \mu_{\ulinem_3}
    \nonumber \\ 
    & \leq & (I-\mu_{\ulinem_1}) \otimes I \otimes I   + I \otimes (I - \mu_{\ulinem_2}  ) \otimes I - I \otimes I \otimes  \mu_{\ulinem_3} \label{b} \\
    &=& (I-\mu_{\ulinem_1}) \otimes I \otimes I   + I \otimes (I - \mu_{\ulinem_2}  ) \otimes I - I \otimes I \otimes  \mu_{\ulinem_3} 
    + I \otimes I \otimes I - I \otimes I \otimes I  
    \nonumber \\ 
    &=& (I-\mu_{\ulinem_1}) \otimes I \otimes I   + I \otimes (I - \mu_{\ulinem_2}  ) \otimes I +  I \otimes I \otimes (I - \mu_{\ulinem_3}  ) - I \nonumber \\
    & \leq & (I-\mu_{\ulinem_1}) \otimes I \otimes I   + I \otimes (I - \mu_{\ulinem_2}  ) \otimes I +  I \otimes I \otimes (I - \mu_{\ulinem_3}  ), \nonumber
\end{eqnarray}
here, inequalities (\ref{a}, \ref{b}) follow from the fact that $ \mu_{\ulinem_1} \leq I$, and $ (I -\mu_{\ulinem_2}) \leq I  $.
\newpage
\bibliographystyle{IEEEtran}
{\bibliography{CosetCdsFor3CQChnls}}
\end{document}

%% file: LatexPackages.tex
\usepackage[utf8]{inputenc}
\usepackage{stackengine}
\usepackage{graphicx}
\usepackage{amssymb}
\usepackage{amsmath}
\usepackage{mathrsfs}
\usepackage{ulem}
\usepackage{epsf,epsfig}
\usepackage{cite}
\usepackage{color,xcolor}
\usepackage{dsfont}
\usepackage{bbm}
\usepackage{amsthm}
\usepackage{physics}
\usepackage{tcolorbox}
\usepackage{wrapfig}
\usepackage{url}
\usepackage{graphicx}
\usepackage{stmaryrd}
\usepackage{hyperref}
\usepackage{eurosym}
\usepackage{enumitem}
\usepackage{nopageno}
\usepackage{fancyhdr}
\usepackage{mathtools}
\usepackage{tabularx}
\usepackage{array} 
\usepackage{multirow}

%% file: LatexMacroDefinitions.tex
\def\naturals{\mathbb{N}}

\def\complex{\mathbb{C}}

\def\integers{\mathbb{Z}}

\def\fieldq{\mathcal{F}_{q}}
\def\Expectation{\mathbb{E}}

\newcommand{\msout}[1]{\text{\sout{\ensuremath{#1}}}}

\definecolor{darkgreen}{rgb}{0,0.5,0}
\definecolor{darkmagenta}{rgb}{.5,0,.5}
\definecolor{darkblue}{rgb}{0,0,0.5}

\def\CalA{\mathcal{A}}
\def\CalB{\mathcal{B}}
\def\CalC{\mathcal{C}}
\def\CalD{\mathcal{D}}
\def\CalE{\mathcal{E}}
\def\CalF{\mathcal{F}}
\def\CalG{\mathcal{G}}
\def\CalH{\mathcal{H}}

\def\CalL{\mathcal{L}}
\def\CalM{\mathcal{M}}

\def\CalP{\mathcal{P}}
\def\CalQ{\mathcal{Q}}

\def\CalT{\mathcal{T}}
\def\CalU{\mathcal{U}}
\def\CalV{\mathcal{V}}
\def\CalW{\mathcal{W}}
\def\CalX{\mathcal{X}}
\def\CalY{\mathcal{Y}}

\def\ulineG{\underline{G}}

\def\ulineQ{\underline{Q}}
\def\ulineR{\underline{R}}

\def\ulineT{\underline{T}}
\def\ulineU{\underline{U}}
\def\ulineV{\underline{V}}
\def\ulineW{\underline{W}}
\def\ulineX{\underline{X}}
\def\ulineY{\underline{Y}}
\def\ulineZ{\underline{Z}}

\def\ulinea{\underline{a}}

\def\ulinem{\underline{m}}

\def\ulineq{\underline{q}}

\def\ulineu{\underline{u}}
\def\ulinev{\underline{v}}

\def\ulinex{\underline{x}}

\def\ulinez{\underline{z}}

\def\hata{\hat{a}}
\def\hatb{\hat{b}}

\def\hatm{\hat{m}}

\def\hatu{\hat{u}}
\def\hatv{\hat{v}}
\def\hatw{\hat{w}}

\def\haty{\hat{y}}

\def\CalA{\mathcal{A}}
\def\CalB{\mathcal{B}}
\def\CalC{\mathcal{C}}
\def\CalD{\mathcal{D}}
\def\CalE{\mathcal{E}}
\def\CalF{\mathcal{F}}
\def\CalG{\mathcal{G}}
\def\CalH{\mathcal{H}}

\def\CalL{\mathcal{L}}
\def\CalM{\mathcal{M}}

\def\CalP{\mathcal{P}}
\def\CalQ{\mathcal{Q}}

\def\CalT{\mathcal{T}}
\def\CalU{\mathcal{U}}
\def\CalV{\mathcal{V}}
\def\CalW{\mathcal{W}}
\def\CalX{\mathcal{X}}
\def\CalY{\mathcal{Y}}

\def\ulineCalM{\underline{\mathcal{M}}}

\def\ulineCalQ{\underline{\mathcal{Q}}}

\def\ulineCalU{\underline{\mathcal{U}}}
\def\ulineCalV{\underline{\mathcal{V}}}

\def\ulineCalX{\underline{\mathcal{X}}}

\def\parsec{\par\noindent}
\def\med{\medskip\parsec}

\def\define{\mathrel{\ensurestackMath{\stackon[1pt]{=}{\scriptstyle\Delta}}}}

\def\olineB{\overline{B}}

\def\olineG{\overline{G}}

\def\olineS{\overline{S}}
\def\olineT{\overline{T}}

\def\ScrC{\mathscr{C}}

\newtheorem{remark}{Remark}
\newtheorem{theorem}{Theorem}
\newtheorem{corollary}{Corollary}
\newtheorem{definition}{Definition}

\newtheorem{example}{Example}

\newtheorem{fact}{Fact}

\newtheorem{prop}{Prop}

\def\span{\mbox{span}}